\documentclass[journal,twoside,web]{ieeecolor}
\usepackage{generic}
\usepackage{cite}
\usepackage{amsmath,amssymb,amsfonts}
\usepackage{algorithmic}
\usepackage{graphicx}
\usepackage{algorithm,algorithmic}
\usepackage{hyperref}
\usepackage{subcaption}
\hypersetup{hidelinks=true}
\usepackage{textcomp}
\usepackage{array}
\usepackage{tabularx}
\usepackage{booktabs}
\usepackage{multirow,booktabs}
\usepackage{makecell}
\usepackage{comment}

\newtheorem{definition}{Definition}

\newtheorem{theorem}{Theorem}
\newtheorem{lemma}{Lemma}

\newtheorem{remark}{Remark}

\def\BibTeX{{\rm B\kern-.05em{\sc i\kern-.025em b}\kern-.08em
    T\kern-.1667em\lower.7ex\hbox{E}\kern-.125emX}}
\begin{document}
 \title{Privacy-Preserving Resilient Vector Consensus}
\author{Bing Liu, Chengcheng Zhao, \IEEEmembership{Member, IEEE}, Li Chai,  \IEEEmembership{Member, IEEE},\\ Peng Cheng, \IEEEmembership{Member, IEEE}, and Jiming Chen, \IEEEmembership{Fellow, IEEE}
\thanks{The conference version of this paper was presented at the
2024 American Control Conference (ACC), July 10–12, 2024, Toronto, Ontario, Canada.}
\thanks{This work was supported in part by the National Natural Science Foundation of China under Grants 62273305, U2441244, and 62293515, in part by the Zhejiang Provincial Natural Science Foundation of China under Grant
LR25F030002, in part by Zhejiang University Special Project of the State Key Laboratory of Industrial Control Technology ICT2025C05,  and in part by the Young Elite Scientist Sponsorship Program by cast of the China Association for Science and Technology under Grant YESS20210158 (Corresponding author: Chengcheng Zhao and Jiming Chen).}
\thanks{Bing Liu, Chengcheng Zhao, Li Chai, and Peng Cheng are with the State Key Laboratory of Industrial Control Technology, Zhejiang University, Hangzhou, Zhejiang 310027, China (email: \{bing\_liu, chengchengzhao, chaili, lunarheart\}@zju.edu.cn).}
\thanks{Jiming Chen is with the State Key Laboratory of Industrial Control Technology, Zhejiang University, Hangzhou, Zhejiang 310027, China,  and
also with the School of Automation, Hangzhou Dianzi University, Hangzhou, Zhejiang 310018, China (email: cjm@zju.edu.cn).}
}
\maketitle

\begin{abstract}
This paper studies privacy-preserving resilient vector consensus in multi-agent systems against faulty agents, where normal agents can achieve consensus within the convex hull of their initial states while protecting state vectors from being disclosed. 
Specifically, we consider a modification of an existing algorithm known as Approximate Distributed Robust Convergence Using Centerpoints (ADRC), i.e., Privacy-Preserving ADRC (PP-ADRC). Under PP-ADRC, each normal agent introduces multivariate Gaussian noise to its state during each iteration. We first provide sufficient conditions to ensure that all normal agents' states can achieve mean square convergence under PP-ADRC. Then, we analyze convergence accuracy from two perspectives, i.e., the Mahalanobis distance of the final value from its expectation and the Hausdorff distance-based alteration of the convex hull caused by noise when only partial dimensions are added with noise. Then, we employ concentrated geo-privacy to characterize privacy preservation and conduct a thorough comparison with differential privacy. Finally, numerical simulations demonstrate the theoretical results.

\end{abstract}


\section{Introduction}
\label{sec:introduction}
\IEEEPARstart{E}{nabling} multi-agent systems to collaborate effectively in solving complex problems or accomplishing tasks has garnered significant interest in fields such as aviation, robotics, and others \cite{du2017distributed,park2017fault,angrisani2012distributed}. In many application scenarios, such as distributed machine learning, multi-robot systems, and platoon, \emph{vector consensus} is frequently used, aiming to make all agents' state vectors reach agreement under a predefined rule. However, in multi-agent systems, there may be faulty or adversarial agents\footnote{ Note that our primary focus is on attacks at the agent layer, attacks can also occur at the communication layer \cite{he2020almost,he2021secure}. In the subsequent text, both faulty and adversarial agents are collectively referred to as faulty agents.} against the rule to break the consensus of the system or drive the consensus to an unsafe value. To maintain the safe consensus of normal agents, the concept of \emph{resilient vector consensus} has been developed. Resilient vector consensus means that all normal agents can achieve vector consensus where the final value lies within the convex hull of the initial states of normal agents despite the presence of faulty agents. Under resilient vector consensus, many distributed systems can work correctly under complex and changing circumstances. Therefore, the concept of resilient vector consensus is of paramount importance in the field of distributed systems.

\par In most consensus/resilient algorithms, agents directly send their states to their neighbors. However, if the transmitted states are intercepted by a malicious entity, this could lead to the disclosure of agents' sensitive information, resulting in privacy leakage. A typical example is group decision-making \cite{french1981consensus}. For instance, in a voting process, a group of individuals aims to reach a unified opinion using consensus algorithms. Individuals can be reluctant to disclose their true opinions to others due to social pressure, fear of judgment, or other concerns \cite{xia2022asynchronous}. Therefore, it becomes essential to protect the initial state of each participant, ensuring that their opinions remain confidential while still enabling the group to achieve consensus. This highlights the importance of privacy preservation as a critical aspect of resilient vector consensus.

\par Due to the scarcity of communication and computation resources in distributed systems and the high demand of encryption algorithms on these resources, differential privacy has become the primary choice for privacy preservation in consensus algorithms. Typically, in resilient vector consensus, we aim to protect the initial state vectors of the normal agents without the existence of a trusted control center, where local differential privacy (local-DP) \cite{kasiviswanathan2011can} is suitable to be utilized. Different from centralized differential privacy (central-DP), local-DP protects information during transmission from being intercepted, unrelated to differential attacks, and thus does not involve the concept of neighboring inputs. Local-DP also ensures that any pair of inputs generates similar outputs. 
However, the initial state range for each agent can be very large, and requiring pairs of distant initial states to produce similar outputs would reduce the system's utility\footnote{ Utility refers to the extent to which a system can still perform its original functions after the application of privacy-preserving mechanisms. In this paper, we especially refer to the convergence accuracy.}. This motivates us to use the concept of geo-privacy, characterizing the Euclidean distance between inputs as one of the metrics for privacy preservation. Specifically, we propose to use the concept of concentrated geo-privacy (CGP) \cite{liang2023concentrated}, which is a generalization of differential privacy, to describe the privacy preservation of resilient vector consensus. Apart from a more precise description of privacy preservation, it also offers advantages such as the Gaussian mechanism and advanced composition.

Much attention has been paid to resilient consensus algorithms, which can be roughly classified into two categories. One is to detect and preclude faulty agents\cite{silvestre2017stochastic,ramos2023discrete}, while the other one is to try to find a safe value despite the presence of faulty ones\cite{kieckhafer1994reaching,vaidya2012iterative,leblanc2013resilient,zhang2012robustness}. In this paper, we focus on the second one as it has more analytical results. Related works on resilient scalar/vector consensus and privacy-preserving resilient scalar consensus are provided below.

\textbf{Resilient scalar consensus} is designed for normal agents to converge to a common state that lies between the maximum and minimum values of normal agents' initial states. There are two frequently used ways:  One way is to use the median to filter extreme values and reach resilience. In \cite{zhang2012simple}, Zhang \emph{et al.} proposed a median consensus algorithm, where each normal agent updates by using its own value and the median of the states from its neighborhood. As for the other way, most of the algorithms rely on the strategy of simply disregarding suspicious values. For example, in a series of algorithms named Mean-Subsequence-Reduced (MSR) algorithm \cite{kieckhafer1994reaching,vaidya2012iterative}, the main idea is to discard a fixed number of largest and smallest values in each normal agent's neighborhood. Conversely, in the modification of the Weighted MSR (W-MSR) algorithm \cite{leblanc2013resilient,zhang2012robustness}, each normal agent only discards values that are greater or less than its own to maintain its state.

\textbf{Resilient vector consensus} is ensured by that each normal agent must seek a point that is always located in the convex hull of its normal neighboring state vectors at each iteration. The existing work can be classified into three categories. The first category utilizes Tverberg partitions to compute Tverberg points. Specifically, it includes the Byzantine vector consensus algorithm \cite{vaidya2013byzantine} and Approximate Distributed Robust Convergence algorithm \cite{park2017fault}. The second category calculates the intersection of multiple convex hulls. This involves Algorithm $1$ proposed in \cite{yan2022resilient} and the multidimensional approximate agreement algorithm described in \cite{mendes2013multidimensional}. The algorithm of the last category utilizes the concept of ``centerpoint", exemplified by ``Approximate Distributed Robust Convergence Using Centerpoint" (ADRC) algorithm \cite{abbas2022resilient}. However, these algorithms are rather time-consuming, and most of them require approximation algorithms, which consequently reduce fault tolerance. Notably, among these algorithms, the ADRC \cite{abbas2022resilient} stands out with its better tolerance for the mature approximation algorithm. In each iteration, every normal agent computes the centerpoint (an extension of the median in higher dimensions) of its neighborhood and updates its state as the weighted average of its current state and the centerpoint. Resilient vector consensus can be achieved as long as the communication topology and the number of faulty agents within the neighborhood meet certain requirements.

  \par \textbf{Privacy-preserving resilient scalar consensus} can primarily be achieved through three methods: cryptography-based methods \cite{hou2023privacy}, observability-based methods \cite{zhang2022privacy,tu2023privacy}, 
  and noise-adding-based methods \cite{hou2023privacy,fiore2019resilient}. Cryptography-based methods offer higher convergence accuracy; however, the processes of encryption, information exchange, and decryption are time-consuming, significantly reducing efficiency. Observability-based methods primarily reduce the observability of the system, making the system's states unobservable to external entities. The most widely applied approach in this category is state decomposition. The remaining substates interact and update internally with the exchanged substate, thereby preserving privacy. However, this approach lacks complete privacy preservation and does not provide a quantifiable metric for the strength of privacy preservation. On the other hand, noise-adding-based methods are typically more flexible, as they allow for the adjustment of noise intensity and distribution based on specific application scenarios and privacy requirements. In \cite{fiore2019resilient}, Fiore \emph{et al.} proposed a Differentially Private MSR (DP-MSR) algorithm, which adds decaying and zero-mean Laplace noise to the scalar states of the normal agents. They also analyzed resilient scalar consensus in terms of differential privacy and asymptotics in expectation.

\par Note that privacy-preserving resilient vector consensus is still an open issue. One simple way is to conduct a privacy-preserving resilient scalar consensus algorithm $d$ times, where $d$ is the dimension of states. However, this may cause the final value to fall outside the convex hull of the initial states of the normal agents \cite{abbas2022resilient}. Additionally, analyzing the performance of privacy-preserving resilient vector consensus by noise addition involves several key challenges: (i) The noise added makes the convergence nondeterministic, random, and asymptotic, rendering the existing definition of resilient vector consensus inapplicable. (ii) Due to the infinite sequence of noise and resilience, deriving an analytical solution for the final value is infeasible, making it difficult to quantify convergence accuracy. Additionally, high-dimensional convex hulls, unlike their one-dimensional counterparts defined by just two parameters, are characterized by numerous vertices and boundaries, which change randomly and sharply after noise is added. This randomness complicates both resilience and convergence accuracy analysis, particularly in arbitrarily high-dimensional settings. (iii) At each iteration, each agent transmits and updates its state information and the specific distribution of local state at each iteration is nondeterministic. We need to calculate the privacy leakage of the initial states caused by information transmission over infinitely many iterations, which is also a challenging task.

To solve the above challenges, we consider a modification of ADRC, named Privacy-Preserving ADRC (PP-ADRC), and provide rigorous theoretical analysis for convergence and privacy performance.
Compared to our conference version \cite{liu2024trade}, we use Gaussian noise instead of Laplace noise and conduct a more rigorous performance analysis. Meanwhile, we add the final value distribution analysis by characterizing the deviation of the final value from the expectation. Moreover, we use CGP instead of DP to characterize the privacy preservation and compare it with DP comprehensively. We also add more simulations under a $3$-dimensional case to illustrate our theoretical results. The contributions of this paper are given as follows.
\begin{itemize}
\item We consider a modification of ADRC, named PP-ADRC, where each agent adds Gaussian noise to the local state vector for local interaction. We provide sufficient conditions to ensure resilient vector consensus in expectation.
\item 
We analyze the final value from two perspectives. First, we utilize the Mahalanobis distance to analyze the residuals of the final value from the expected one given a specified probability upper bound. Secondly, we employ the Hausdorff distance to assess the change between the convex hulls with and without noise with a specified probability upper bound, where only partial fixed dimensions are added with noise for each iteration.
\item 
We employ $\rho$-CGP to characterize privacy preservation without detailed distributions of the outputs. Through a detailed comparison with differential privacy, we demonstrate that $\rho$-CGP can provide advantages for resilient vector consensus. 
\end{itemize}
\par The organization of this work is as follows. We introduce preliminaries, problem formulation, and algorithm design in Sections \ref{sec:pre} and \ref{sec:problem}, respectively. Then, we analyze the convergence condition and final value in Section \ref{sec:convergence}. The privacy analysis is presented in Section \ref{sec:privacy}. Simulation results are presented in Section \ref{sec:sim}. Finally, conclusions and avenues for future research are outlined in Section \ref{sec:conclusion}.

\section{PRELIMINARIES}\label{sec:pre}
\par In this section, we introduce some preliminaries on the basic notations, reachable graph sequence, and privacy notions.
\par \textbf{Basic Notations:} Throughout this paper, we denote by $\mathbb{R}^{d}$ the $d$-dimensional Euclidean space, by $\mathbb{R}^{m\times d}$ the space of all $m\times d$-dimensional matrices, and by $(\mathbb{R}^{m\times d})^\mathbb{N}$ the space of matrix-valued sequences in $\mathbb{R}^{m\times d}$, where $\mathbb{N}$ denotes the set of natural numbers. The natural logarithm is denoted by $\log$. Let $\mathbf{A}$ be a matrix. Then, we denote by $[\mathbf{A}]_i,[\mathbf{A}]^j,[\mathbf{A}]_{ij},\mathbf{A}^{\top}$ the $i$ th row, the $j$ th column, the $(i,j)$ th element, and the transpose of the matrix $\mathbf{A}$. We then consider two matrices $\mathbf{A}$ and $\mathbf{B}$, and if $[\mathbf{A}]_{ij}\leq [\mathbf{B}]_{ij}$ holds for any pair of $i,j$, we say $\mathbf{A}\leq \mathbf{B}$. For a vector $x\in \mathbb{R}^d$, we denote by $x_{\operatorname{max}}$ and $x_{\operatorname{min}}$ maximum element and minimum element of it, respectively. Then, we consider a matrix consisting of $n$ row vectors $\mathbf{x}=[x_1,x_2,\ldots,x_n]^{\top}$, where $x_i$ denotes the $i$th vector and $x_{i,k}$ stands for the $k$th element of the vector $x_i$. We denote by $\mathbf{1}$ a column vector of ones, and by $I_d$ the $d$-dimensional identity matrix. 
A matrix is row-stochastic if all its elements are nonnegative and each row sums up to $1$. For a given point set $C\subseteq\mathbb{R}^d$, we denote by $\text{conv}(C)$ the convex hull of the point set $C$, by $|C|$ the cardinality of $C$. For two different vectors $x,x'\in \mathbb{R}^d$, we represent by $\mathrm{dist}(x,x')=||x-x'||_2$ the Euclidean distance between $x$ and $x'$. For two $n$-tuples of vectors $\mathbf{x}=[x_1,x_2,\ldots,x_n]^\top$ and $\mathbf{x}'=[x_1',x_2',\ldots,x_n']^\top$, $\mathrm{dist}(\mathbf{x},\mathbf{x}')$ stands for the maximum distance among all the vectors, which is $\max\limits_i \operatorname{dist}(x_i,x_i')$. In probability theory, given a random variable $z \in \mathbb{R}$, we denote by $\mathbb{E}(z)$ and $\operatorname{var}(z)$ the expectation and the variance of $z$, respectively. For a random vector $z=[z_1,\ldots,z_d]^\top \in \mathbb{R}^d$, the expectation of $z$ is $\mathbb{E}(z)=[\mathbb{E}(z_1),\ldots,\mathbb{E}(z_d)]^\top$. For a topological space $\mathcal{T}$, we denote by $\mathcal{B}(\mathcal{T})$ the set of Borel subsets of $\mathcal{T}$, and $\mathrm{P}$ is the corresponding probability measure. 
For a zero-mean $d$-dimensional Gaussian distribution, the probability density function (pdf) is given by $G(x;\Sigma)=\frac{1}{\sqrt{(2\pi)^d|\Sigma|}}\mathrm{exp}\left(-\frac{1}{2}x^{\top}\Sigma^{-1}x\right)$. For two functions $f(x)$ and $g(x)$, $f(x)=\Omega(g(x))$ means that $g(x)$ is the asymptotic lower bound of $f(x)$.
\par \textbf{Reachable Graph Sequence:} We denote by $\mathcal{G}(t)=(\mathcal{V},\mathcal{E}(t)),\text{ }t=0,1,2,\ldots$ a time-varying directed graph, where $\mathcal{V}$ is the set of vertices, and $\mathcal{E}(t)\subseteq \mathcal{V} \times \mathcal{V}$ is the set of edges. Then, we introduce two definitions of reachability with a relationship to a graph sequence below \cite{park2017fault}.
\begin{definition}
\label{def:jR}
(Jointly Reachable Graph Sequence): Given $j\in\mathbb{N}$ and a finite time sequence $T_j, T_j+1,\ldots, T_{j+1}-1$, the corresponding finite sequence of graphs $\mathcal{G}(T_j), \mathcal{G}(T_j+1), \ldots, \mathcal{G}(T_{j+1}-1) $ is said to be a jointly reachable graph sequence if in the union of graphs $\bigcup_{t=T_j}^{T_{j+1}-1}\mathcal{G}(t) = \left(\mathcal{V},\bigcup_{t=T_j}^{T_{j+1}-1}\mathcal{E}(t)\right)$, there is a vertex $v\in\mathcal{V}$ such that $\forall v'\ne v$, we can find a path from $v'$ to $v$ in the union of graphs. 
\end{definition}
\begin{definition}
\label{def:ijR}
(Repeatedly Reachable Graph Sequence): An infinite sequence of graphs $\mathcal{G}(0), \mathcal{G}(1), \ldots$ is said to be a repeatedly reachable graph sequence if there is an infinite time sequence, $T_{\mathrm{inf}}:0=T_1 < T_2 < \ldots$, such that for any $j\in\mathbb{N}$, the subsequence $\mathcal{G}(T_j), \mathcal{G}(T_j + 1),\ldots, \mathcal{G}(T_{j+1} - 1)$ is a jointly reachable graph sequence.
\end{definition}
\par \textbf{Privacy Notions:} We first give two formal definitions of the standard local-DP and then introduce CGP.
\begin{definition}
($\varepsilon$-DP): Given spaces $U$, $V$, and $\varepsilon\in \mathbb{R} \geq 0$, a randomized function $M:U\rightarrow V$ is $\varepsilon$-differentially private, if for any pair of inputs $x$, $x'\in U$ and any $S\subseteq V$, we have $\mathrm{P}\{M(x)\in S\}\leq e^{\varepsilon}\mathrm{P}\{M(x')\in S\}.$
\end{definition}
\begin{definition}
($(\varepsilon,\delta)$-DP): Given spaces $U$, $V$, and $\varepsilon$, $\delta \in \mathbb{R} \geq 0$, a randomized function $M$ is $(\varepsilon,\delta)$-differentially private, if for any pair of inputs $x$, $x'\in U$ and any $S\subseteq V$,  $\mathrm{P}\{M(x)\in S\}\leq e^{\varepsilon}\mathrm{P}\{M(x')\in S\}+\delta$ holds. 
\end{definition}
\par DP introduces randomness into the output, making it nondeterministic and following a probability distribution. $\varepsilon$-DP ensures that the probability ratio of obtaining the same result for any two inputs is at most $e^\varepsilon$. In contrast, $(\varepsilon, \delta)$-DP relaxes this, allowing the ratio to exceed $e^\varepsilon$ with a probability no greater than $\delta$, i.e., $\varepsilon$-DP is achieved with a probability of at least $1-\delta$. Note that in both definitions, the inputs $x$ and $x'$ can be in the real domain and consist of one or more scalars or vectors. The parameter $\varepsilon$ in both definitions characterizes the privacy preservation strength, with smaller values of $\varepsilon$ providing stronger privacy preservation.
\par Additionally, the Gaussian mechanism is only applicable to $(\varepsilon, \delta)$-DP, due to the light tail of the Gaussian distribution. The pdf ratio for two Gaussian distributions with different means tends to infinity, making it impossible to bound the ratio by $e^\varepsilon$. In contrast, the Laplace mechanism, with a bounded probability density ratio for different means, is suitable for $\varepsilon$-DP. For further details, we refer to \cite{dwork2014algorithmic}.
\par It is worth noting that Definitions $3$-$4$ impose rather strict requirements for resilient vector consensus. As the range of possible initial states of an agent is typically large, ensuring that two far apart initial states produce similar outputs would cause significant system utility reduction. To solve this issue, geo-privacy has been widely studied, which uses Euclidean distance as one metric for privacy preservation. Supporting the Gaussian mechanism and better composition, CGP is selected here, which is established on the R\'enyi divergence given below \cite{renyi1961measures}.
\begin{definition}
(R\'enyi Divergence): Given two distributions $F$ and $G$ on domain $\operatorname{dom}(x)$ with pdf $f(x)$ and $g(x)$, respectively, R\'enyi divergence of order $\alpha>1$ is defined as
\begin{equation}\label{eq:renyi}
\begin{aligned}
D_\alpha(F\|G)&=\frac1{\alpha-1}\log\left(\int_{\operatorname{dom}(x)}f(x)^\alpha g(x)^{1-\alpha}\mathrm{d}x\right)\\
&=\frac{1}{\alpha-1}\log\left(\underset{x\thicksim f(x)}{\mathbb{E}}\bigg[\left(\frac{f(x)}{g(x)}\right)^{\alpha-1}\bigg]\right).
\end{aligned}
\end{equation}
\end{definition}

Note that $D_\alpha(F\|G)$ is determined by the expectation of the $\alpha - 1$ power of the ratio of two distributions at each point, quantifying the difference between two distributions. Meanwhile, larger values of $\alpha$ emphasize parts where the ratio is greater, while smaller values of $\alpha$ focus more on the overall average differences. The roles of logarithm and $\frac{1}{\alpha - 1}$ are intended to prevent numerical overflow in the results. R\'enyi divergence is monotonically non-decreasing with $\alpha$. 

\begin{definition}
($\rho$-CGP): Given spaces $U$, $V$, and $\rho\in \mathbb{R}\geq 0$, a randomized function $M:U\rightarrow V$ is said to satisfy $\rho$-CGP, if for any inputs $x,x'\in U$ and all $\alpha>1$, it holds that
\begin{equation}
D_{\alpha}(M(x)\|M(x'))\leq\alpha\rho \cdot \mathrm{dist}(x,x')^2. 
\end{equation}
\end{definition}
\par In CGP, $\rho$ provides the upper bound on $D_{\alpha}(M(x)\|M(x'))/\alpha/ \mathrm{dist}(x,x')^2$, characterizing the difference between the output distributions of $M$ for any pair of inputs, thereby quantifying the strength of privacy preservation. A smaller $\rho$ indicates that $M$ is less sensitive to the inputs, thereby providing stronger privacy preservation. Besides, the privacy should degrade gracefully as $\mathrm{dist}(x, x')$ increases. The parameter $\alpha$ plays a crucial role as it allows for the adjustment of privacy preservation strictness. A larger $\alpha$ emphasizes the worst-case scenario by highlighting the regions where the two distributions differ the most, while a smaller $\alpha$ focuses on the average difference. This flexibility enables the adaptation of privacy preservation levels according to varying scenario requirements.
  \par Under local-DP, the privacy parameter $\varepsilon$ provides a uniform upper bound based on the maximum distance between inputs, often resulting in stronger privacy for most input pairs than $\varepsilon$ indicates. In contrast, CGP integrates input distance into its privacy analysis, offering tighter, distance-dependent privacy guarantees. This allows CGP to achieve equivalent privacy preservation with less noise, thereby improving system utility.

\section{Problem Formulation and Algorithm Design}\label{sec:problem}
\par In this section, we first present the network model of the multi-agent system with faulty agents and define privacy attackers. We then introduce the PP-ADRC algorithm and provide a brief explanation of centerpoint and its computation methods. Finally, we state several interesting problems.
\subsection{Network Model}
\par We consider a multi-agent system with $n$ agents modeled by a time-varying directed graph $\mathcal{G}(t)=(\mathcal{V},\mathcal{E}(t)),\text{ }t=0,1,2,\ldots$, where $\mathcal{V}=\{1,2,\ldots,n\}$ is the set of agents, and $\mathcal{E}(t)\subseteq\mathcal{V}\times \mathcal{V}$ is the set of edges. The set of in-neighbors of an agent $i$ is denoted by $\mathcal{N}_{i}^{\text{in}}(t)=\{j\in\mathcal{V}|(j,i)\in\mathcal{E}(t)\}$, while the set of out-neighbors is $\mathcal{N}_i^{\text{out}}(t)=\{j\in\mathcal{V}|(i,j)\in\mathcal{E}(t)\}$.
\par The system has two types of agents, i.e., normal agents and faulty agents. Each normal agent updates its state through local interaction based on predefined rules. Conversely, faulty agents can behave arbitrarily and unpredictably. They are classified into two categories: Byzantine agents and malicious agents. Byzantine agents send different arbitrary states to different neighbors, while malicious agents can only send the same state to all neighbors \cite{leblanc2013resilient}. Malicious agents are actually a specific type of Byzantine agents. Therefore, for a more general case, we assume that the faulty agents in this paper are Byzantine agents. We denote by $\mathcal{F}\in\mathcal{V}$ the set of faulty agents. Let $\overline{\mathcal{V}}=\mathcal{V}-\mathcal{F}$, and $\overline{n}=n-|\mathcal{F}|$, where $\overline{\mathcal{V}}$ is the set and $\overline{n}$ is the number of all normal agents. Without loss of generality, we suppose $\overline{n}$ ones in front are the normal agents,  meaning that $\overline{\mathcal{V}}=\{1,2,\ldots,\overline{n}\}$. Besides, we denote by a $\overline{n}\times d$ matrix $\overline{\mathbf{x}}(t)=[x_1(t),x_2(t),\ldots,x_{\overline{n}}(t)]^\top$ the states of all normal agents, where $x_i(t)$ denotes the state of a normal agent $i\in \overline{\mathcal{V}}$, and by $\overline{\mathbf{x}}_0$ the initial states of the normal agents. We define the topology of normal agents by a time-varying directed graph $\overline{\mathcal{G}}(t)=(\overline{\mathcal{V}},\overline{\mathcal{E}}(t))$, where $\overline{\mathcal{E}}(t)\subseteq{\overline{\mathcal{V}}}\times \overline{\mathcal{V}}$. For each normal agent $i\in\overline{\mathcal{V}}$, we denote by $\overline{\mathcal{N}}_i^{\text{in}}(t)$ and $n_{f_i}(t)$ its set of in-neighbors in $\overline{\mathcal{G}}(t)$ and the number of its faulty in-neighbors in $\mathcal{G}(t)$, respectively. 
\subsection{Privacy Attackers}
Regarding privacy preservation, we consider two types of privacy attackers, as referenced in \cite{wang2019privacy}.
  \begin{enumerate}
    \item Privacy attackers are within the network. This type of attacker includes two categories: the honest-but-curious attacker and the malicious attacker. The former follows the pre-defined algorithm just like normal agents but attempts to steal other agents' privacy through the information received from network interactions. The latter refers to the previously mentioned Byzantine agents, who seek to disrupt consensus while simultaneously stealing the privacy of other agents.
    \item Privacy attackers are outside the network, i.e., eavesdroppers. This type of attacker is aware of the network topology and can wiretap communication channels to access the information exchanged between agents.
  \end{enumerate}

\subsection{Algorithm Design}
\par We consider a modification of the existing ADRC algorithm, i.e.,  PP-ADRC, as detailed in Algorithm \ref{alg:PP-ADRC}. To preserve privacy, zero-mean decaying Gaussian noise is added to the state of each normal agent during the communication phase. In ADRC, the key to achieving resilience lies in calculating the centerpoint $s_i(t)$ during the calculation phase. 

\begin{algorithm}[htbp]
  \small
  \caption{PP-ADRC}\label{alg:PP-ADRC}
  \begin{algorithmic}
  \STATE 
  \STATE {\bf Input:} $\mathcal{G}(t)$, $\overline{\mathbf{x}}(t)$, $\Sigma(t)=\lambda^2\upsilon^{2t} I_d$, $\gamma_i(t)$, and number of iterations $T$.
  \STATE {\bf Output:} final value $\xi$ 
  \STATE {\bf Initialization:} initialize states of normal agents $\overline{\mathbf{x}}_0$, $t=0$
   \FOR {$\{t=0,1,\ldots,T-1\}$}
  \FOR{each normal agent $i\in \overline{\mathcal{V}}$} 
    \STATE {\bf Transmission phase:}
    \STATE 	Add noise to its state, i.e.,
   \begin{equation}\label{eq:add_noise}
   y_i(t) = x_i(t)+\eta_i(t),    
   \end{equation}
  \STATE where $\eta_i(t)\in \mathbb{R}^d$ is a zero-mean decaying $d$-dimensional Gaussian noise with covariance matrix $\Sigma(t)$.
  \STATE Transmit $y_i(t)$ to its out-neighbors.
    \STATE {\bf Calculation phase:}
    \STATE Calculate centerpoint $s_i(t)$.
    \STATE {\bf Update phase:}	
    \STATE Update its state following 
    \begin{equation}\label{eq:update}
  x_i(t+1)=\gamma_i(t)s_i(t)+(1-\gamma_i(t))x_i(t),
  \end{equation} 
  where $0<\gamma_l\leq \gamma_i(t)\leq \gamma_m<1$ and $\gamma_l>1-\upsilon$.
    \ENDFOR
    \ENDFOR
  \end{algorithmic}
  \vspace{-2pt}
  \end{algorithm}

 \par Here, we choose Gaussian noise for it is canonical under CGP. We can obtain explicit analytical results compared to Laplace noise. In addition, in resilient vector consensus, there are more advantages that the Gaussian mechanism can bring to the system, which are also detailed as follows. 
  \begin{itemize}
  \item Gaussian noise is easier to implement in non-one-dimensional cases compared to Laplace noise in practical systems, which reduces computational complexity.
  \item The magnitude of Gaussian noise grows at a slower rate concerning the dimension $d$. Specifically, the scale of Gaussian noise is $\tilde{O}(\sqrt{d})$, whereas Laplace noise scales as $\tilde{O}(d)$. In resilient vector consensus, where the dimension $d$ can be arbitrarily high, Gaussian noise offers greater flexibility in parameter selection.
  \item Gaussian noise is an additive noise, meaning that the sum of two Gaussian noises is still a Gaussian noise. In many application scenarios of resilient vector consensus, such as multi-robot rendezvous, geometric data measurement is involved. However, many geometric measurements inherently contain Gaussian noise. The linearity of Gaussian noise makes it easier for the algorithm to achieve the desired privacy preservation parameters.
  \end{itemize}
\par We introduce decaying Gaussian noise to protect the initial states of normal agents. Privacy leakage occurs not only during the initial transmission but also at each subsequent iteration, though it diminishes over time. Thus, adding noise with the same intensity at every iteration is unnecessary, and injecting noise only to the initial states is also inadequate. A decaying variance better aligns with our privacy preservation requirements while ensuring system convergence. Using a fixed noise variance would prevent normal agents from reaching consensus, causing their state variances to diverge to infinity as time tends to infinity. Therefore, injecting decaying noise is essential for both privacy preservation and consensus. This is also a common approach in privacy-preserving scalar consensus \cite{fiore2019resilient,nozari2017differentially,mo2016privacy}.

\subsection{Centerpoint Introduction}
We provide a brief introduction to the concept and computation of the centerpoint.
\par For a set of $n$ points in general position in $\mathbb{R}^d$, there always exists a centerpoint $p$, which is guaranteed to lie within the relative interior of the convex hull of any subset containing more than $\frac{nd}{d+1}$ points from the given point set \cite{abbas2022resilient}. Therefore, if the fraction of faulty agents in a normal agent's neighborhood is lower than $\frac{1}{d+1}$, the centerpoint always lies in the relative interior of the convex hull formed by the normal agents' state vectors. However, the fraction $\frac{1}{d+1}$ is just a theoretical value (This fraction corresponds to the depth of the centerpoint.). When the dimensionality exceeds three, calculating a centerpoint becomes a coNP-Complete problem, where we need approximate algorithms. The specific computation method is as follows.
 \par The computation of the centerpoint varies depending on the dimension. In the two-dimensional case, the prune and search algorithm proposed in \cite{jadhav1993computing} can find the ideal centerpoint with time complexity of $O(n)$. For three or more dimensions, the exact algorithm for computing the ideal centerpoint was proposed in \cite{chan2004optimal}, with time complexity of $O(n^{d-1})$. Chan demonstrated that centerpoint computation can be formulated as an Implicit Linear Programming problem. However, the algorithm's complexity grows significantly with higher dimensions, restricting its practicality to low-dimensional cases. In high-dimensional settings, approximation algorithms are necessary, categorized into depth approximation and location approximation. The depth approximation algorithm, proposed in \cite{miller2009approximate}, has time complexity of $O(n^{\log d})$. This algorithm projects the point set onto a line and uses an $r$-partition, yielding a centerpoint with depth $\Omega\left(\frac{1}{d^2}\right)$.  For location approximation, an algorithm proposed in \cite{cherapanamjeri2024computing} provides a polynomial-time solution. It assigns each point a Gaussian perturbation with scale $\varepsilon$ and then uses the Radial Isotropic Transformation to obtain an $\varepsilon$-approximate centerpoint with depth $\Omega\left(\frac{1}{d}\right)$. Specifically, the algorithm guarantees that the point returned is at most $\varepsilon$ away from any half-space defined by the centerpoint with depth at least $\Omega\left(\frac{1}{d}\right)$ in any direction.

\subsection{Problem of Interests}
\par Different from the resilient vector consensus without adding noise, the randomness of noise makes the result under PP-ADRC become nondeterministic. Thus, we need to characterize metrics for resilient vector consensus and privacy preservation after noise addition, which are provided below.
\begin{definition}\label{def:consensus}
(Resilient Vector Consensus): The resilient vector consensus is achieved if state vectors of all normal agents satisfy the following two conditions:
\begin{itemize}
\item $\emph{Safety}$: For any normal agent $i\in \overline{\mathcal{V}}$ and any $t>0$, $\mathbb{E}[x_i(t)]$ must be in the convex hull formed by the initial states of the normal agents, i.e.  $\mathbb{E}[x_i(t)] \in \operatorname{conv}\left(x_1(0),x_2(0),\ldots,x_{\overline{n}}(0)\right)$.
\item $\emph{Agreement}$: For any pair of normal agents $i,j\in \overline{\mathcal{V}}$, it always holds that $\lim\limits_{t\rightarrow\infty}\mathbb{E}[||x_{i}(t)-x_{j}(t)||_2^2]=0$.
\end{itemize}
\end{definition}
\par To characterize the privacy preservation of the state vector sequences, we first denote by three matrices $\overline{\mathbf{x}}(t),\ \overline{\boldsymbol{\eta}}(t),\ \text{and}\ \overline{\mathbf{y}}(t)\in \mathbb{R}^{\overline{n}\times d}$ the states, noises, and noisy states of all the agents at time $t$, respectively. Then, we define the sequences of matrices $\mathbf{X}=\{\overline{\mathbf{x}}(t)\}_{t=0}^\infty, \mathbf{N}=\{\overline{\boldsymbol{\eta}}(t)\}_{t=0}^\infty,\ \text{and}\ \mathbf{Y}=\{\overline{\mathbf{y}}(t)\}_{t=0}^\infty$ from $t=0$ to $\infty$, and they are in the sample space $\Omega=(\mathbb{R}^{\overline{n}\times d})^\mathbb{N}$. Since the algorithm is specific, given an initial state $\overline{\mathbf{x}}_0$, the sequences $\mathbf{X}$ and $\mathbf{Y}$ are uniquely determined by the noise sequence $\mathbf{N}$ and we define the corresponding function as $\mathbf{Y}_{\overline{\mathbf{x}}_0}(\mathbf{N})=\mathbf{Y}$.
\begin{definition}
\label{def:CGP2}
($\rho$-CGP in Resilient Vector Consensus): Given any pair of initial states $\overline{\mathbf{x}}_0$, $\overline{\mathbf{x}}_0'$, and $\rho>0$, the system is said to satisfy $\rho$-CGP if and only if for all $\alpha>1$
\begin{equation}\label{eq:cgp_rvs}
D_\alpha(\mathbf{y}_{\overline{\mathbf{x}}_0}(\mathbf{N})\|\mathbf{y}_{\overline{\mathbf{x}}_0'}(\mathbf{N})) \leq \alpha \rho \cdot \mathrm{dist}(\overline{\mathbf{x}}_0,\overline{\mathbf{x}}_0')^2,     
\end{equation}
where $\mathbf{y}_{\overline{\mathbf{x}}_0}(\mathbf{N})$ and $\mathbf{y}_{\overline{\mathbf{x}}_0'}(\mathbf{N})$ are the pdfs of $\mathbf{Y}_{\overline{\mathbf{x}}_0}(\mathbf{N})$ and $\mathbf{Y}_{\overline{\mathbf{x}}_0'}(\mathbf{N})$, respectively.
\end{definition}
\begin{remark}
Privacy leakage can occur at both the initial iteration and each subsequent iteration. To address this, we focus on the privacy performance of the infinite matrix sequence of noisy states $\mathbf{Y}$. CGP ensures that for any two initial states $\overline{\mathbf{x}}_0$ and $\overline{\mathbf{x}}_0'$, the difference in the distributions of the infinite sequences measured by Rényi divergence is bounded by $\alpha \rho \cdot \mathrm{dist}(\overline{\mathbf{x}}_0, \overline{\mathbf{x}}_0')^2$. Essentially, an observer cannot distinguish between $\overline{\mathbf{x}}_0$ and $\overline{\mathbf{x}}_0'$ based on the output sequences.
\end{remark}

\par After noise addition, convergence analysis and privacy quantization become challenging. For convergence analysis, we can only describe the situation in terms of expectation. But for the final value, although the expectation is ensured to lie within the convex hull of the normal agents' initial states, we cannot obtain an analytical solution. Specifically, the exact result of each convergence and its deviation from the expectation are unknown. In addition, characterizing the change of the convex hull after adding noise is very challenging because the vertices and boundaries are constantly changing. Without analytical expressions, it is hard to quantitatively describe the changes in the convex hull. For privacy preservation, we add an infinite amount of noise, and the output is an infinite sequence of state matrices without an explicit distribution expression. Analyzing the Rényi divergence of such infinite matrix sequences is inherently complex and necessitates a specialized approach. Consequently, we are interested in the following questions of the system performing PP-ADRC:
\begin{enumerate}
\item Under what conditions can we demonstrate that resilient vector consensus can be achieved?
\item How can we determine the change in the convex hull of the initial states of normal agents after adding noise? Furthermore, although the expectation of the final value remains within the original convex hull, how can we quantify the deviation from it for each single final value?
\item How can we derive the upper bound of $\rho$-CGP without knowing the specific distribution of two infinite matrix sequences? What are the advantages compared to  DP?
\end{enumerate}

\section{Convergence Analysis}
\label{sec:convergence}
\par In this section, we provide sufficient conditions to ensure resilient vector consensus. Then, we analyze the convergence accuracy from two perspectives, i.e., the deviation of the final value from the expectation and the change of the convex hull.
\subsection{Resilient Vector Consensus}
\par We first formulate the state evolution of all normal agents as a linear time-varying (LTV) system using the properties of the centerpoint. Then, by sequentially multiplying the iterative formula of the LTV system from the initial iteration, we can leverage repeated reachability and the characteristics of stochastic matrices to establish the resilient vector consensus of PP-ADRC \cite{park2017fault}.
\begin{lemma}\label{ltv}
Under PP-ADRC, if 
\begin{equation}
\label{eq:neighbor}
n_{f_i}(t)<N_{f_i}(t)=\begin{cases}
\frac{|\mathcal{N}_i^{\text{in}}(t)|}{d+1}\ & \text{if } d=2,3\\
\Omega(\frac{|\mathcal{N}_i^{\text{in}}(t)|}{d^2})\ &\text{if } d>3
\end{cases}    
\end{equation}
holds for any $t\geq 0$ and any $i \in \overline{\mathcal{V}}$, then we have 
\begin{equation}\label{eq:ltv}
\overline{\mathbf{x}}(t+1)=\mathbf{M}(t)\overline{\mathbf{x}}(t)+\mathbf{H}(t)\mathbf{v}(t),\ t=0,1,2,\ldots,
\end{equation}
where $\mathbf{M}(t)\in \mathbb{R}^{\overline{n}\times \overline{n}}$ is a row-stochastic matrix with $[\mathbf{M}(t)]_{ii}=1-\gamma_i(t)$, $\mathbf{H}(t) \in \mathbb{R}^{\overline{n} \times \overline{n}}$ has zero diagonal entries and differs from $\mathbf{M}(t)$ only in the diagonal elements, and $\mathbf{v}(t)=\left[\eta_1(t),\eta_2(t),\ldots,\eta_{\overline{n}}(t)\right]^\top$.
\end{lemma}
\begin{proof}
If the condition in \eqref{eq:neighbor} holds for any normal agent $i\in \overline{\mathcal{V}}$ and $t\geq 0$, then by the Proposition V.1 in \cite{park2017fault}, we can write $s_i(t)$ as a convex combination of its normal in-neighbors' noisy states
\begin{equation}
s_i(t)=\sum_{j\in \overline{\mathcal{N}}_i^\text{in}(t)}a_{ij}(t)y_j(t)=\sum_{j\in \overline{\mathcal{N}}_i^\text{in}(t)}a_{ij}(t)[x_j(t)+\eta_j(t)],
\nonumber
\end{equation}
where $a_{ij}(t)>0$. We can obtain the update equation 
\begingroup
\small
\begin{equation}
\begin{aligned}
x_i(t+1)&=\gamma_i(t)\sum_{j\in \overline{\mathcal{N}}_i^\text{in}(t)}a_{ij}(t)[x_j(t)+\eta_j(t)]+(1-\gamma_i(t))x_i(t).
\end{aligned}
\nonumber
\end{equation}
\endgroup
Finally, the system can be represented as 
\begin{equation}
\overline{\mathbf{x}}(t+1)=\mathbf{M}(t)\overline{\mathbf{x}}(t)+\mathbf{H}(t)\mathbf{v}(t),\ t=0,1,2,\ldots,
\nonumber
\end{equation}
where $\overline{\mathbf{x}}(t)=[x_1(t),\ldots,x_{\overline{n}}(t)]^\top$ and $\mathbf{v}(t)=\left[\eta_1(t),\eta_2(t),\ldots,\eta_{\overline{n}}(t)\right]^\top$. Matrices $\mathbf{M}(t)$ and $\mathbf{H}(t)$ are 
\begin{equation}\label{eq:MT}
[\mathbf{M}(t)]_{ij}=\begin{cases}
1-\gamma_i(t) & \text{if } i=j\\
\gamma_i(t)a_{ij}(t) & \text{if } i \neq j \text{ and } j\in \overline{\mathcal{N}}_i^\text{in}(t)\\
0  & \text{otherwise,}
\end{cases}
\end{equation}
and
\begin{equation}
[\mathbf{H}(t)]_{ij}=\begin{cases}
\gamma_i(t)a_{ij}(t) & \text{if } i \neq j \text{ and } j\in \overline{\mathcal{N}}_i^\text{in}(t)\\
0  & \text{otherwise.}
\end{cases}
\end{equation} 
\end{proof}

\par We denote by $\Psi_{t,t_s}$ the backward product of the sequence $\{\mathbf{M}(t),t\geq0\}$, i.e., 
$\Psi_{t,t_s}:=\mathbf{M}(t-1)\dots \mathbf{M}(t_s),\mathrm{for}\text{ }t>t_s\geq0,\text{ }\Psi_{t,t}:=I$.   
Then, we obtain 
\begin{equation}\label{eq1}
\overline{\mathbf{x}}(t+1)=\Psi_{t+1,0}\overline{\mathbf{x}}(0)+\sum_{q=0}^{t}\Psi_{t+1,q+1}\mathbf{H}(q)\mathbf{v}(q).
\end{equation}

\begin{theorem}
\label{th:consensus}
Under PP-ADRC, if the condition in \eqref{eq:neighbor} holds and $\overline{\mathcal{G}}(0), \overline{\mathcal{G}}(1), \overline{\mathcal{G}}(2), \ldots$ is repeatedly reachable, then for any $i \in \overline{\mathcal{V}}$, there exists a random vector $\xi = [\xi_{1}, \ldots, \xi_{d}]^\top$ such that $\lim\limits_{t \rightarrow \infty} \mathbb{E}\left[||x_i(t) - \xi||_2^2\right] = 0$, and $\mathbb{E}[x_i(t)] \in \operatorname{conv} \left(x_1(0), x_2(0), \ldots, x_{\overline{n}}(0)\right)$, where $t = 0, 1, 2, \ldots$.

\end{theorem}

\begin{proof}
   Here, we mainly use strong ergodicity\footnote{ Weak ergodicity implies the rows of a backward product of row-stochastic matrices become identical as time approaches infinity, while strong ergodicity additionally guarantees the existence of this limit \cite{huang2012stochastic}.} to show the convergence. Considering that weak and strong ergodicity are equivalent for the backward product of row-stochastic matrices \cite{seneta2006non}, we only need to prove weak ergodicity.
  \par Following \cite{nedic2016convergence}, we first establish a positive lower bound for the nonzero entries of $\mathbf{M}(t)$ for $t \geq 0$. The lower bound for diagonal entries is $1 - \gamma_m$. For non-diagonal entries, each consists of two positive terms: $\gamma_i(t)$, which is lower-bounded by $\gamma_l$, and $a_{ij}(t)$, which lacks a direct lower bound. However, under the condition in \eqref{eq:neighbor}, the centerpoint can be expressed as a convex combination of normal agents' noisy states, allowing us to determine specific values of $a_{ij}(t)$. Defining $a_l$ as the minimum value among all $a_{ij}(t)$, we obtain a lower bound of $\gamma_l a_l$ for the non-diagonal entries. Thus, the overall lower bound for the nonzero entries of $\mathbf{M}(t)$ is $\min(1-\gamma_m, \gamma_l a_l)$.
   Then, by using repeatedly reachable graph sequence, we can establish the weak ergodicity of $\{\mathbf{M}(t),t\geq 0\}$ according to the Theorem V.2 in \cite{park2017fault}, which leads to the strong ergodicity. Using equation \eqref{eq1}, along with the independence and zero-mean properties of the noise and the strong ergodicity, for any $i \in \overline{\mathcal{V}} $ and $k = 1, 2, \dots, d$, we obtain that 
\begin{equation}\label{eq:weak}
\begin{aligned}
&\lim\limits_{t\rightarrow\infty}\mathbb{E}\left[||x_{i,k}(t+1)-x_{j,k}(t+1)||_2^2\right]=0,\ \text{and}\\ 
&\lim\limits_{t\rightarrow\infty} \mathbb{E}\left[||x_{i,k}(t+1)-x_{i,k}(t)||_2^2\right]=0.
\end{aligned}
\end{equation}
Hence, there must exists a random variable $\xi_k$ satisfying $\lim\limits_{t\rightarrow \infty}\mathbb{E}\left[||x_{i,k}(t)-\xi_k||_2^2\right]=0$, i.e.,
\begin{equation}
\lim\limits_{t\rightarrow \infty}\mathbb{E}\left[||x_i(t)-{\xi}||_2^2\right]=0,\ 
\forall i\in \overline{\mathcal{V}}.
\end{equation}
Additionally, by the independence and zero-mean properties of noise, we obtain
$\mathbb{E}\left[\overline{\mathbf{x}}(t+1)\right]=\mathbb{E}\left[\Psi_{t+1,0}\overline{\mathbf{x}}(0)\right]$. As $\Psi_{t+1,0}$ is row-stochastic, one derives
$\mathbb{E}\left[x_i(t+1)\right] \in \operatorname{conv}\left(x_1(0),x_2(0),\ldots,x_{\overline{n}}(0)\right)$. 
\end{proof}
\begin{remark}
It should be pointed out that the condition in \eqref{eq:neighbor} ensures ``safety", while repeated reachability ensures the ergodicity of backward product $\Psi_{t,t_s}$, thereby guaranteeing ``agreement". After adding noise, we can only guarantee the ``mean square convergence" under expectation.  As the goal is to characterize the effect of noise term by expectation, independence, and zero mean are the crucial factors, rather than the specific distribution of the noise. This means that resilient vector consensus can be achieved for other noise distributions with independence and zero-mean properties.
\end{remark}
\begin{remark}
Theorem \ref{th:consensus} provides only sufficient conditions, making it essential to explore the necessary conditions as well. The key challenge lies in finding the necessary conditions of establishing ergodicity of $\{\mathbf{M}(t),t\geq0\}$. This requires that the product of $\mathbf{M}(t)$ over a finite iterations has an ergodic coefficient upper bounded by a positive value less than $1$. Given that the topology $\overline{\mathcal{G}}(t)$ is directed and time-varying, and considering that the elements of $\mathbf{M}(t)$ lack analytical expressions due to the centerpoint-based algorithm, existing ergodicity results are inapplicable. This presents a highly challenging problem. Additionally, while \cite{nedic2016convergence} provides a uniform lower bound for the nonzero entries of the system matrix to analyze convergence rate, we can only prove the existence of such a bound without deriving an explicit expression. The main challenge lies in analyzing the centerpoint, which is known to lie within the relative interior of the convex hull of the normal agents' states \cite{park2017fault,abbas2022resilient}, but its exact location and analytical expression remain unknown. This issue also appears in other resilient vector consensus studies \cite{park2017fault,abbas2022resilient,yan2022resilient}, where the absence of an explicit lower bound precludes convergence rate analysis. In our future work, we aim to derive the minimal topological conditions necessary to ensure ergodicity and derive a precise lower bound for the nonzero entries of $\mathbf{M}(t)$.
\end{remark}

\subsection{Distribution of Final Value}
\par From Theorem \ref{th:consensus}, although the expectation of the final value is in the convex hull of the normal agents' initial states, we cannot guarantee that each single convergence result always lies in the convex hull. This uncertainty arises because the specific distribution of the final value remains unknown. Thus, we investigate the accuracy of each final value in terms of its residuals, i.e., the deviation from the expected value. Specifically, we propose to use the Mahalanobis distance \cite{de2000mahalanobis} and the multivariate Chebyshev's inequality \cite{chen2007new} to analyze the residuals. The details are provided as follows.
\begin{definition}
\label{def:mahalanobis}
Given a random vector $z$ and its non-singular covariance matrix $\Sigma$, the Mahalanobis distance $D_M(z)$ from $z$ to its expectation $\mathbb{E}(z)$ is defined as
\begin{equation}
D_M(z)=\sqrt{(z-\mathbb{E}(z))^\top\Sigma^{-1}(z-\mathbb{E}(z))}.
\end{equation}
\end{definition}
\par In resilient vector consensus, our primary concern is the convergence accuracy of the random vector $\xi$. The commonly used Euclidean distance is not suitable in this scenario due to varying variances of each component in $\xi$ and existing correlations among them. These factors introduce different measurement scales across dimensions. Therefore, we opt for the Mahalanobis distance to quantify the deviation of the final value from its expected value. This metric offers a comprehensive approach to distance measurement by incorporating variances and correlations between components. By mitigating the influence of disparate scales on distance calculation, the Mahalanobis distance provides a more precise evaluation.

\par Here, we provide the details of the multivariate Chebyshev's inequality \cite{chen2007new}.
\begin{lemma}
\label{lm:multi_chebyshev}
For a random vector $z\in \mathbb{R}^d$ with non-singular covariance matrix $\Sigma$, we have
\begin{equation}
\mathrm{P}\left\{[z-\mathbb{E}(z)]^\top\Sigma^{-1}[z-\mathbb{E}(z)]\geq\chi\right\}\leq\frac d\chi,\quad\forall\chi>0.
\end{equation}
\end{lemma}

If the covariance matrix of $\xi$ is non-singular, then by Definition \ref{def:mahalanobis} and Lemma \ref{lm:multi_chebyshev}, we obtain the following result.

\begin{theorem}
\label{th:madis}
Under PP-ADRC, if the condition in \eqref{eq:neighbor} holds, $\overline{\mathcal{G}}(0),\overline{\mathcal{G}}(1),\overline{\mathcal{G}}(2),\ldots$ is repeatedly reachable, and the covariance matrix of $\xi$ is non-singular, then 
\begin{equation}
\label{eq:conv_acc}
\begin{aligned}
\mathrm{P}\left\{[D_M(\xi)]^2\leq \chi\right\}\geq1-\frac d\chi,\quad\forall\chi>0. 
\end{aligned}
\nonumber
\end{equation}
\end{theorem}
\begin{remark}
The points sharing the same Mahalanobis distance from the expected value define a hyperspheroid centering at the expectation of the final value. Therefore, Theorem \ref{th:madis} encapsulates the final value within a hyperspheroid centered at a point (the expectation) inside the convex hull. In fact, if we know the position of the expectation, we can determine the probability that the final value lies within the convex hull more specifically.  Additionally, note that the principal axis parameters of the hyperspheroid are determined by the eigenvalues of the covariance matrix and $\chi$ together. Thus, we can show that the stronger the noises are, the bigger the hyperspheroid is, implying a lower convergence accuracy. Particularly, the volume of the hyperspheroid is $V=V_d\left[\sqrt{|\Sigma|\chi}\right]^d$, where $V_d$ characterizes the volume of a d-dimensional unit hypersphere \cite{duda2006pattern}. Meanwhile, we can write the covariance matrix $\Sigma$ as
\begin{equation}
\left.[\Sigma]_{kk'}=\left\{\begin{array}{ll}\operatorname{var}(\xi_k)&\text{if }k=k'\\\tau_{k,k'}\sqrt{\operatorname{var}(\xi_k)\operatorname{var}(\xi_{k'})}&\text{if }k\neq k',\end{array}\right.\right.
\end{equation}
\end{remark}
where $\tau_{k,k'}$ is the correlation coefficient. Consequently, we can derive the determinant of $\Sigma$, i.e., $|\Sigma|=h(\tau_{1,2},\ldots,\tau_{d-1,d})\operatorname{var}(\xi_1)\times\cdots\times\operatorname{var}(\xi_d)$, where $h(\cdot)$ represents a polynomial. It means that the volume is directly proportional to the variances. 
\par For the case where the covariance matrix of $\xi$ is singular, we can obtain the following result \cite{holton2003value}.
\begin{lemma}\label{lm:singular}
For a random vector $z\in \mathbb{R}^d$, if its covariance matrix $\Sigma$ is singular, then there exists at least one non-zero row vector $\zeta$ such that $\operatorname{var}(\zeta z)=0$, implying $\zeta z$ is a constant.
\end{lemma}
\begin{remark}
Lemma \ref{lm:singular} implies that certain components of the random variable can be expressed as linear polynomials of the remaining components. Identifying these extraneous components allows us to obtain a random vector $Z' \in \mathbb{R}^{d'}$ with a non-singular covariance matrix $\Sigma'$. We consider the case where the covariance matrix of $\xi$ is singular and without loss of generality, the latter $d-d'$ components of $\xi$ are supposed to be extraneous. It means that Theorem \ref{th:madis} still holds for the remaining $d'$ components, where the result simply degenerates into a $d'$-dimensional ellipsoid (or a line if $d' = 1$).
\end{remark}

\subsection{Convex Hull Change}
\par  In contrast to scalar consensus, the high-dimensional convex hull is defined by numerous vertices, dimensions, and coefficients. Moreover, the stochastic nature of noise can significantly alter the convex hull of the normal agents, making it challenging to characterize. Furthermore, following a single iteration, the new convex hull may exhibit no overlap with the original configuration.
\par To solve this issue, we consider a special case that only partial dimensions need to be protected. Consequently, the convex hull of dimensions without adding noise retains its original form. Our focus shifts to analyzing the transformation of the convex hull in the remaining dimensions of the state. It is worth noting that analyzing high-dimensional convex hulls can be challenging, while the $1$-dimensional final value $\xi_k$ can be assessed by examining its expectation and variance. Hence, we start the convex hull change analysis from the perspective of one dimension.
\par 
We first evaluate the expectation and variance of $\xi_k$. Then, we determine the one-dimensional convergence accuracy using Chebyshev's inequality. Finally, we quantize the change of the convex hull caused by Gaussian noise and its associated probability. The properties of $\xi_k$ are given below.
\begin{lemma}\label{lm:variance}
Under PP-ADRC, if the condition in \eqref{eq:neighbor} holds, and $\overline{\mathcal{G}}(0),\overline{\mathcal{G}}(1),\overline{\mathcal{G}}(2),\ldots$ is repeatedly reachable, then 
\begin{equation}
\mathbb{E}\left[\xi_k\right]=\psi [\overline{\mathbf{x}}_0]^k,
\operatorname{var}(\xi_k)\leq \frac{\lambda^2}{1-\upsilon^2}, 1\leq k \leq d,
\nonumber
\end{equation}
where $\psi$ is a $1\times \overline{n}$ non-negative row vector satisfying $\psi \mathbf{1}=1$.
\end{lemma}

\begin{proof}
 Following \cite{fiore2019resilient}, by using the independence and zero mean of the noise, along with the strong ergodicity (indicating that the limit of $\Psi_{t+1,0}$ exists), we can obtain that
  \begin{equation}
    \mathbb{E}\left[\xi_k\right]=\lim\limits_{t\rightarrow \infty}\mathbb{E}\left[[\Psi_{t+1,0}]_i[\overline{\mathbf{x}}_0]^k\right]=\psi [\overline{\mathbf{x}}_0]^k.
    \nonumber
  \end{equation}
Then for the variance, using $\mathbf{H}(t)\leq \mathbf{M}(t)$, one derives
\begin{equation}
\begin{aligned}
\operatorname{var}(\xi_k)&\leq\frac{1}{\overline{n}}\lim\limits_{t\rightarrow \infty}\sum_{q=0}^{t} \sum_{j=1}^{\overline{n}}\left[\Psi_{t+1,q}^T\Psi_{t+1,q}\right]_{jj}\text{var}([\mathbf{v}]_{jk}(q))\\
&=\frac{\lambda^2}{1-\upsilon^2}.
\end{aligned}
\nonumber
\end{equation} 
\end{proof}

\par Then, we establish the accuracy in expectation by \cite{nozari2017differentially}.
\begin{lemma}\label{lm:accuracy}
Under PP-ADRC, if the condition in \eqref{eq:neighbor} holds, and $\overline{\mathcal{G}}(0),\overline{\mathcal{G}}(1),\overline{\mathcal{G}}(2),\ldots$ is repeatedly reachable, then
\begin{equation}\label{eq:hau_pr}
\begin{aligned}
\mathrm{P}\{|\xi_k-\psi [\overline{\mathbf{x}}_0]^k|\leq (l_k+r_k)\}&\geq1-\frac{\operatorname{var}(\xi_k)}{(l_k+r_k)^2}\\
&\geq 1-\frac{\lambda^2/(1-\upsilon^2)}{(l_k+r_k)^2},    
\end{aligned}  
\end{equation}
\end{lemma}
where $l_k=\operatorname{min}\{\psi [\overline{\mathbf{x}}_0]^k-[\overline{\mathbf{x}}_{0}]^k_{\operatorname{min}},[\overline{\mathbf{x}}_{0}]^k_{\operatorname{max}}-\psi[\overline{\mathbf{x}}_0]^k\}$, $r_k>\max\{0,\lambda \sqrt{\frac{1}{1-\upsilon^2}}-l_k\}$, and $k=1,2,\ldots,d$.
\par 
\begin{proof}
By Chebyshev's inequality, we determine the one-dimensional convergence accuracy,  
\begin{equation}
\begin{aligned}
\mathrm{P}\{|\xi_k-\psi [\overline{\mathbf{x}}_0]^k|\leq (l_k+r_k)\}&\geq1-\frac{\operatorname{var}({\xi_k})}{(l_k+r_k)^2}\\
&\geq1-\frac{\lambda^2/(1-\upsilon^2)}{(l_k+r_k)^2}, 
\end{aligned}
\nonumber
\end{equation}
where $r_k>\max\{0,\lambda \sqrt{\frac{1}{1-\upsilon^2}}-l_k\}$ provides a tight lower bound. 
\end{proof}
\begin{remark}
Since $\psi$ is a row vector with elements summing to one, $\psi [\overline{\mathbf{x}}_0]^k$ must lie within the range bounded by the maximum and minimum values of $[\overline{\mathbf{x}}_0]^k$. Consequently, Lemma \ref{lm:accuracy} implies that $\xi_k$ lies within the interval $[\overline{\mathbf{x}}_{0}]^k_{\operatorname{max}}+r_k$ to $[\overline{\mathbf{x}}_{0}]^k_{\operatorname{min}}-r_k$ with a probability of at least $1-\frac{\lambda^2/(1-\upsilon^2)}{(l_k+r_k)^2}$.     
\end{remark}

Then, we consider the case that a fraction $(1-\beta)$, with $0<\beta<1$, of the state's dimensions do not need to be protected, i.e., only part of the dimensions are added with noises. The convex hull formed by these dimensions remains unchanged. For the remaining $\beta d$ dimensions, we apply Lemma \ref{lm:accuracy} $\beta d$ times. Then, we provide three convex hull definitions.
\begin{definition}\label{def:convexhull}
\par $\ $
\begin{enumerate}
\item $\mathcal{A}$: Convex hull $\mathcal{A}$ is defined as the convex hull of the initial states of the normal agents.
\item $\mathcal{B}$: 
Projecting convex hull $\mathcal{A}$ onto $(1-\beta)d$ dimensions without adding noise, we then translate it sequentially along each remaining dimension from its minimum value to its maximum value among all normal agents, forming convex hull $\mathcal{B}$ at the end.
\item $\mathcal{C}$: 
Convex hull $\mathcal{C}$ is formed by extending the translation range of $\mathcal{B}$, moving from the minimum value minus $r_k$ to the maximum value plus $r_k$ for the $k$th component of the normal agents' states.
\end{enumerate}
\end{definition}
\par Then, we can precisely determine the shape of the new convex hull $\mathcal{C}$. We take a $3$-dimensional case as an example, where the coordinates of $6$ points are $(1,0,0)$, $(0,2,0)$, $(-1/4,-1/4,0)$, $(0,0,2)$, $(0,1,0)$, and $(1/4,1/2,0)$. The noise is only added to the z-axis component of the states, ensuring $0\leq\psi [\overline{\mathbf{x}}_0]^3\leq 2$. The three distinct convex hulls are depicted in Fig. \ref{fig:conv_compare}. The convex hull $\mathcal{A}$ represents the original convex hull formed by the $6$ points. The convex hulls $\mathcal{B}$ and $\mathcal{C}$ are obtained by translating the part of $\mathcal{A}$ in the $xOy$ plane along the $z$-axis. The convex hull $\mathcal{B}$ translates from $0$ to $2$, while $\mathcal{C}$ translates from $-r_3$ to $2+r_3$. Therefore, utilizing Lemma \ref{lm:accuracy}, we can calculate that the probability of the final value belonging to the convex hull $\mathcal{C}$ is at least $1-\frac{\lambda^2/(1-\upsilon^2)}{(l_k+r_k)^2}$.
\begin{figure}[htbp]
\centering
\begin{subfigure}{0.32\linewidth}
\vspace{-10pt}
		\centering
		\includegraphics[width=1.0\linewidth]{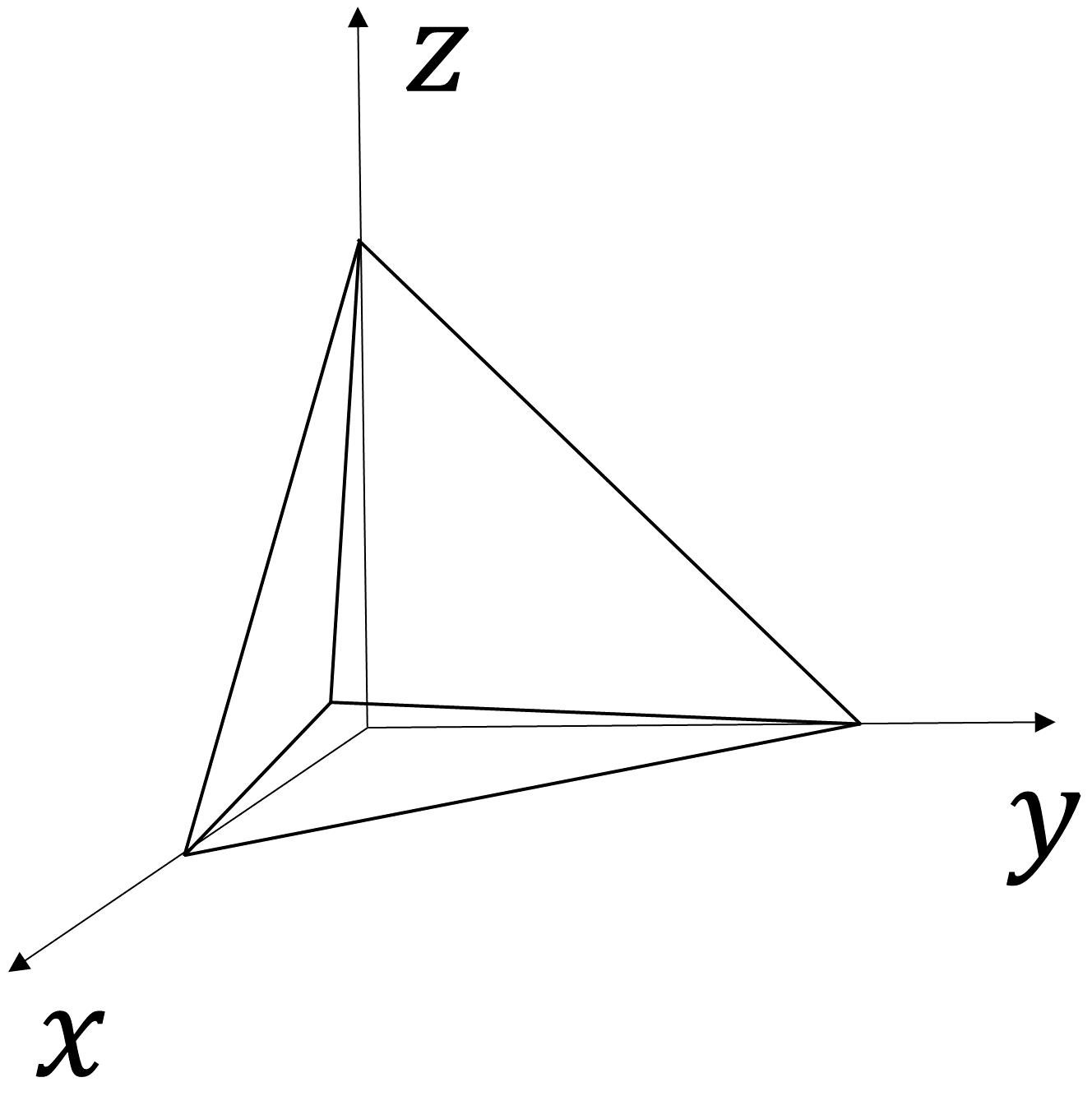}
		\caption{$\mathcal{A}$}
		\label{fig:A}
	\end{subfigure}
\begin{subfigure}{0.32\linewidth}
		\centering
		\includegraphics[width=1.0\linewidth]{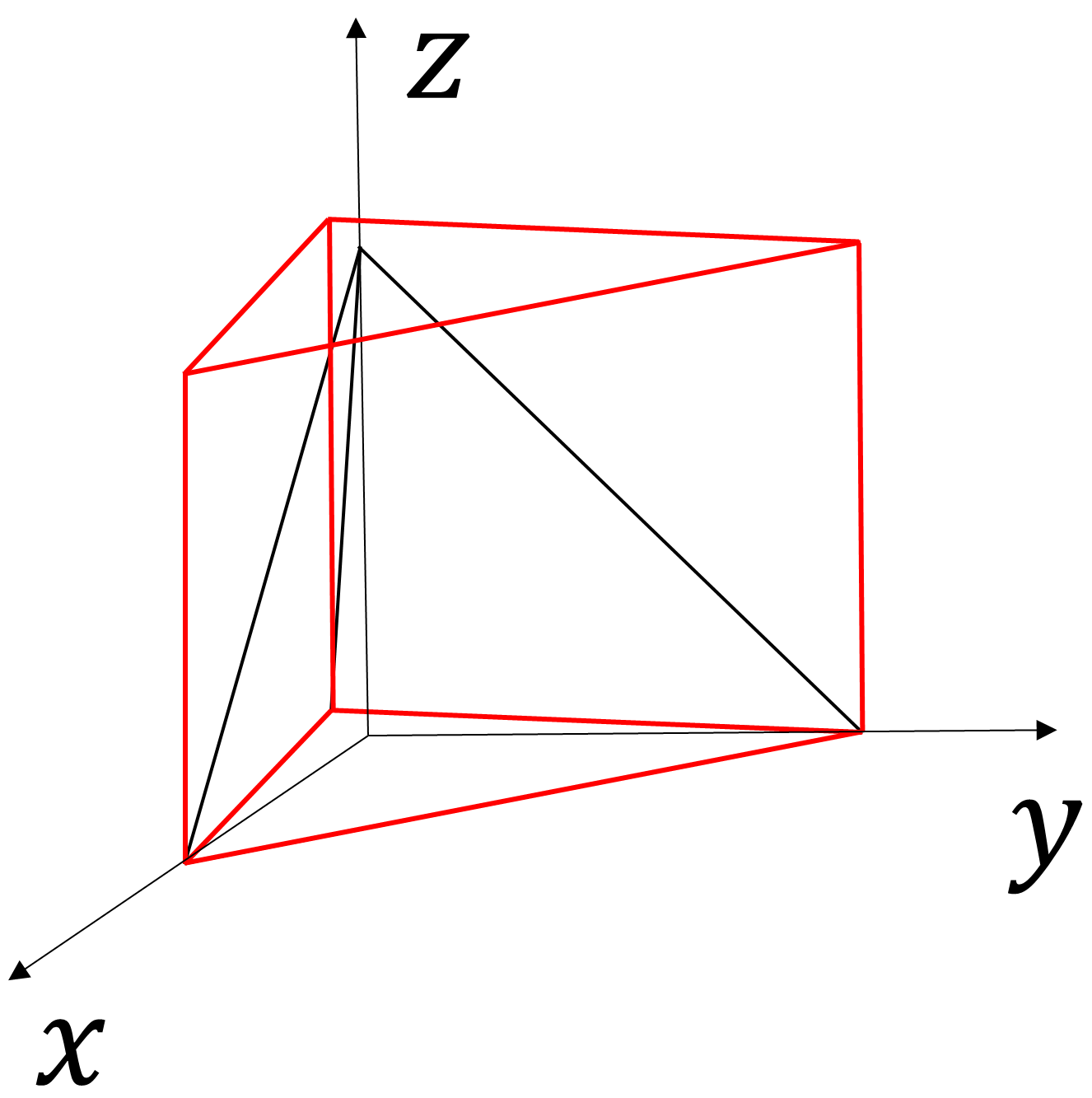}
		\caption{$\mathcal{A}$ and $\mathcal{B}$}
		\label{fig:AB}
	\end{subfigure}
\begin{subfigure}{0.32\linewidth}
		\centering
            \includegraphics[width=1.0\linewidth]{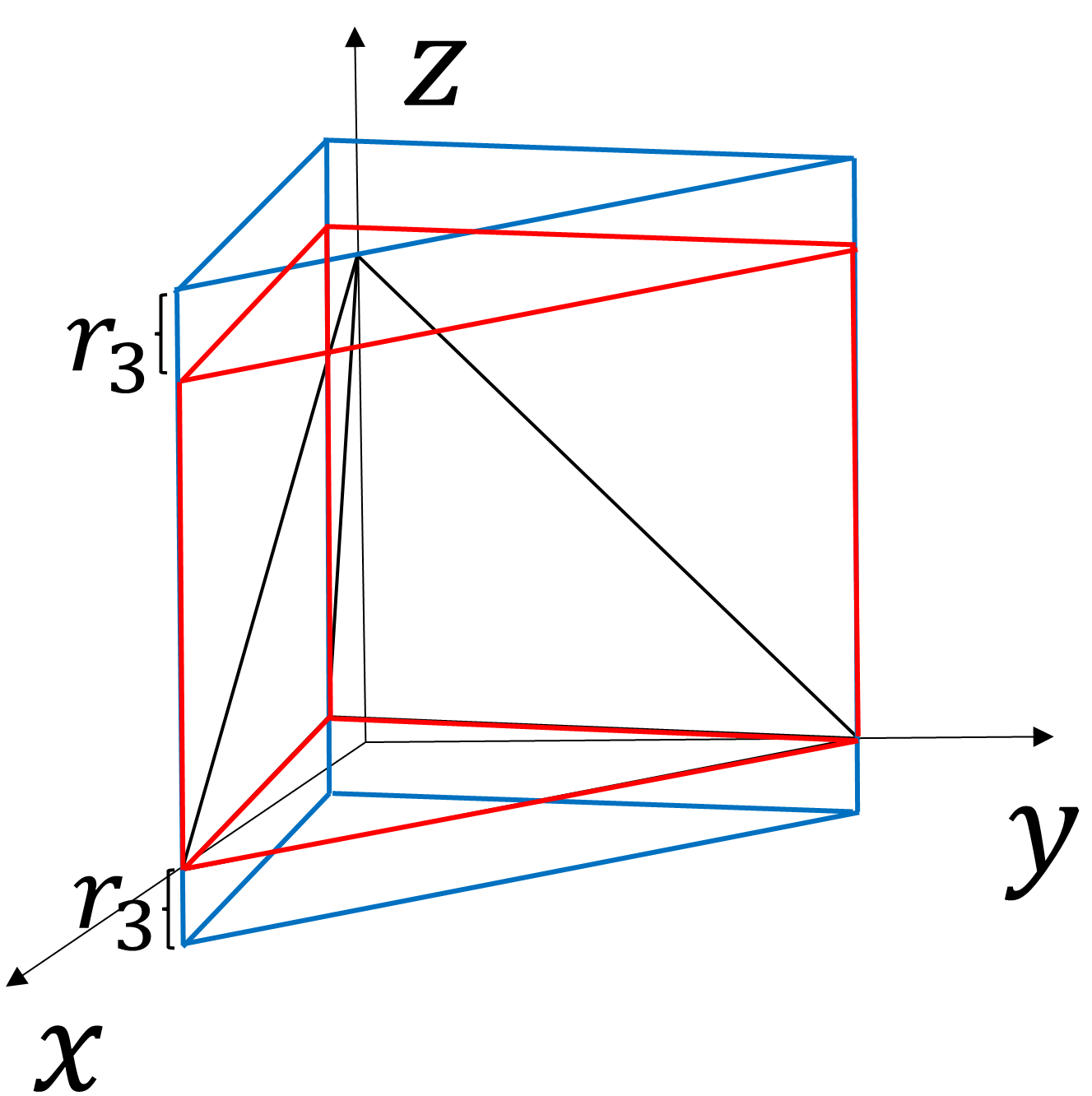}
		\caption{$\mathcal{A}$, $\mathcal{B}$ and $\mathcal{C}$}
		\label{fig:ABC}
	\end{subfigure}
\caption{Comparisons of Three Convex Hulls.}
\vspace{-10pt}
\label{fig:conv_compare}
\end{figure}

\par Inspired by \cite{abbas2020interplay}, we further use Hausdorff distance to measure the distance between the two convex hulls $\mathcal{A}_1$ and $\mathcal{A}_2$, as well as the diameter of a convex hull.
\begin{definition}\label{def:hausdorff}
Given two convex hulls $\mathcal{A}_1$ and $\mathcal{A}_2\subset \mathbb{R}^d$, the Hausdorff distance between $\mathcal{A}_1$ and $\mathcal{A}_2$ is defined as
\begin{equation}\label{eq:hausdorff}
D_H(\mathcal{A}_1,\mathcal{A}_2) = \max_{a_1\in\partial\mathcal{A}_1}\{\min_{a_2\in\partial\mathcal{A}_2}\{\mathrm{dist}(a_1,a_2)\}\},
\end{equation}
where $\mathrm{dist}(a_1,a_2)$ is the Euclidean distance between $a_1$ and $a_2$, and $\partial\mathcal{A}_1$, $\partial\mathcal{A}_2$ denote the boundaries of $\mathcal{A}_1$ and $\mathcal{A}_2$, respectively.
\end{definition}
\par We denote by $\mu(\mathcal{A}_1)$ the diameter of $\mathcal{A}_1$, i.e.,
\begin{equation}
\mu(\mathcal{A}_1)=\max_{a,b\in\partial\mathcal{A}_1}\{\mathrm{dist}(a,b)\}.
\end{equation}
We utilize both $D_H(\mathcal{A}_1,\mathcal{A}_2)$ and $\mu(\mathcal{A}_1)$ to quantify the distance between $\mathcal{A}_1$ and $\mathcal{A}_2$.

\begin{theorem}\label{th:distance}
The convex hull $\mathcal{C}$ satisfies 
\begin{equation}\label{eq:analysis1}
D_H(\mathcal{A},\mathcal{C}) \leq \sqrt{\frac{d}{2}}\;\mu(\mathcal{A})+\sqrt{r_1^2+r_2^2+\cdots+r_{\beta d}^2},     \nonumber   
\end{equation}
where $0 < \beta < 1$, $r_k$ is a freely chosen parameter for $k = 1, 2, \ldots, \beta d$, and $\mathcal{A}$ is the convex hull of the initial states of the normal agents. Under PP-ADRC, if the condition in \eqref{eq:neighbor} holds, $\overline{\mathcal{G}}(0),\overline{\mathcal{G}}(1),\overline{\mathcal{G}}(2),\ldots$ is repeatedly reachable, and $(1-\beta)$ fraction of the dimensions of the state vector are not added with noises, the probability $\mathrm{P}\{\xi \in \mathcal{C}\}$ that the final value lies within the convex hull $\mathcal{C}$ satisfies that     
\begin{equation}\label{eq:acc_conv}
\mathrm{P}\{\xi \in \mathcal{C}\}\geq \prod\limits_{k=1}^{\beta d} \left[1-\frac{\lambda^2/(1-\upsilon^2)}{(l_k+r_k)^2}\right].
\end{equation}
\end{theorem}
\begin{proof}
We characterize the upper bound of $D_H(\mathcal{A},\mathcal{C})$ by splitting it into $D_H(\mathcal{A},\mathcal{B})$ and $D_H(\mathcal{B},\mathcal{C})$. Thus, the proof is divided into two parts.
\par $D_H(\mathcal{A},\mathcal{B})$: Let us discuss a special case as presented in \cite{abbas2020interplay}. We define a new convex hull
$\mathcal{D}=\mathcal{A}_1\times\mathcal{A}_2\times\ldots\times\mathcal{A}_d$, where $\mathcal{A}_k=\operatorname{conv}(x_{1,k}(0),x_{2,k}(0),\ldots,x_{\overline{n},k}(0))$, $k=1,2,\ldots,d$. The example in $\mathbb{R}^3$ is shown in Fig. \ref{fig:axis-para}. Apparently, $D_H(\mathcal{A},\mathcal{D})$ is bound to greater than or equal to $D_H(\mathcal{A},\mathcal{B})$. Moreover, it is also obtained that $D_H(\mathcal{A},\mathcal{D})\leq \sqrt{\frac{d}{2}}\;\mu(\mathcal{A})$. Therefore, we obtain
$D_H(\mathcal{A},\mathcal{B})\leq D_H(\mathcal{A},\mathcal{D})\leq\sqrt{\frac{d}{2}}\;\mu(\mathcal{A})$.

\begin{figure}[htbp]
\centering
\begin{subfigure}{0.4\linewidth}
\vspace{-10pt}
		\centering
		\includegraphics[width=0.9\linewidth]{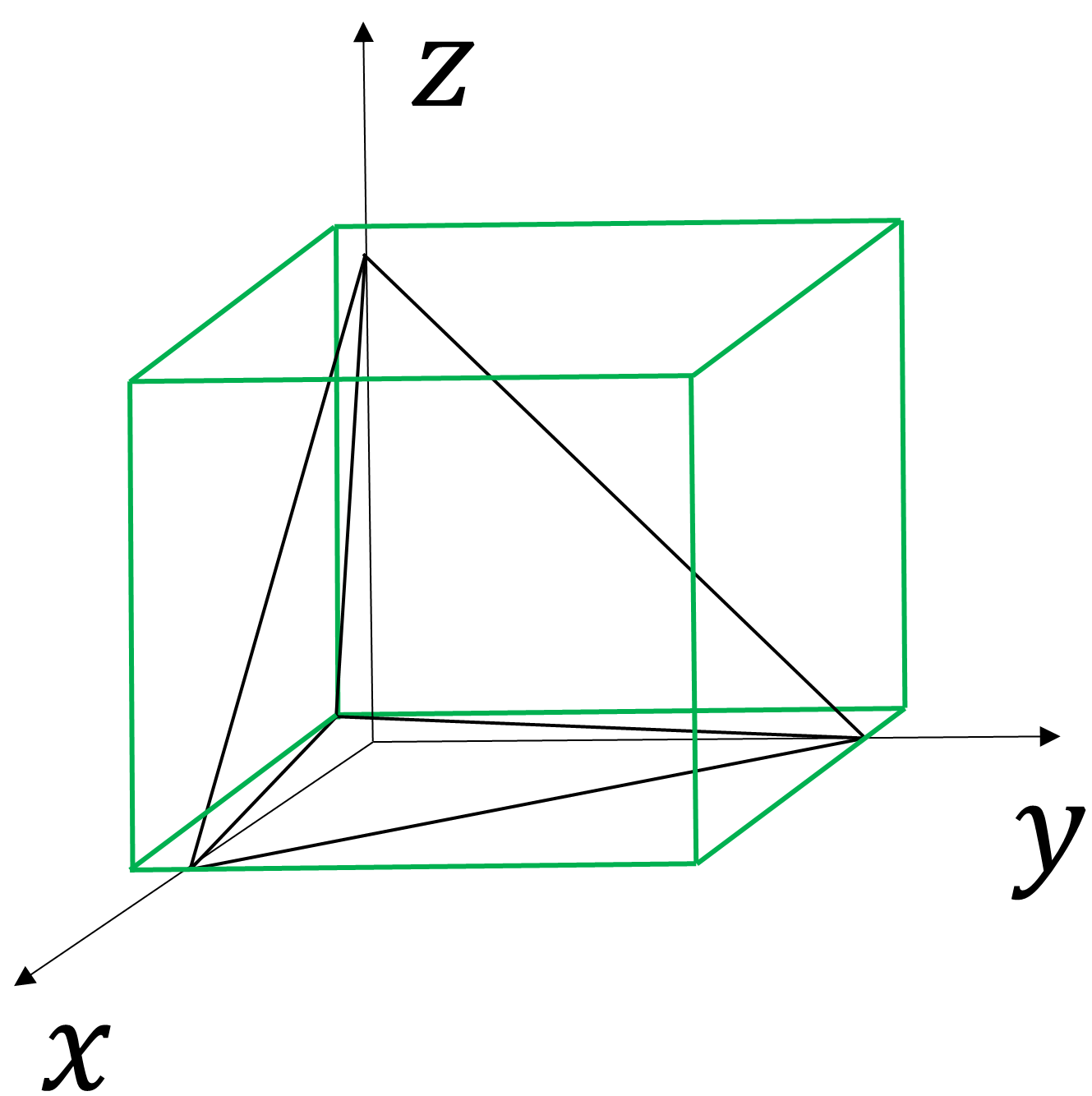}
		\caption{$\mathcal{A}$ and $\mathcal{D}$}
		\label{fig:AD}
	\end{subfigure}
 \begin{subfigure}{0.4\linewidth}
		\centering
		\includegraphics[width=0.9\linewidth]{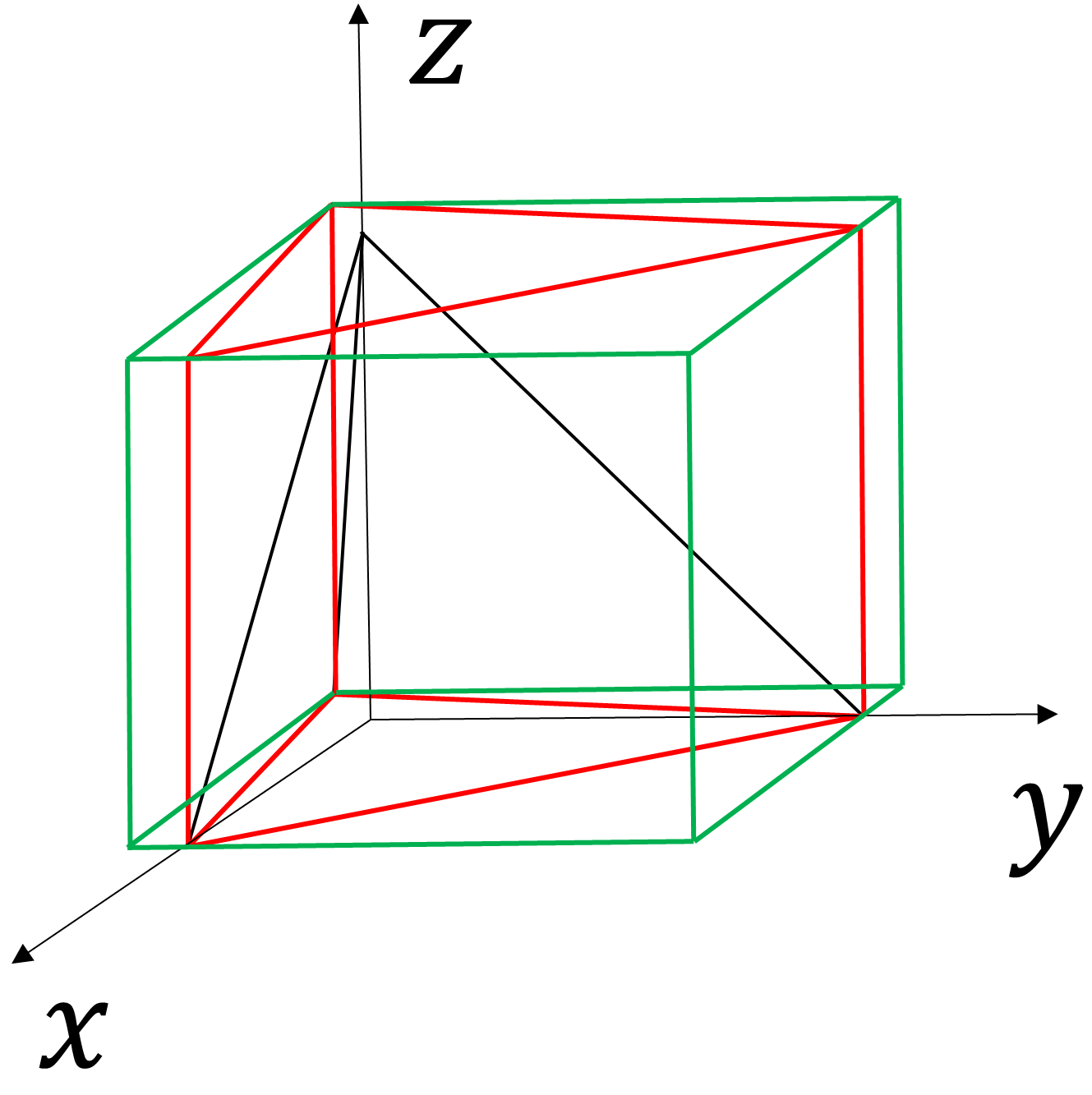}
		\caption{$\mathcal{A}$, $\mathcal{B}$, and $\mathcal{D}$}
		\label{fig:ABD}
	\end{subfigure}
\caption{Hyperrectangle in $\mathbb{R}^3$.}
\vspace{-5pt}
\label{fig:axis-para}
\end{figure}

\par $D_H(\mathcal{B},\mathcal{C}):$ we consider a surface of $\mathcal{B}$ and denote its vertices by $\mathcal{B}_1, \ldots, \mathcal{B}_{\kappa}$. We then characterize the nearest surface of $\mathcal{C}$ parallel to the surface we just mentioned, and the corresponding vertices are $\mathcal{C}_1,\ldots,\mathcal{C}_{\kappa}$. By recalling the definition of the Hausdorff distance, it must correspond to the distance between points on two surfaces. We denote by $b$ and $c$ any two points on the surface of $\mathcal{B}$ and $\mathcal{C}$, respectively. We can obtain $b=b_1\mathcal{B}_1+\ldots+b_{\kappa}\mathcal{B}_{\kappa}$ and $c=c_1\mathcal{C}_1+\ldots+c_{\kappa}\mathcal{C}_{\kappa}$, where $b_1+\ldots+b_\kappa=c_1+\ldots+c_{\kappa}=1$, and $b_i\geq0,c_i\geq0,\ i=1,2,\ldots,\kappa$. Moreover, from Definition \ref{def:convexhull}, we can obtain
$\mathcal{B}_i-\mathcal{C}_i=[\pm r_1,\pm r_2,\ldots,\pm r_{\beta d}]^\top, i=1,2,\ldots,\kappa$,
where the sign depends on the position of the surface. Hence, the distance $\mathrm{dist}(b,c)$ between $b$ and $c$ satisfies
\begin{equation}
\begin{aligned}
\mathrm{dist}(b,c)&=\sqrt{\sum\limits_{i=1}^{\kappa} ||b_i\mathcal{B}_i-c_i\mathcal{C}_i||_2^2}\\
&=\sqrt{\sum\limits_{i=1}^{\kappa} \left[(b_i-c_i)^2||\mathcal{C}_i||_2^2+b_i^2(r_1^2+\ldots+r_{\beta d}^2)\right]}.  
\end{aligned}
\nonumber
\end{equation}
First, we assume that $b_i$ is a constant and then find the lower bound of $\mathrm{dist}(b,c)$. Apparently, when $c_i=b_i$, $\mathrm{dist}(b,c)$ is lower bounded, which is $\min\limits_{c_i}(\mathrm{dist}(b,c))=\sqrt{\sum\limits_{i=1}^\kappa \left[b_i^2(r_1^2+\ldots+r_{\beta d}^2)\right]}$. Then, the problem becomes finding the upper bound of $\min\limits_{c_i}(\mathrm{dist}(b,c))$, which is denoted by $\max\limits_{b_i}\min\limits_{c_i}(\mathrm{dist}(b,c))$. Given that $b_1+\ldots+b_{\kappa}=1$ and $b_1^2+\ldots+b_{\kappa}^2\leq (b_1+\ldots+b_{\kappa})^2$, we can know that the upper bound of $b_1^2+\ldots+b_{\kappa}^2$ is $1$ when $b_{i_0}=1,\  i_0\in{1,2,\ldots,\kappa}\ \text{ and }\ b_{i}=0,\ i\neq i_0$. Therefore, we can obtain
$\max\limits_{b_i}\min\limits_{c_i}(\mathrm{dist}(b,c))=\sqrt{r_1^2+\ldots+r_{\beta d}^2}$.
According to the definition, $\max\limits_{b_i}\min\limits_{c_i}(\mathrm{dist}(b,c))$ is exactly the Hausdorff distance $D_H(\mathcal{B},\mathcal{C})$ between $\mathcal{B}$ and $\mathcal{C}$.
So we can obtain
\begin{equation}
\begin{aligned}
D_H(\mathcal{A},\mathcal{C})&\leq D_H(\mathcal{A},\mathcal{B})+D_H(\mathcal{B},\mathcal{C})\\
&\leq \sqrt{\frac{d}{2}}\;\mu(\mathcal{A})+\sqrt{r_1^2+r_2^2+\cdots+r_{\beta d}^2}.
\end{aligned}
\nonumber
\end{equation}
The corresponding probability satisfies 
\begin{equation}
\begin{aligned}
\mathrm{P}\{\xi \in \mathcal{C}\}&\geq \prod\limits_{k=1}^{\beta d} \left[1-\frac{\operatorname{var}(\xi_k)}{(l_k+r_k)^2}\right]\\
&\geq \prod\limits_{k=1}^{\beta d} \left[1-\frac{\lambda^2/(1-\upsilon^2)}{(l_k+r_k)^2}\right].
\end{aligned}
\nonumber
\end{equation}
\end{proof}
\begin{remark}
The one-dimensional results can yield an analytical solution because there is no need to consider the correlations between dimensions, making the analysis easier. The probability of convergence accuracy is influenced by several parameters. Specifically, if the Hausdorff distance between the convex hulls increases, the probability decreases. Similarly, an increase in the noise parameters $\lambda$ and $\upsilon$, which corresponds to an increase in variance, also leads to a decrease in probability.
\end{remark}

\section{Privacy Analysis}
\label{sec:privacy}
\par In this section, we analyze the $\rho$-CGP of PP-ADRC and compare it with $(\varepsilon,\delta)$-DP. At last, we make a detailed discussion on the relationship between convergence accuracy and privacy.
\subsection{\texorpdfstring{$\rho$-}{\rho}CGP in Resilient Vector Consensus}
\par We first investigate the $\rho$-CGP properties of PP-ADRC. As defined in Definition \ref{def:CGP2}, in the resilient vector consensus, the input is the initial states of normal agents $\overline{\mathbf{x}}_0$ and the output is the infinite sequence of noisy states $\mathbf{Y}=\{\overline{\mathbf{y}}(t)\}_{t=0}^{\infty}$.

\begin{theorem}
\label{th:cgp_rvs}
Under PP-ADRC, if the condition in \eqref{eq:neighbor} holds, and $\overline{\mathcal{G}}(0),\overline{\mathcal{G}}(1),\overline{\mathcal{G}}(2),\ldots$ is repeatedly reachable, then we achieve convergence under $\rho$-CGP, where $\rho=\frac{n\upsilon^2}{2\lambda^2[\upsilon^2-(1-\gamma_l)^2]}$.
\end{theorem}
\begin{proof}
Here, we present a general proof overview: Initially, we identify a variable that equalizes the outputs under two initial conditions following those in \cite{nozari2017differentially, fiore2019resilient}. Subsequently, we apply a variable transformation to standardize the integral domain and compute the R\'enyi divergence.

\par For any pair of initial conditions $\overline{\mathbf{x}}_0$ and $\overline{\mathbf{x}}_0'$ and an arbitrary set $\mathcal{O}\in \mathcal{B}((\mathbb{R}^{\overline{n}\times d})^\mathbb{N})$ of output, we consider the corresponding input domain.  For any $k\geq 0$, let the domain $R_k=\{ \mathbf{N}_k \in \Omega_k | \mathbf{Y}_{k,\overline{\mathbf{x}}_0} (\mathbf{N}_k) \in \mathcal{O}_k  \},$ where $\Omega_k=(\mathbb{R}^{\overline{n}\times d})^{k+1}$ is the sample space until time $k$ and $\mathcal{O}_k$ is the output set obtained by truncating the elements of $\mathcal{O}$ to finite subsequences 
of length $k+1$. We have the same definition for $\overline{\mathbf{x}}_0'$, which is $R_k'=\{ \mathbf{N}_k \in \Omega_k | \mathbf{Y}_{k,\overline{\mathbf{x}}_0'} (\mathbf{N}_k) \in \mathcal{O}_k \}$. Based on the continuity of probability \cite{durrett2019probability}, the probability can be written as
\begin{equation}
\mathrm{P}\{ \mathbf{N} \in \Omega | \mathbf{Y}_{\overline{\mathbf{x}}_0} (\mathbf{N}) \in \mathcal{O} \}=\lim_{k\to\infty} \int_{R_k} \mathsf{f}_{\overline{n}(k+1)}(\mathbf{N}_k) \mathrm{d}\mathbf{N}_k    
\end{equation}
and
\begin{equation}
\mathrm{P}\{ \mathbf{N} \in \Omega | \mathbf{Y}_{\overline{\mathbf{x}}_0'} (\mathbf{N}) \in \mathcal{O} \}=\lim_{k\to\infty} \int_{R_k'} \mathsf{f}_{\overline{n}(k+1)}(\mathbf{N}_k') \mathrm{d}\mathbf{N}_k',    
\end{equation}
where $\mathsf{f}_{\overline{n}(k+1)}$ represents the pdf of a joint multivariate Gaussian distribution in $\overline{n}(k+1)$ dimensions, defined as
\begin{equation}
\label{eq:joint_pdf}
\begin{split}
\mathsf{f}_{\overline{n}(k+1)}(\mathbf{N}_k)&= \prod\limits_{h=0}^k \prod\limits_{i=1}^{\overline{n}} G(\eta_i(h);\Sigma(h)).
\end{split}
\end{equation}
Therefore, we can determine the distribution of the output of the first execution, $\mathbf{y}_{\overline{\mathbf{x}}_0}(\mathbf{N}) = \lim\limits_{k \to \infty} \mathsf{f}_{\overline{n}(k+1)}(\mathbf{N}_k)$, and the output of the second execution, $\mathbf{y}_{\overline{\mathbf{x}}_0'}(\mathbf{N}) = \lim\limits_{k \to \infty} \mathsf{f}_{\overline{n}(k+1)}(\mathbf{N}_k')$.
 The next step is to calculate the Rényi divergence of these two distributions, but the integral variables are different, making it unrealizable for us to make a direct computation. Therefore, we need to find the relationship between $\mathbf{N}$ and $\mathbf{N}'$ so that the integral variables and domains of $\mathbf{y}_{\overline{\mathbf{x}}_0}(\mathbf{N})$ and $\mathbf{y}_{\overline{\mathbf{x}}_0'}(\mathbf{N}')$ are the same, differing only in the integral expression. Given that $\overline{\mathbf{x}}_0$ and $\overline{\mathbf{x}}_0'$ are two arbitrary initial conditions, we assume that $x_i'(0) = x_i(0) + \delta_i$, $i = 1, 2, \ldots, \overline{n}$, where $\delta_i \in \mathbb{R}^d$ and $\|\delta_i\|_2 \leq \mathrm{dist}(\overline{\mathbf{x}}_0, \overline{\mathbf{x}}_0')$. Then, for any $\mathbf{N}_k \in  R_k$, we define $\mathbf{N}_k'$ by
\begin{equation}
\label{eq:privacy_noise}
\eta_i'(h)=\begin{cases} \eta_i(h)-\prod_{t=0}^{h-1}\left[1-\gamma_i(t)\right]\delta_i&\text{if }h>0.\\ \eta_i(h)-\delta_i &\text{if }h=0,\end{cases}
\end{equation}
for $i=1,2,\cdots,\overline{n}$. 
\par Then, we show that the outputs $\mathbf{Y}_{k,\overline{\mathbf{x}}_0}(\mathbf{N}_k)$ and $\mathbf{Y}_{k,\overline{\mathbf{x}}_0'}(\mathbf{N}_k')$ share the same distribution under the above definitions. This is proven by the method of mathematical induction. From equation \eqref{eq:add_noise}, we can derive $y_i'(0)=x_i'(0)+\eta_i'(0)=x_i(0)+\delta_i+\eta_i(0)-\delta_i=y_i(0), i=1,2,\cdots,\overline{n}$. We can conclude that for every normal agent $i$, the calculated centerpoints $s_i(0)$ and $s_i'(0)$ are the same. According to equation \eqref{eq:update}, we can derive $x_i'(1)=x_i(1)+(1-\gamma_i(0))\delta_i$. By induction, we can easily obtain that $x_i'(h)=x_i(h)+\prod_{t=0}^{h-1}\left[1-\gamma_i(t)\right]\delta_i, h=1,2,\cdots,k$, followed by $y_i(h)=y_i'(h)$. It means that $\mathbf{Y}_{k,\overline{\mathbf{x}}_0}(\mathbf{N}_k)=\mathbf{Y}_{k,\overline{\mathbf{x}}_0'}(\mathbf{N}_k')$, indicating that $\mathbf{N}_k'\in R_k'$ and we successfully build a bijective correspondence. For any $\mathbf{N}_k'\in R_k'$, there exists $(\mathbf{N}_k,\Delta\mathbf{N}_k)\in R_k \times (\mathbb{R}^{\overline{n}\times d})^{k+1}$ such that $\mathbf{N}_k'=\mathbf{N}_k+\Delta\mathbf{N}_k$. Therefore, we derive the integral by changing variables $\mathrm{P}\{ \mathbf{N} \in \Omega | \mathbf{Y}_{\overline{\mathbf{x}}_0'} (\mathbf{N}) \in \mathcal{O} \}=\lim\limits_{k\to\infty} \int_{R_k} \mathsf{f}_{\overline{n}(k+1)}(\mathbf{N}_k+\Delta\mathbf{N}_k ) \mathrm{d}\mathbf{N}_k,$ indicating that the distribution $\mathbf{y}_{\overline{\mathbf{x}}_0'}(\mathbf{N})$ can be expressed as $\lim\limits_{k\rightarrow \infty}\mathsf{f}_{\overline{n}(k+1)}(\mathbf{N}_k+\Delta\mathbf{N}_k )$.
\par With all the preparations above, we can further analyze the Rényi divergence of $\mathbf{y}_{\overline{\mathbf{x}}_0}(\mathbf{N})$ and $\mathbf{y}_{\overline{\mathbf{x}}_0'}(\mathbf{N})$ as equation \eqref{eq:renyi_inter1}. 
\begin{figure*}[htbp]
\vspace*{8pt} 
\begin{equation}
\label{eq:renyi_inter1}
\begin{aligned}
&D_\alpha(\mathbf{y}_{\overline{\mathbf{x}}_0}(\mathbf{N})\|\mathbf{y}_{\overline{\mathbf{x}}_0'}(\mathbf{N}))\\
&=\lim_{k\to\infty} \frac{1}{\alpha-1} \mathrm{log} \bigg[\int  [\mathbf{y} _{\overline{\mathbf{x}}_0}(\mathbf{N}_k)]^\alpha [\mathbf{y}_{\overline{\mathbf{x}}_0'}(\mathbf{N}_k)]^{1-\alpha} \text{d}\mathbf{N}_k\bigg]\\
&=\lim_{k\to\infty}\frac{1}{\alpha-1} \mathrm{log} \bigg[ \int_{R_k} [\mathsf{f}_{\overline{n}(k+1)}(\mathbf{N}_k)]^\alpha[\mathsf{f}_{\overline{n}(k+1)} (\mathbf{N}_k+\Delta\mathbf{N}_k)]^{1-\alpha} \mathrm{d} \mathbf{N}_k\bigg]\\
\end{aligned}
\end{equation}
\end{figure*}

\begin{figure*}[htbp]
\vspace*{8pt} 
Then, we need to calculate detailed expressions of $[\mathsf{f}_{\overline{n}(k+1)}(\mathbf{N}_k)]^\alpha[\mathsf{f}_{\overline{n}(k+1)} (\mathbf{N}_k+\Delta\mathbf{N}_k)]^{1-\alpha}$.
\begin{equation}
\label{eq:renyi_inter2}
\begin{aligned}
&[\mathsf{f}_{\overline{n}(k+1)}(\mathbf{N}_k)]^\alpha[\mathsf{f}_{\overline{n}(k+1)} (\mathbf{N}_k+\Delta\mathbf{N}_k)]^{1-\alpha}\\
&=\bigg[\mathsf{f}_{\overline{n}(k+1)}(\mathbf{N}_k)\bigg]^\alpha \times \bigg[\mathsf{f}_{\overline{n}(k+1)}(\mathbf{N}_k+\Delta\mathbf{N}_k)\bigg]^{1-\alpha}\\
&=\prod\limits_{h=0}^k\bigg[\prod\limits_{i=1}^{\overline{n}}G(\eta_i(h);\Sigma(h))\bigg]^\alpha \times \bigg[\prod\limits_{i=1}^{\overline{n}}G(\eta_i(h)+\Delta\eta_i(h);\Sigma(h))\bigg]^{1-\alpha}\\
&=\prod\limits_{h=0}^k\prod\limits_{i=1}^{\overline{n}}\bigg[\frac{1}{\sqrt{(2\pi \sigma^2(h))^d}}\mathrm{exp}\left(-\frac{\alpha||\eta_i(h)||^2_2}{2\sigma^2(h)}\right)\cdot \mathrm{exp}\left(-\frac{(1-\alpha)||\eta_i(h)+\Delta\eta_i(h)||^2_2}{2\sigma^2(h)}\right)\bigg]\\
&=\prod\limits_{h=0}^k\prod\limits_{i=1}^{\overline{n}}\bigg[\frac{1}{\sqrt{(2\pi \sigma^2(h))^d}}\mathrm{exp}\left(-\frac{||\eta_i(h)-(1-\alpha)\Delta\eta_i(h)||^2_2-||(1-\alpha)\Delta\eta_i(h)||^2_2+(1-\alpha)||\Delta\eta_i(h)||^2_2}{2\sigma^2(h)}\right)\bigg]
\end{aligned}
\end{equation}
\end{figure*}

\begin{figure*}[htbp]
\vspace*{8pt}
Substituting equation \eqref{eq:renyi_inter2}  into \eqref{eq:renyi_inter1} , we can obtain the Rényi divergence between the two distributions.
\begin{equation}
\label{eq:renyi_inter3}
\begin{aligned}
& D_\alpha(\mathbf{y}_{\overline{\mathbf{x}}_0}(\mathbf{N})\|\mathbf{y}_{\overline{\mathbf{x}}_0'}(\mathbf{N}))\\
&=\lim_{k\to\infty}\frac{1}{\alpha-1}\mathrm{log}\bigg\{\int_{R_{k,1}}\int_{R_{k,2}}\cdots\int_{R_{k,\overline{n}}}\\
& \prod\limits_{h=0}^k\prod\limits_{i=1}^{\overline{n}}\bigg[\frac{1}{\sqrt{(2\pi \sigma^2(h))^d}}\mathrm{exp}\left(-\frac{||\eta_i(h)-(1-\alpha)\Delta\eta_i(h)||^2_2-||(1-\alpha)\Delta\eta_i(h)||^2_2+(1-\alpha)||\Delta\eta_i(h)||^2_2}{2\sigma^2(h)}\right)\bigg]\\
&\mathrm{d}\eta_1(h)\mathrm{d}\eta_2(h)\cdots\mathrm{d}\eta_{\overline{n}}(h)\bigg\}\\
&=\lim_{k\to\infty}\frac{1}{\alpha-1} \mathrm{log}\prod\limits_{h=0}^k\{\prod\limits_{i=1}^{\overline{n}}\underset{\eta_i(h){\sim}G((1-\alpha)\Delta\eta_i(h),\sigma^2(h)I_d)}{\mathbb{E}}\bigg[\mathrm{exp}\left(-\frac{-||(1-\alpha)\Delta\eta_i(h)||^2_2+(1-\alpha)||\Delta\eta_i(h)||^2_2}{2\sigma^2(h)}\right)\bigg]\}\\
&=\lim_{k\to\infty}\sum\limits_{h=0}^k\bigg[\sum\limits_{i=1}^{\overline{n}}\frac{\alpha ||\Delta \eta_i(h)||^2_2}{2\lambda^2\upsilon^{2h}}\bigg] \leq \bigg[\frac{\alpha \overline{n}}{2\lambda^2}+\lim_{k\to\infty}\sum\limits_{h=1}^k \frac{\alpha \overline{n}(1-\gamma_l)^{2h}}{2\lambda^2\upsilon^{2h}}\bigg]\mathrm{dist}(\overline{\mathbf{x}}_0,\overline{\mathbf{x}}_0')^2\\
&\leq \alpha\bigg[\frac{n\upsilon^2}{2\lambda^2[\upsilon^2-(1-\gamma_l)^2]}\bigg]\mathrm{dist}(\overline{\mathbf{x}}_0,\overline{\mathbf{x}}_0')^2.
\end{aligned}
\end{equation}
Here, we use $n$ instead of $\overline{n}$, as $n$ is a known parameter. It is worth noting that, as mentioned earlier, the information exchange during each iteration results in decaying leakage of the initial states. As derived from the inequality above, there is a decay factor of $(1-\gamma_l)$. Therefore, adding a decaying noise is sufficient to achieve privacy preservation.
\end{figure*}

\end{proof}
\begin{remark}\label{rm:CGP}
We can observe that the results do not contain any parameters related to faulty agents, indicating that their presence does not affect the privacy preservation outcome. While this may seem counterintuitive, it is actually reasonable. Each normal agent independently adds Gaussian noise to its local state before transmission to ensure privacy preservation, and this noise addition process is entirely unrelated to faulty agents. Once the information is transmitted, it may be intercepted by various entities, including faulty agents, honest-but-curious normal agents, and external eavesdroppers. However, all of these entities receive the same information—the noisy state transmitted by each normal agent—without any additional insights. Furthermore, the degree of privacy leakage of the initial states at nonzero iterations is determined by the step size $\gamma$ \cite{fiore2019resilient,nozari2017differentially}. Our analysis specifically addresses privacy preservation under these conditions, and thus the behavior of faulty agents does not impact the overall privacy preservation performance. Moreover, given the variables $\lambda$ and $\upsilon$, we plot the surface of $\rho$ concerning them, as shown in Fig. \ref{fig:cgp}. It can be observed that $\rho$ monotonically decreases as both $\lambda$ and $\upsilon$ increase.
\end{remark}
   \begin{remark}
     Gaussian noise is canonical under CGP but may not be optimal. It is both interesting and essential to identify the optimal noise that best aligns with the specific requirements of the scenario, such as achieving the highest privacy preservation given the same noise strength. This problem presents significant challenges, especially in resilient vector consensus where we consider arbitrarily high-dimensional spaces. The complexity arises from two key aspects: First, the noise has numerous dimensional components, and second, these components can exhibit complicated coupling relationships across different dimensions. These characteristics substantially increase the difficulty in determining the optimal noise distribution. Addressing this challenging issue will be a major focus of our future research.
   \end{remark}
\begin{figure}[htbp]
\vspace{-10pt}
    \centering
    \includegraphics[width=0.6\linewidth]{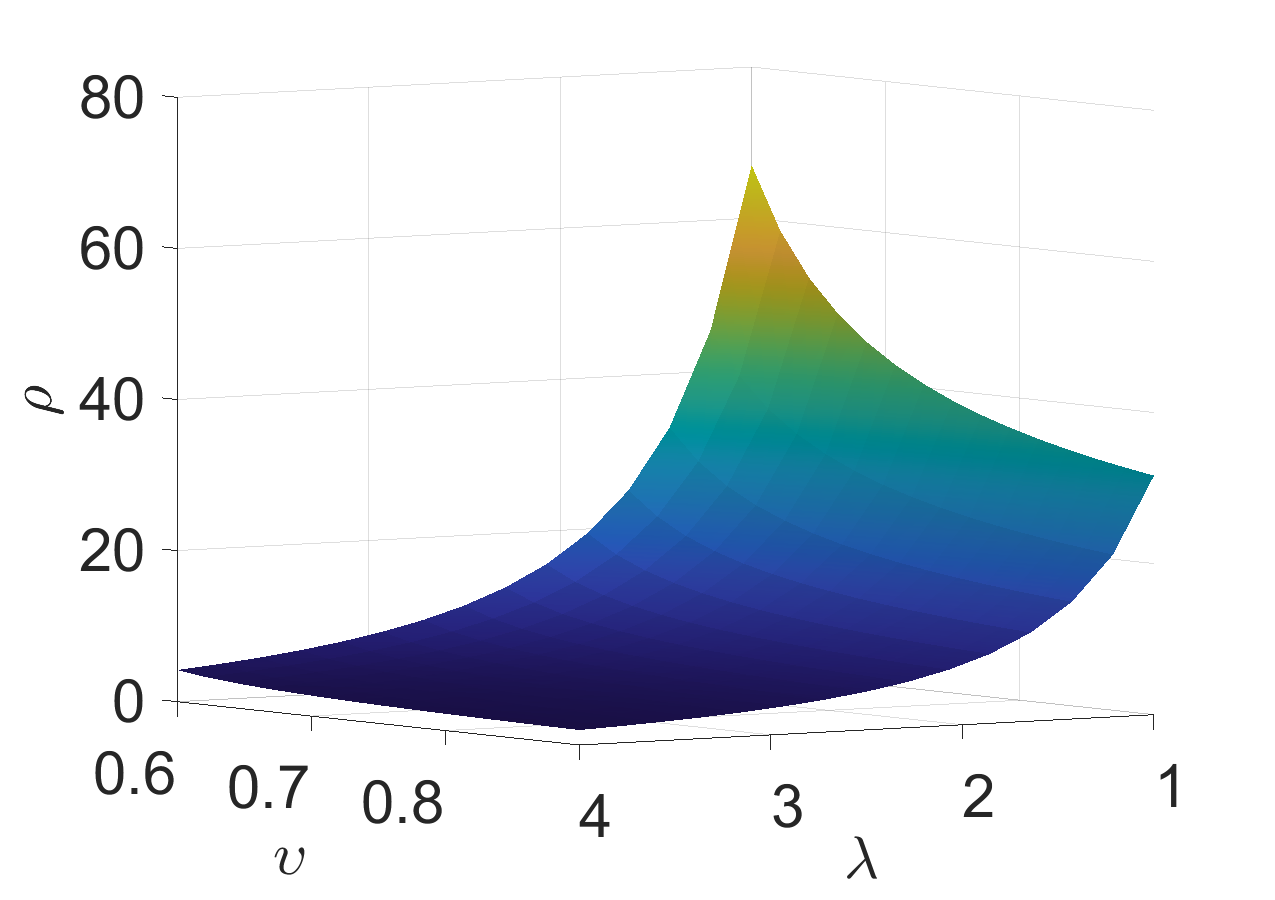}
    \caption{Influence of parameters $\upsilon$ and $\lambda$ on the coefficient $\rho$}
    \label{fig:cgp}
    \vspace{-2pt}
\end{figure}

\subsection{Comparison with $(\varepsilon,\delta)$-DP}
\par We first give the definition of  $(\varepsilon,\delta)$-DP in resilient vector consensus.
\begin{definition}
Given $\ell\in\mathbb{R} \geq 0$, the initial states of normal agents $\overline{\mathbf{x}}_0$ and $\overline{\mathbf{x}}_0'$ are $\ell$-neighboring if, for some $i_0\in \overline{\mathcal{V}}, ||[\overline{\mathbf{x}}_0]_{i_0}-[\overline{\mathbf{x}}_0']_{i_0}||_2\leq \ell$ and $[\overline{\mathbf{x}}_0]_i=[\overline{\mathbf{x}}_0']_i$ if $i\neq i_0$ and $i\in\overline{\mathcal{V}}$. Given $\ell$, $\varepsilon$, and $\delta \in \mathbb{R} \geq 0$, the system satisfies $(\varepsilon,\delta)$-differential privacy if, for any pair of $\overline{\mathbf{x}}_0$ and $\overline{\mathbf{x}}_0'$ of $\ell$-neighboring initial states and any set $\mathcal{O}\in \mathcal{B}((\mathbb{R}^{\overline{n}\times d})^\mathbb{N})$,
\begin{equation}\label{eq:dp}
\mathrm{P}\{\mathbf{N}\in\Omega|\mathbf{Y}_{\overline{\mathbf{x}}_0}(\mathbf{N})\in\mathcal{O}\}\leq e^{\varepsilon}\mathrm{P}\{\mathbf{N}\in\Omega|\mathbf{Y}_{\overline{\mathbf{x}}_0'}(\mathbf{N})\in\mathcal{O}\}+\delta.    
\end{equation}
\end{definition}
\begin{remark}
The definition of neighboring initial states here differs from those in existing papers \cite{liu2024trade,fiore2019resilient,nozari2017differentially}, where the $1$-norm is used. In those papers, differential privacy is $\varepsilon$-DP with Laplace noise, whereas we use $(\varepsilon, \delta)$-DP with Gaussian noise. Therefore, the neighboring initial states are defined using the $2$-norm.
\end{remark}
\par As previously discussed, the system's output is an infinite sequence. The output at each time can be considered as a query. Thus, the final result is the infinite composition of each single query at time $h$ satisfying $(\varepsilon(h),\delta(h))$-DP. We now detail $\varepsilon(h)$ and $\delta(h)$ according to the Theorem A.1. in \cite{dwork2014algorithmic}.
\begin{lemma}
For $h=0,1,2,\ldots$, the output $\overline{\mathbf{y}}(h)$ at each time $h$ can be considered as a query, satisfying $(\varepsilon(h),\delta(h))$-DP. Given $\delta(h)\in \mathbb{R}>0$, $\varepsilon(h)$ holds that $\varepsilon(h)>\frac{\sqrt{2\mathrm{log}\frac{1.25}{\delta(h)}}\ell(1-\gamma_l)^h}{\lambda \upsilon^{h}}$.
\end{lemma}
\par Regarding the composition, we can simply sum $\varepsilon(h)$ and $\delta(h)$ from $h=0$ to $\infty$. This leads to a rapid increase in privacy loss and $\delta$, resulting in a poor outcome. However, computing the tightest possible privacy guarantee for such a composition is \#P-hard \cite{murtagh2015complexity}. Then, we can make a detailed comparison between these two definitions below.
\begin{itemize}
\item \textbf{CGP better meets the privacy preservation needs.} In resilient vector consensus, our privacy preservation target is the initial states of all the normal agents, which may be leaked during communication, rather than privacy leakage caused by the alteration or addition/deletion of an initial state in a data center. Therefore, the definition of $\ell$-neighboring states in DP does not fully meet our needs, as it should consider any pair of initial states. Consequently, CGP describes the degree of output differences for any pair of inputs, making it more suitable for protecting the initial state of each normal agent under resilient vector consensus. It also uses the distance between initial states as a privacy preservation parameter, thereby greatly improving the system's utility.

\item \textbf{CGP admits better advanced composition.} In advanced composition, privacy loss is linear with the square root of the number of queries. Although $(\varepsilon,\delta)$-DP admits advanced composition, the computation process is quite complicated and would make $\delta$ grow linearly, increasing the risk of privacy leakage. However, the advanced composition of CGP is quite straightforward: It involves simply summing up the $\rho$ values without considering any additional parameters. In resilient vector consensus, we consider the noisy states sent by the agents from time zero to infinity as outputs. Each output at every iteration is considered a potential privacy leakage query. The final result we obtain is essentially the outcome of an infinite number of queries. Using CGP results in less privacy loss and easier computation compared to $(\varepsilon,\delta)$-DP, particularly in scenarios with a large number of queries, as it provides a stricter and more accurate upper bound on privacy preservation.
\end{itemize}
\subsection{Trade-off between Privacy and Accuracy}
In this part, we discuss the trade-off between privacy and accuracy. As we analyze the convergence accuracy from two perspectives, i.e., distribution of the final value and convex hull change, we also detail the trade-off from these two perspectives. From Theorem \ref{th:cgp_rvs}, the parameters that determine $\rho$ are provided. The parameters that we can freely choose are the noise parameters $\lambda$ and $\upsilon$. The larger $\lambda$ and $\upsilon$ are, the stronger the privacy preservation is. Therefore, we analyze the trade-off between accuracy and privacy by varying $\lambda$ and $\upsilon$.
\begin{itemize}
\item Distribution of the final value: The Mahalanobis distance is a normalized coefficient, meaning its value is fixed artificially and remains uninfluenced by variance or covariance. However, the volume of the hyperspheroid formed at an equal Mahalanobis distance can be influenced by various parameters. Therefore, the metric that accurately represents precision should be the volume of the hyperspheroid: the larger the volume, the lower the precision. As previously analyzed, increasing the amplitude parameters $\lambda$ and $\upsilon$ of the noise increases the volume of the hyperspheroid, thereby reducing convergence accuracy and decreasing $\rho$, which corresponds to enhanced privacy preservation. 
\item Convex hull change: We use the Hausdorff distance to characterize the distance between convex hulls before and after adding noise. Similar to the Mahalanobis distance, the Hausdorff distance is also specified artificially. However, the probability of convergence within the convex hull is influenced by the noise parameters. The larger the noise amplitude parameters $\lambda$ and $\upsilon$, the lower the probability that the new convex hull will converge within a given Hausdorff distance. This results in a decrease in convergence accuracy but simultaneously leads to a reduction in $\rho$, thereby enhancing privacy preservation.
\end{itemize}
  \par  We can observe that although privacy and resilience are independent in algorithm design, their performance is tightly coupled. Identifying the appropriate noise parameters to achieve a balanced trade-off between them is crucial.
  \par In fact, in addition to the trade-off between privacy and accuracy, there also exists a trade-off between privacy and convergence rate. From a qualitative perspective, increasing the noise magnitude enhances privacy but also slows convergence due to a slower decay of noise. However, we cannot establish a precise trade-off between privacy and convergence rate, as quantifying the convergence rate is infeasible. Analyzing the convergence rate presents two main challenges: First, while the centerpoint lies within the convex hull of the normal agents and is generally near the center, its exact parameters and convex combination bounds remain undetermined. This uncertainty complicates the calculation of agents' movement at each iteration. As a result, most resilient vector consensus methods encounter these challenges and do not provide a quantitative analysis of the convergence rate, as in \cite{park2017fault,abbas2022resilient,yan2022resilient}. Second, the addition of noise increases complexity by introducing randomness into both the direction and magnitude of each agent’s movement, further complicating the analysis of the convergence rate.

\section{Simulations}
\label{sec:sim}
\par In this section, we consider the simulation on the real-world physical system, which is a multi-robot rendezvous \cite{park2017fault}. We perform extensive simulations of a multi-robot system performing PP-ADRC in both $\mathbb{R}^2$ and $\mathbb{R}^3$ to illustrate our theoretical results.
\subsection{$2$-dimensional case}
\par 
In this scenario, we deploy $10$ robots with a time-varying directed network, consisting of $2$ faulty robots and $8$ normal robots. The messages sent by the two faulty agents to each normal agent are randomly selected within the region $[-0.7,-0.3] \times [0.3,0.7]$, causing the normal agents to converge to a point outside the convex hull, specifically at its upper-left corner. Each normal robot's in-neighborhood is selected randomly at each iteration while ensuring the conditions of convergence. The randomly generated initial states are depicted in Fig. \ref{fig:initial}, where blue points represent normal robots' states and the red ones are the states of faulty robots. The green polygon represents the convex hull of normal robots. For parameters, we set $\gamma_i(t)=0.8,\ \lambda=2.0$, and $\upsilon=0.75$.
\begin{figure}[htbp]
\centering
\begin{subfigure}{0.4\linewidth}
		\centering
		\includegraphics[width=1.0\linewidth]{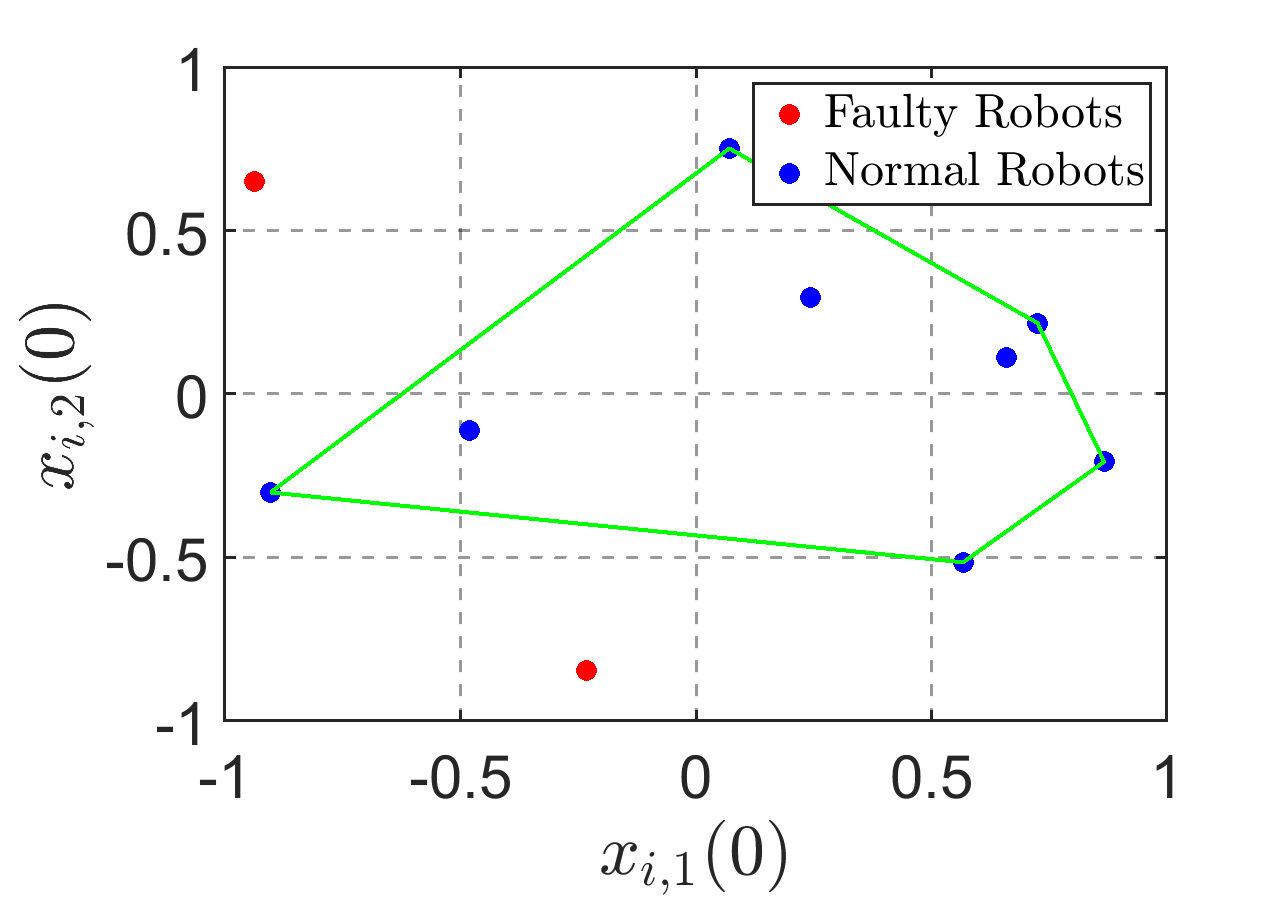}
		\caption{Initial states}
		\label{fig:initial}
	\end{subfigure}
 \begin{subfigure}{0.4\linewidth}
		\centering
		\includegraphics[width=1.0\linewidth]{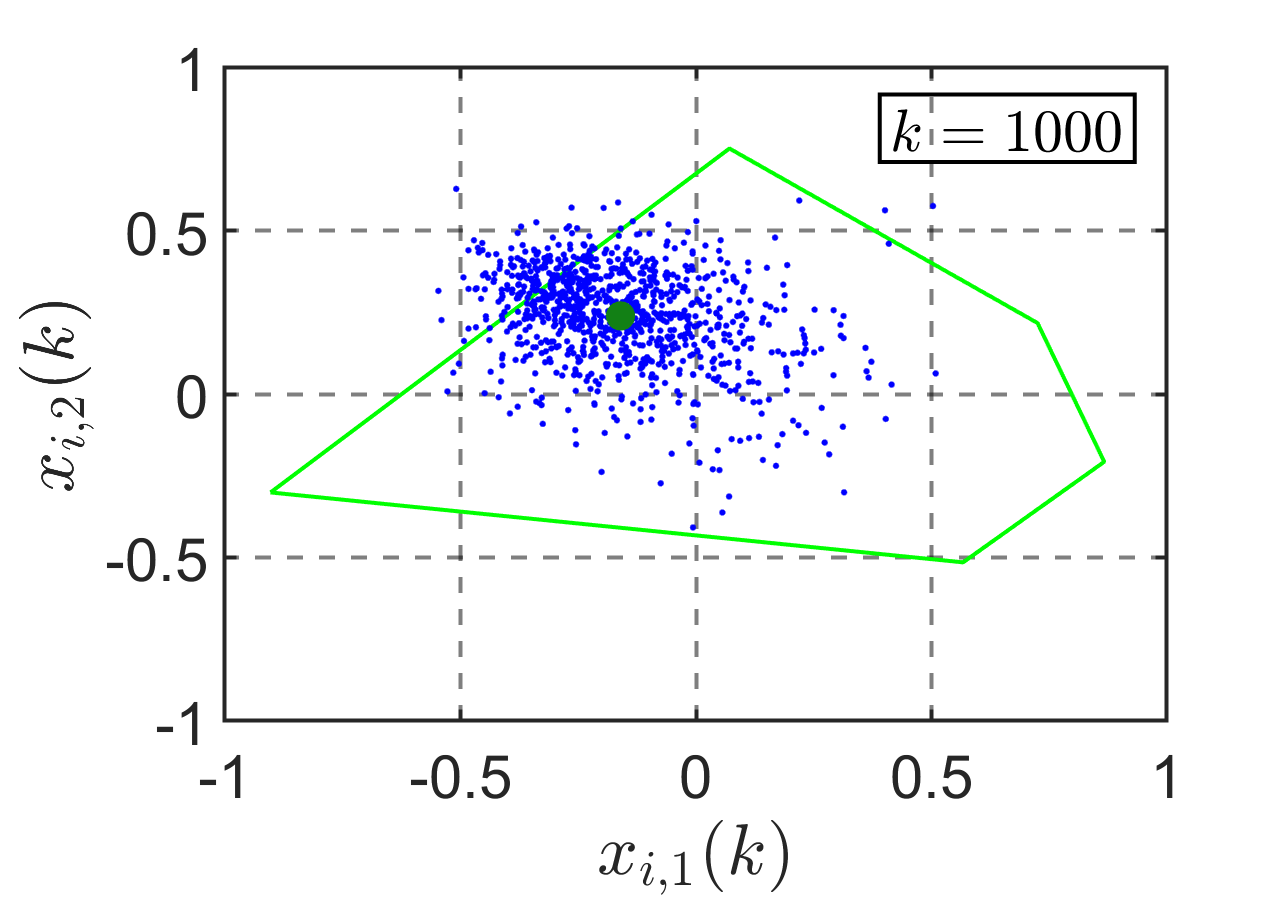}
		\caption{Final value}
		\label{fig:converge_gau}
	\end{subfigure}
\caption{$2$-dimensional case}
\label{fig:expectation}
\vspace{-5pt}
\end{figure}
\par We execute the algorithm $1000$ times, each with $1000$ iterations, and obtain the result in Fig. \ref{fig:converge_gau}, where blue points represent the final values each time, and the green point $(-0.160,0.239)$ is the sample mean of the blue points. We observe that the expectation falls in the convex hull of the initial states of the normal robots.
\par Then, we only add noise to the second dimension to illustrate Theorem \ref{th:distance}. We repeat the process $1000$ times, with $1000$ iterations for three different cases of noise parameters (case 1): $\lambda=2.0,\upsilon=0.75$, case 2): $\lambda=2.5,\upsilon=0.75$, and case 3): $\lambda=2.5,\upsilon=0.85$) and obtain the convergence results shown in Fig. \ref{fig:pr}. The yellow convex hull represents $\mathcal{C}$, with $r_2=0.3$. The variances of the second dimension in order are $0.066,0.117,$ and $0.219$, respectively. In Fig. \ref{fig:pr_compare}, the vertical axis represents the probability of the final value falling within the convex hull and the horizontal axis represents $r_2$. There are two sets of lines, i.e., one obtained through simulation probabilities $\mathrm{P}^{\rm Si}$ and the other $\mathrm{P}^{\rm Th}$ derived theoretically from Theorem \ref{th:distance}. It is observed that in all three cases, the simulation results $\mathrm{P}^{\rm Si}$ are greater than the theoretical ones $\mathrm{P}^{\rm Th}$. When the distance satisfies $r_2=0.3$, we have $\delta(\mathcal{A},\mathcal{C})=0.919$ and $\mu(\mathcal{A})=1.77$, which illustrates our theoretical result  $\delta(\mathcal{A},\mathcal{C})\leq \sqrt{\frac{d}{2}}\mu(\mathcal{A})+r_2$. 
\begin{figure}[htbp]
	\centering
	\begin{subfigure}{0.4\linewidth}
		\centering
		\includegraphics[width=1.0\linewidth]{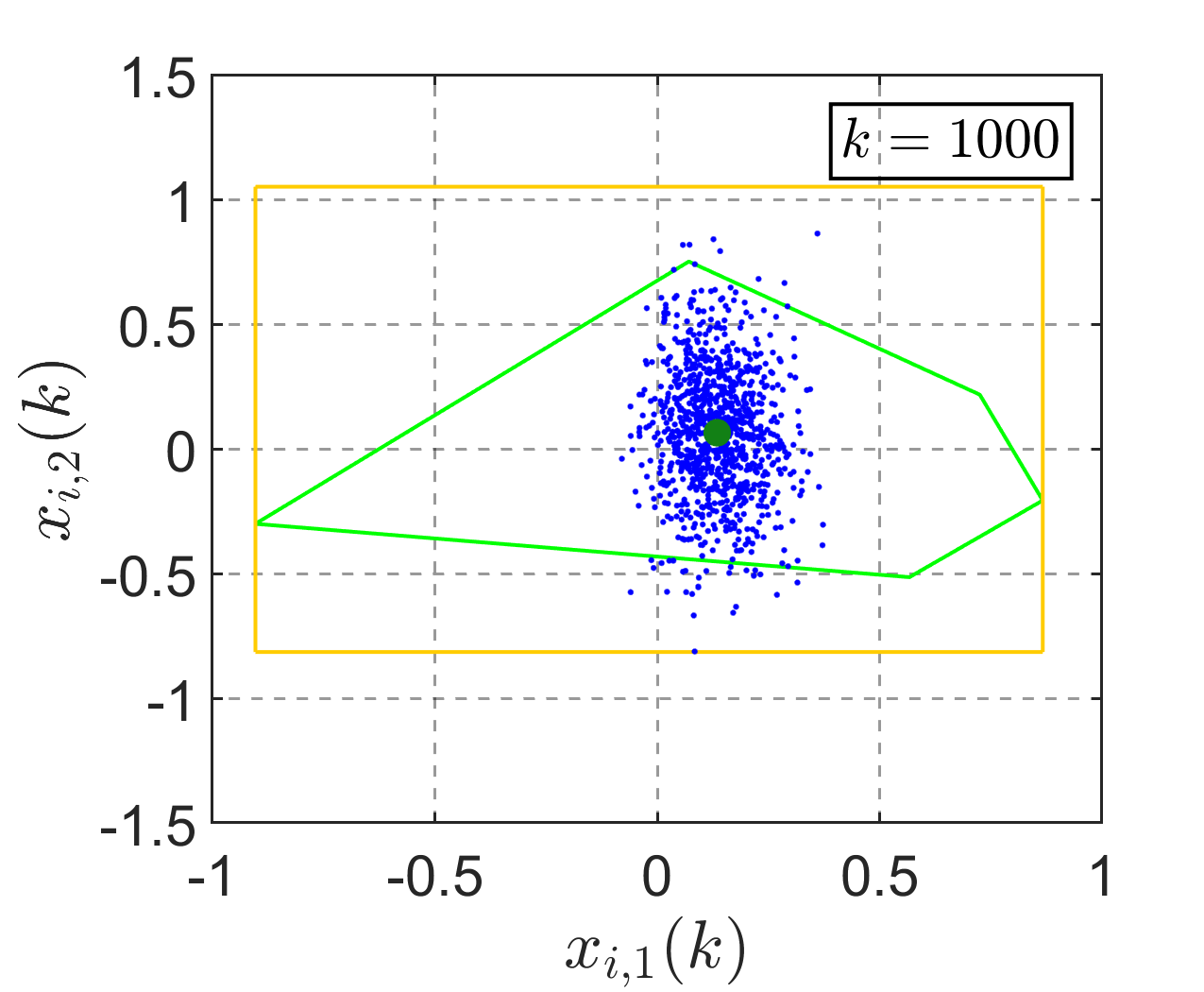}
		\caption{$\lambda=2.0,\upsilon=0.75$}
		\label{fig:pr1}
	\end{subfigure}
	\begin{subfigure}{0.4\linewidth}
		\centering
		\includegraphics[width=1.0\linewidth]{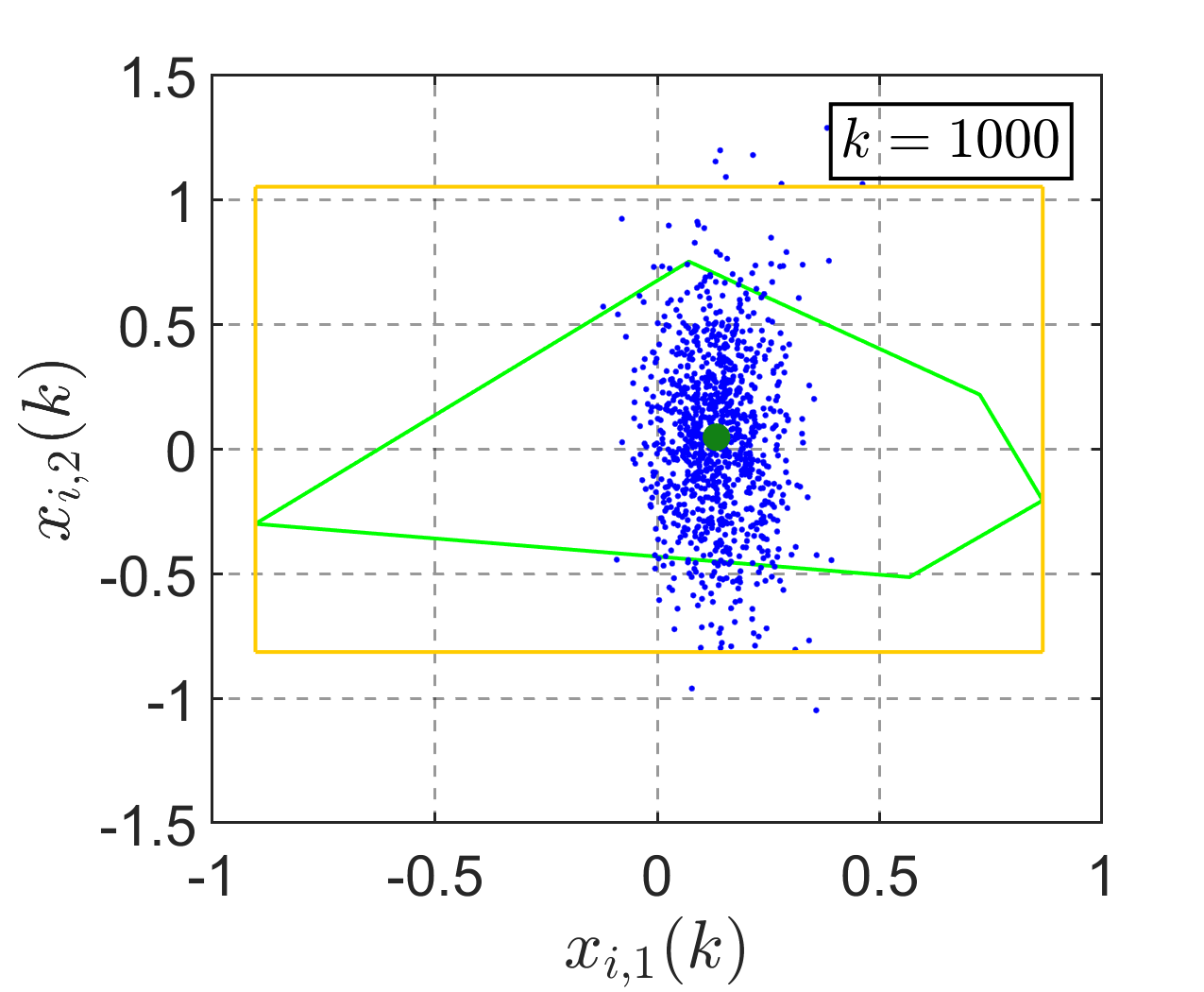}
		\caption{$\lambda=2.5,\upsilon=0.75$}
		\label{fig:pr2}
	\end{subfigure}
	\begin{subfigure}{0.4\linewidth}
		\centering
		\includegraphics[width=1.0\linewidth]{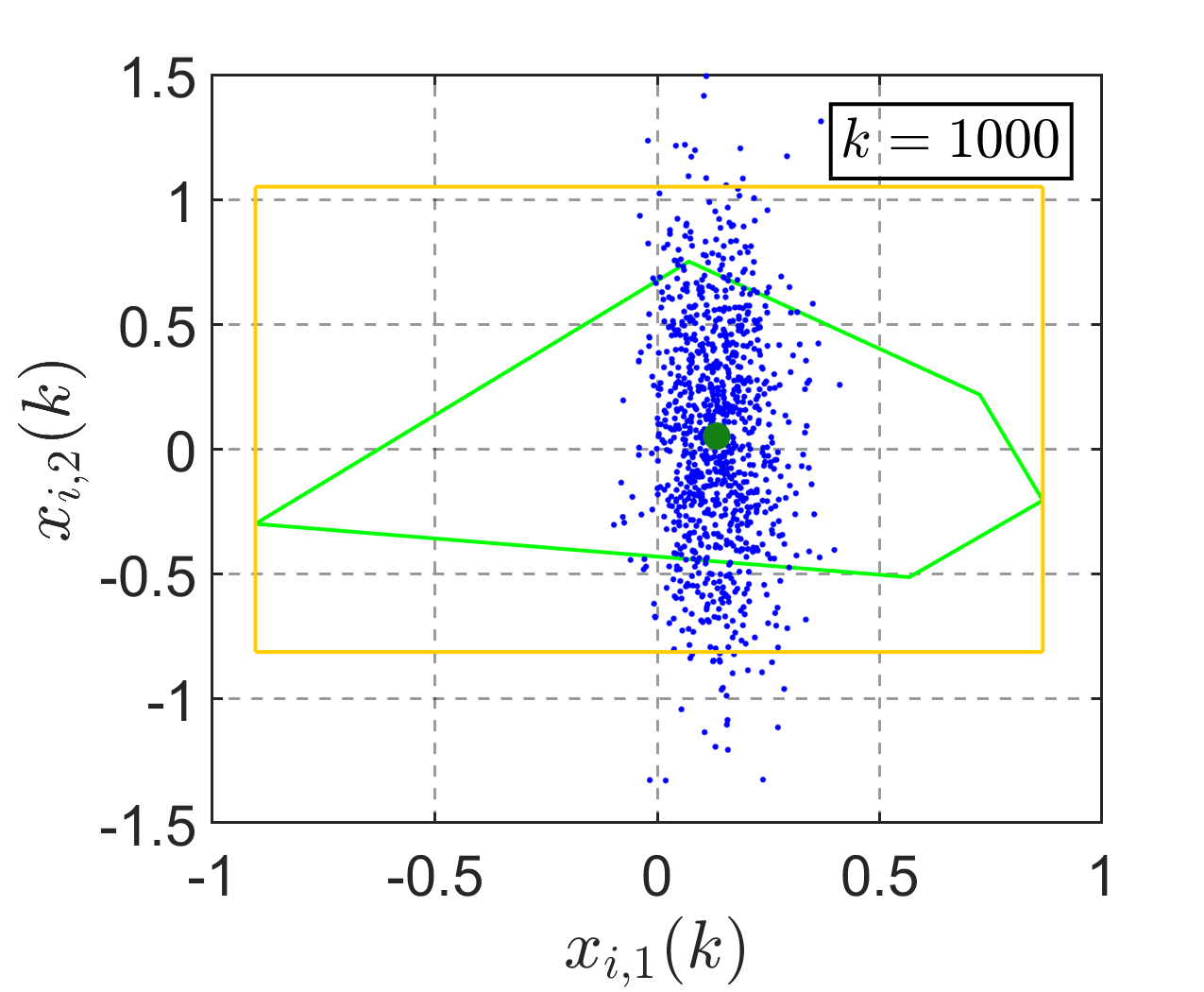}
		\caption{$\lambda=2.5,\upsilon=0.85$}
		\label{fig:pr3}
	\end{subfigure}
	\begin{subfigure}{0.4\linewidth}
		\centering
		\includegraphics[width=1.0\linewidth]{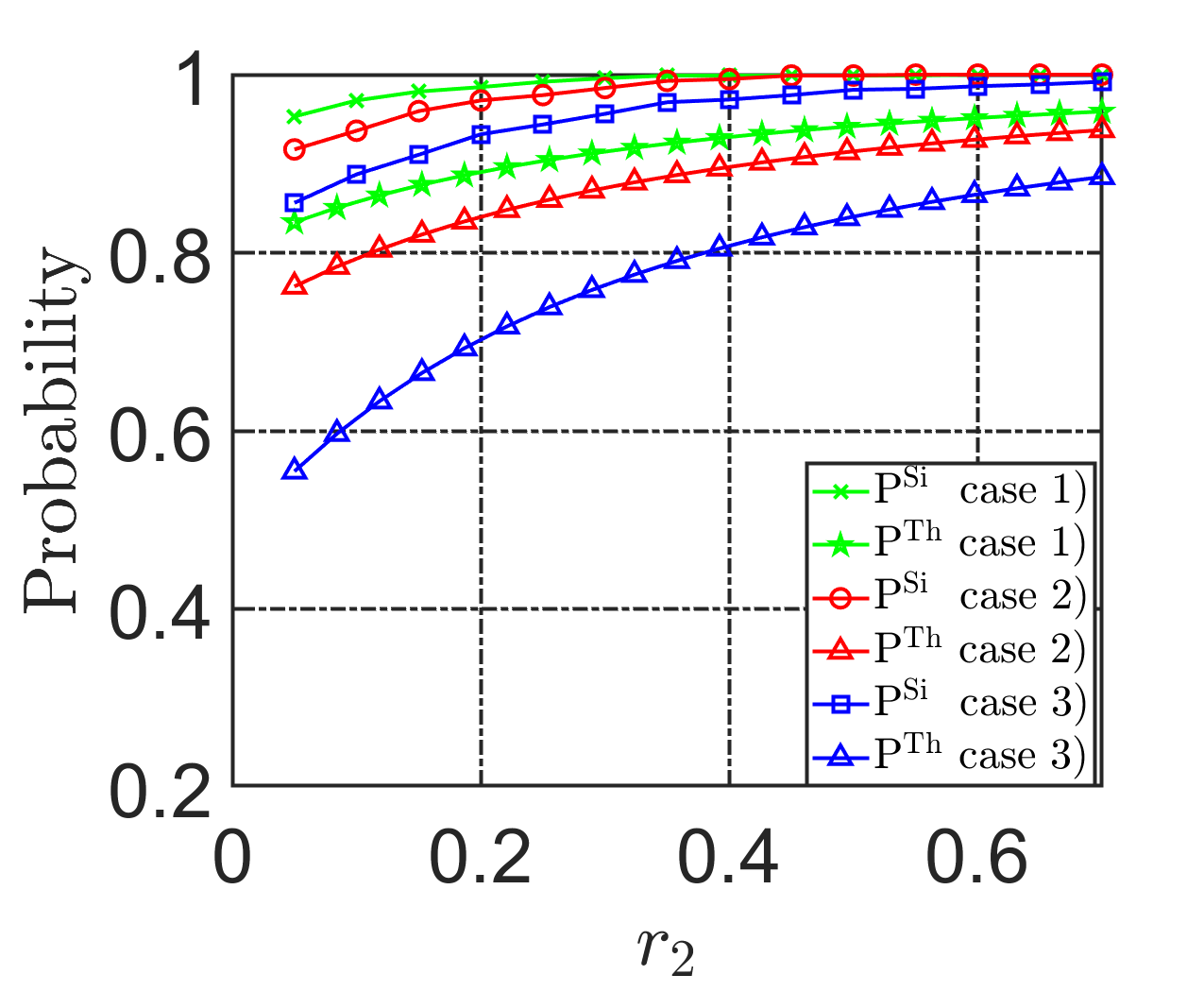}
		\caption{$\mathrm{P}\{ \xi \in \mathcal{C}\}$}
		\label{fig:pr_compare}
	\end{subfigure}
\vspace{-5pt}
\caption{$2$-dimensional Hausdorff Distance Results}
\label{fig:pr}
\end{figure}
\par At last, we illustrate the final value distribution using Mahalanobis distance. We also repeat the process $1000$ times, with $1000$ iterations for three different cases of noise parameters (case 1): $\lambda=2.0,\upsilon=0.75$, case 2): $\lambda=2.5,\upsilon=0.75$, and case 3): $\lambda=2.0,\upsilon=0.65$), and obtain the convergence results shown in Fig. \ref{fig:madis}. Considering that the set of points equidistant in terms of Mahalanobis distance forms an ellipse in $\mathbb{R}^2$, we have drawn $4$ red ellipses under different parameters of $\chi$ for each figure, which are $2,3,4,5$ in order. There are also two sets of lines, i.e., one obtained through simulation probabilities $\mathrm{P}^{\rm Si}$ and the other $\mathrm{P}^{\rm Th}$ derived theoretically from Theorem \ref{th:madis}. It is observed that in all three cases, the simulation results $\mathrm{P}^{\rm Si}$ are greater than the theoretical one $\mathrm{P}^{\rm Th}$.
\begin{figure}[htbp]
	\centering
	\begin{subfigure}{0.4\linewidth}
		\centering
		\includegraphics[width=1.0\linewidth]{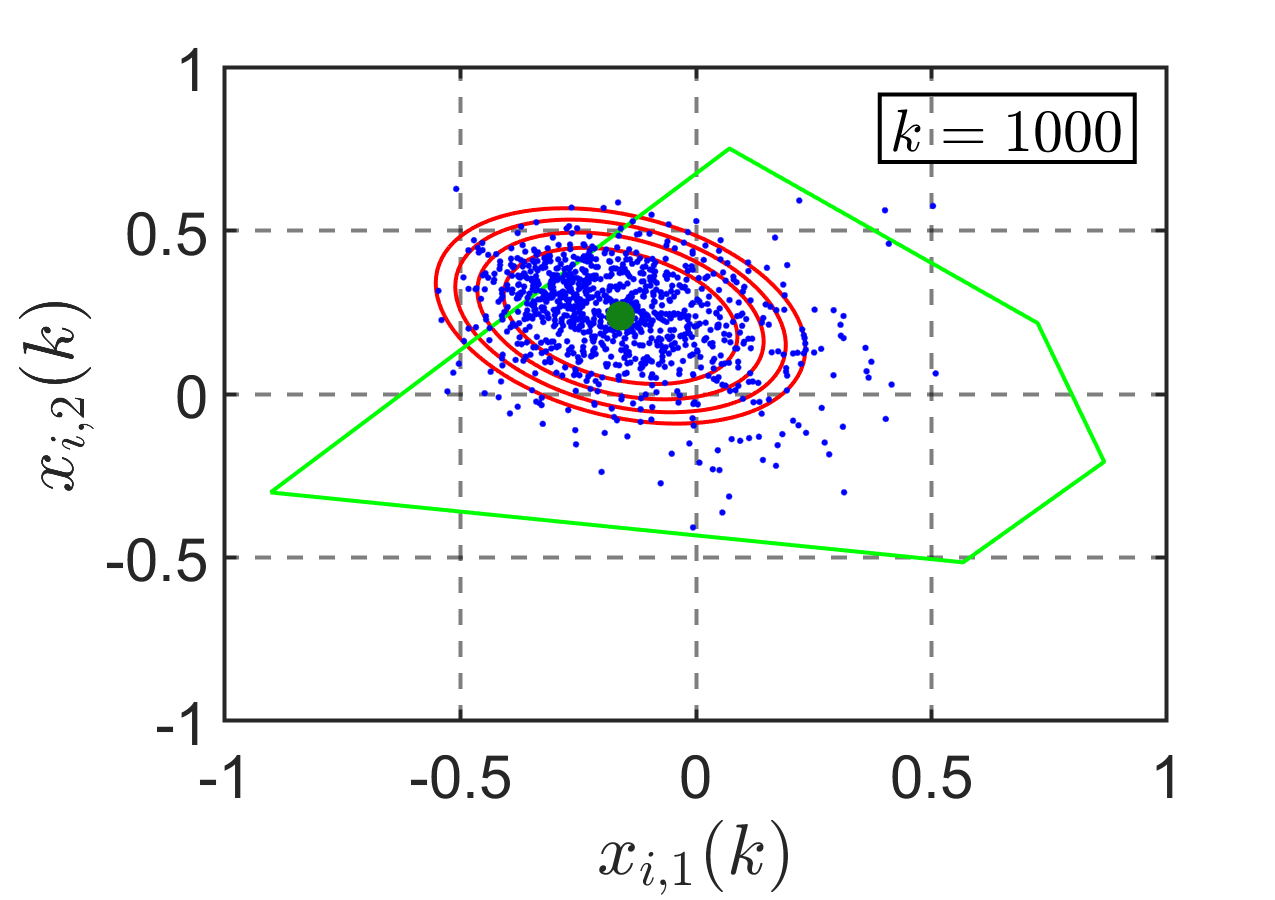}
		\caption{$\lambda=2.0,\upsilon=0.75$}
		\label{fig:pr1_gau}
	\end{subfigure}
	\begin{subfigure}{0.4\linewidth}
		\centering
		\includegraphics[width=1.0\linewidth]{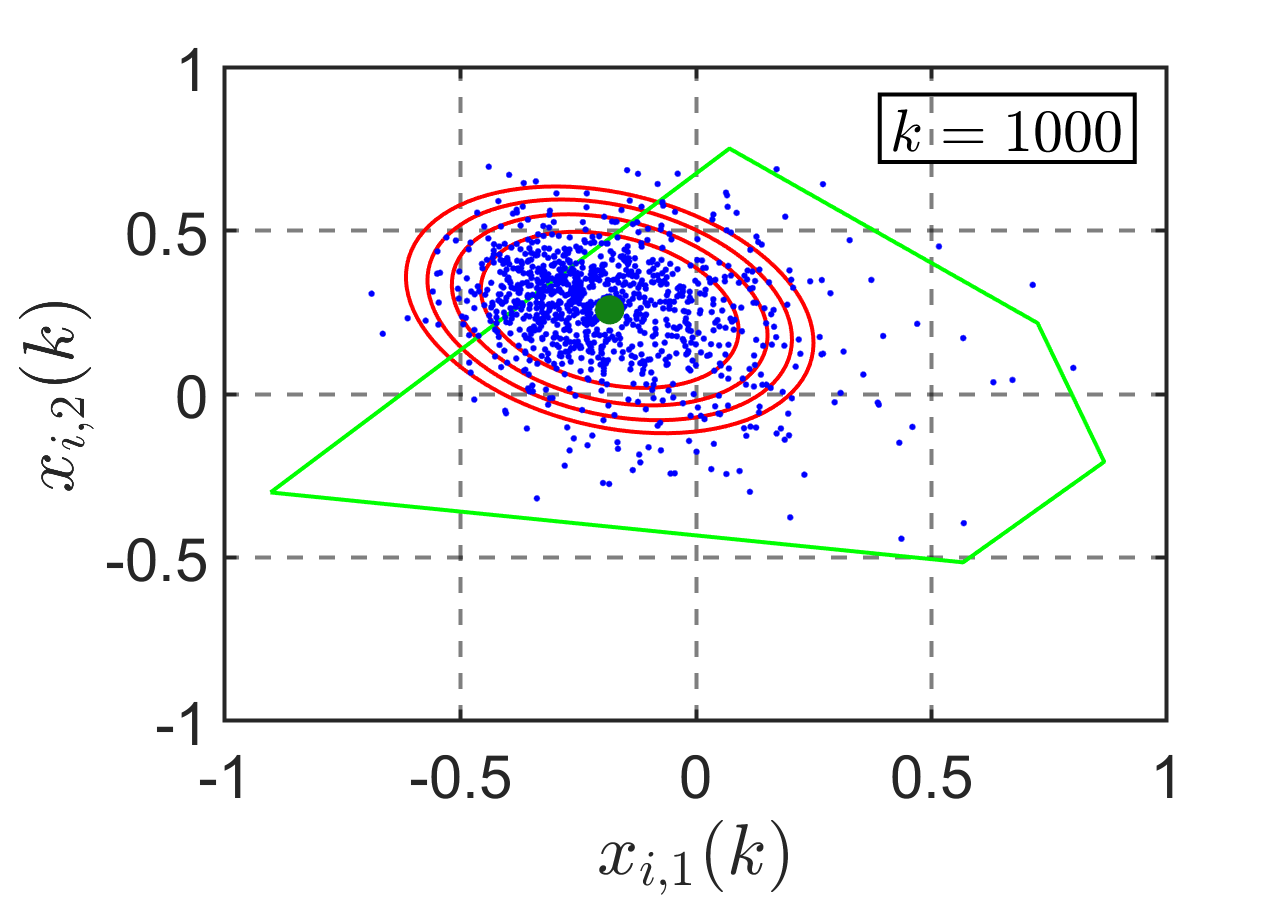}
		\caption{$\lambda=2.5,\upsilon=0.75$}
		\label{fig:pr2_gau}
	\end{subfigure}
	\begin{subfigure}{0.4\linewidth}
		\centering
		\includegraphics[width=1.0\linewidth]{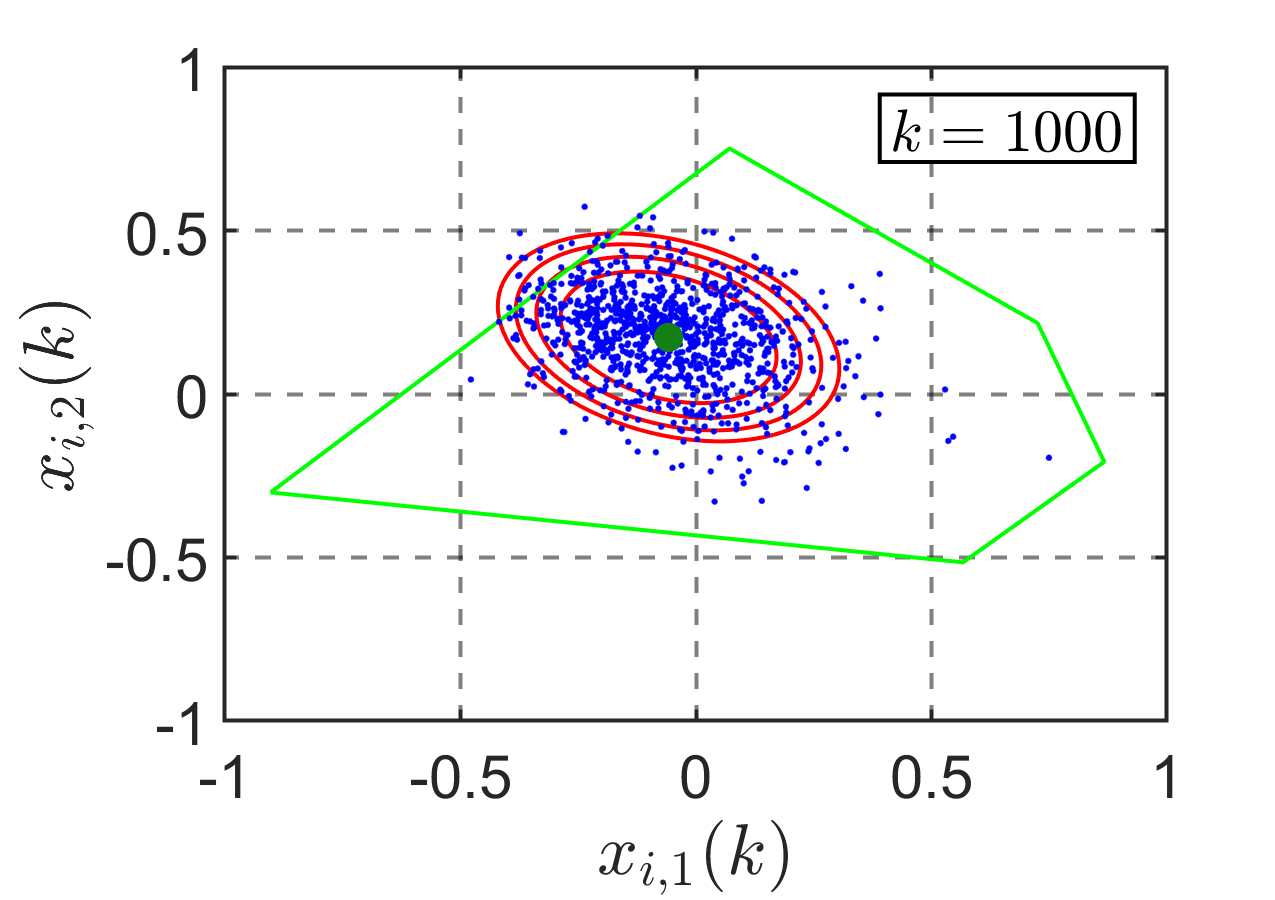}
		\caption{$\lambda=2.0,\upsilon=0.65$}
		\label{fig:pr3_gau}
	\end{subfigure}
	\begin{subfigure}{0.4\linewidth}
		\centering
		\includegraphics[width=1.0\linewidth]{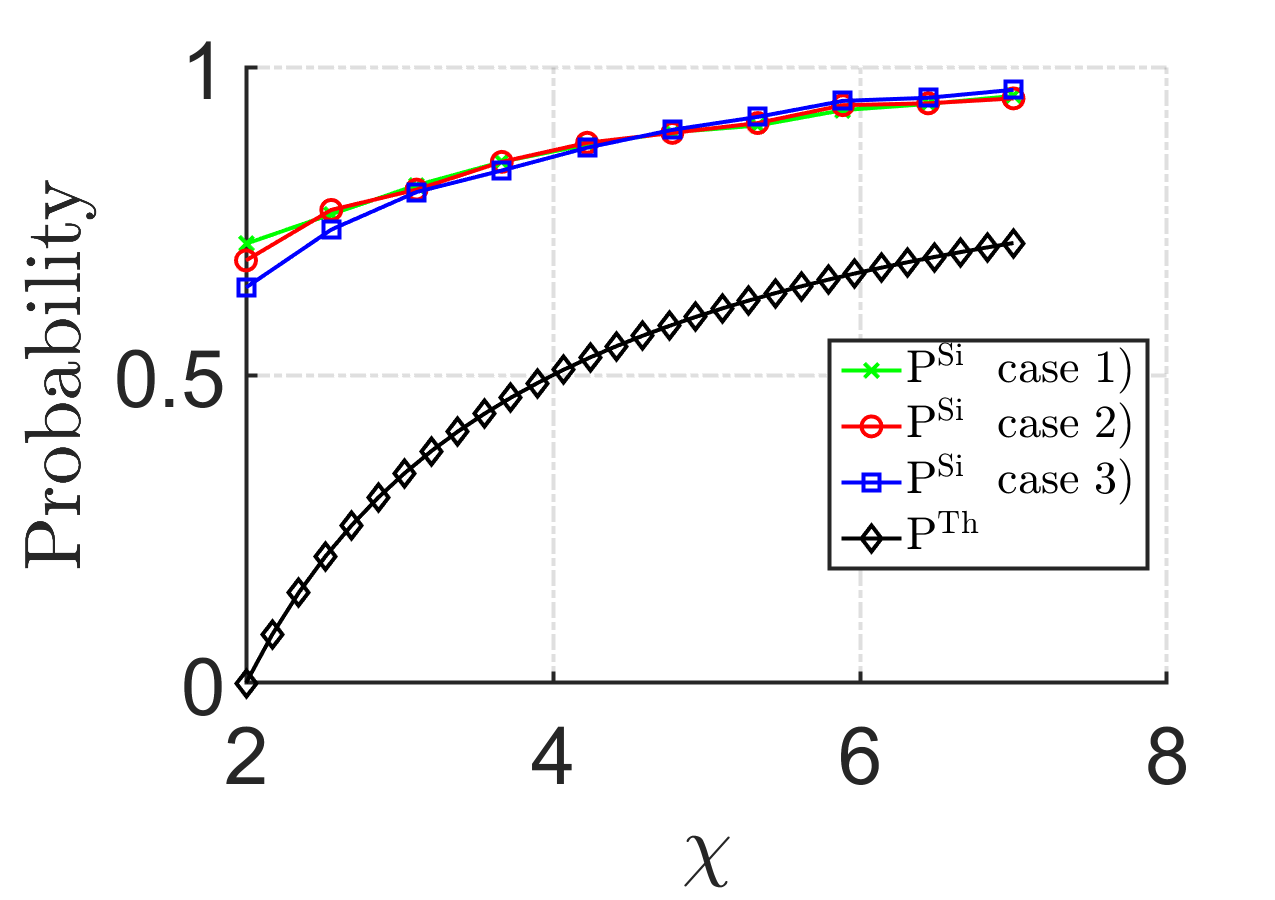}
		\caption{$\mathrm{P}\{[D_M(\xi)]^2\leq \chi\}$}
		\label{fig:pr_compare_gau}
	\end{subfigure}
\vspace{-2pt}
\caption{$2$-dimensional Mahalabinos Distance Results}
\label{fig:madis}
\end{figure}
\subsection{$3$-dimensional case}
\par For the $3$-dimensional case, according to \cite{har2020improved}, although we can find a centerpoint with depth of $\frac{n}{d+1}=\frac{n}{4}$, but the running time is $O(n^{2})$, make the algorithm inefficient. Therefore, Har-Peled \emph{et al.} proposed a compromised algorithm that achieves a depth of $\frac{n}{6}$ for the centerpoint with a runtime of $O(n\operatorname{log}n)$, which is used in our simulations here. We consider multi-robot rendezvous, where there are $12$ robots with a time-varying directed network, consisting of $2$ faulty robots and $10$ normal robots. Each normal robot interacts with one randomly chosen faulty robot and other normal robots, where sufficient conditions for convergence are always guaranteed. The initial states are shown in Fig. \ref{fig:3dinitial}, where the gray polyhedron represents the convex hull formed by the normal robots. We set  $\gamma_i(t) = 0.8$, $\lambda = 2.0$, and $\upsilon = 0.75$, consistent with the $2$-dimensional case. We run the algorithm $1000$ times, each with $1000$ iterations, and obtain the result in Fig. \ref{fig:convergence_3d}. The blue points represent the final values each time, and the green point $(0.051,0.250,-0.312)$ is the sample mean of the final values. We can observe that the expectation falls in the convex hull of the initial states of the normal robots.
\begin{figure}[htbp]
	\centering
	\begin{subfigure}{0.4\linewidth}
		\centering
		\includegraphics[width=1.0\linewidth]{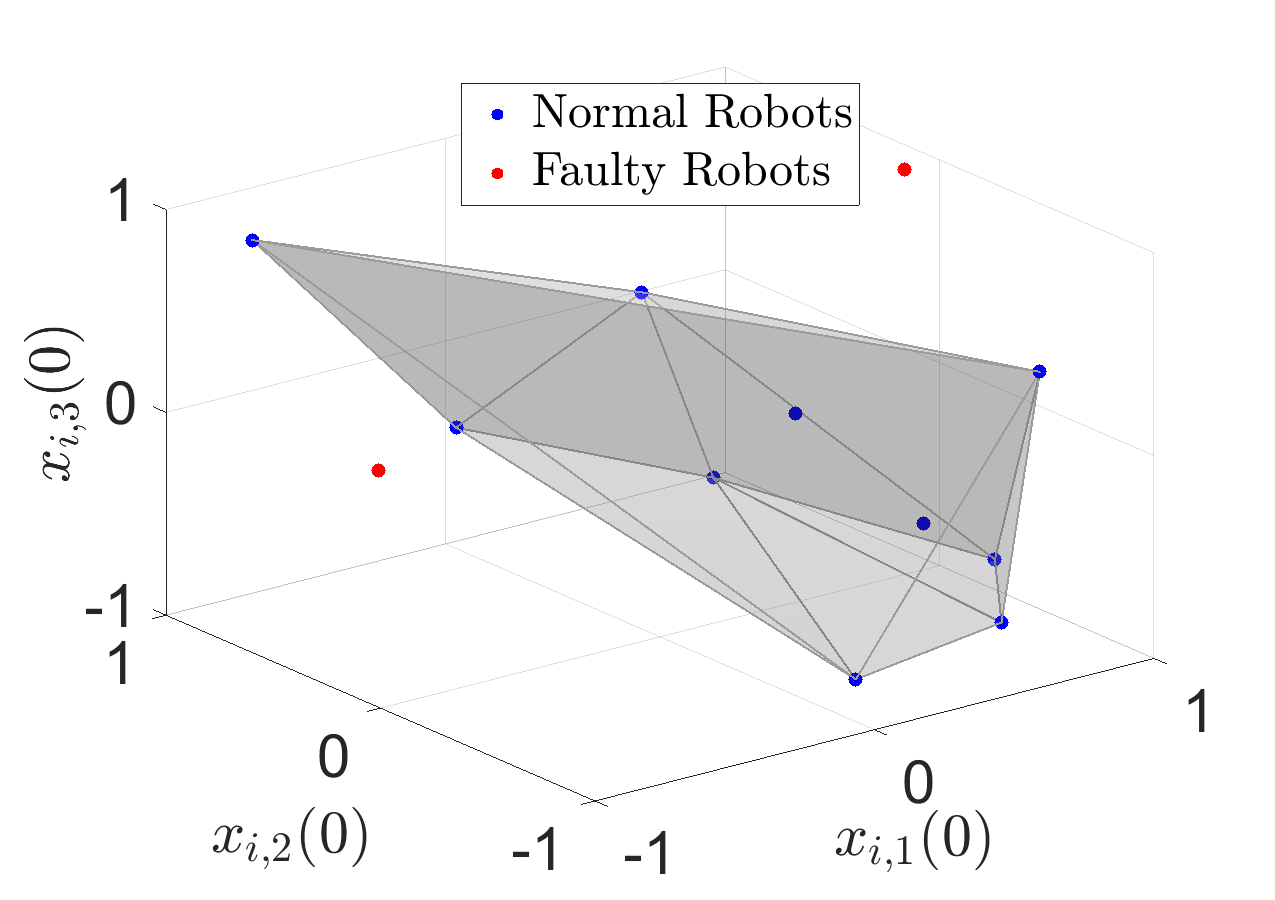}
		\caption{Initial states}
            \label{fig:3dinitial}
	\end{subfigure}
	\begin{subfigure}{0.4\linewidth}
		\centering
		\includegraphics[width=1.0\linewidth]{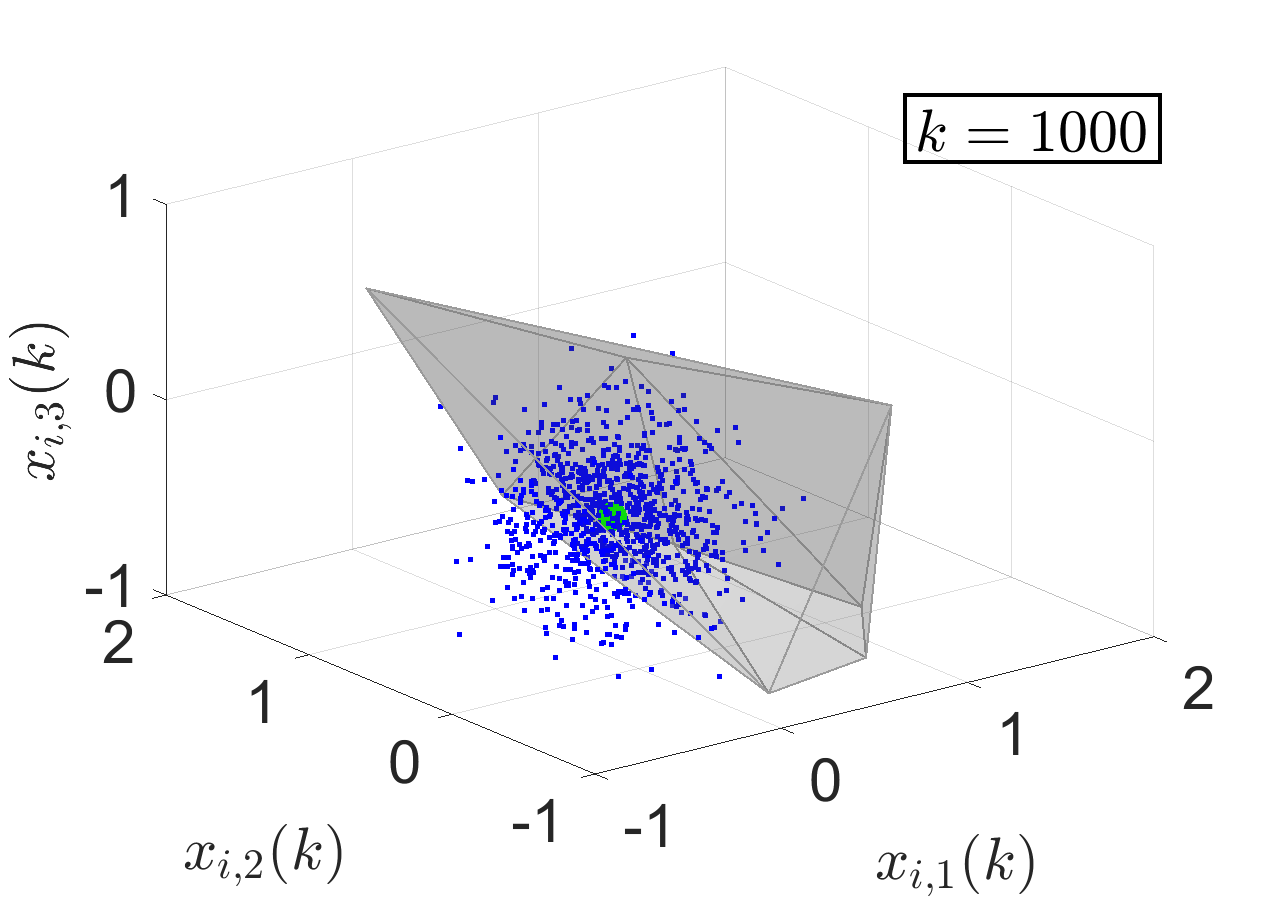}
		\caption{Final value}
            \label{fig:convergence_3d}
	\end{subfigure}
\caption{$3$-dimensional case}
\label{fig:3dmadis}
\vspace{-5pt}
\end{figure}
\par Then, similar to the $2$-dimensional case, we only add noise to the first dimension to illustrate the convex hull change using Hausdorff distance. There are three different cases of noise parameters (case 1): $\lambda=3.0,\upsilon=0.80$, case 2): $\lambda=3.0,\upsilon=0.85$, and case 3): $\lambda=3.5,\upsilon=0.85$). The results shown in Fig. \ref{fig:three_hausdorff} illustrate our results in Theorem \ref{th:distance}. The yellow convex hull represents $\mathcal{C}$, with $r_1=0.3$. The variances of the first dimension in order are $0.307$, $0.520$, and $0.737$, respectively.  When the distance satisifies $r_1=0.3$, we have $\delta(\mathcal{A},\mathcal{C})=1.255$ and $\mu(\mathcal{A})=2.75$, which illustrates our theoretical result  $\delta(\mathcal{A},\mathcal{C})\leq \sqrt{\frac{d}{2}}\mu(\mathcal{A})+r_1$.
\begin{figure}[htbp]
	\centering
	\begin{subfigure}{0.4\linewidth}
		\centering
		\includegraphics[width=1.0\linewidth]{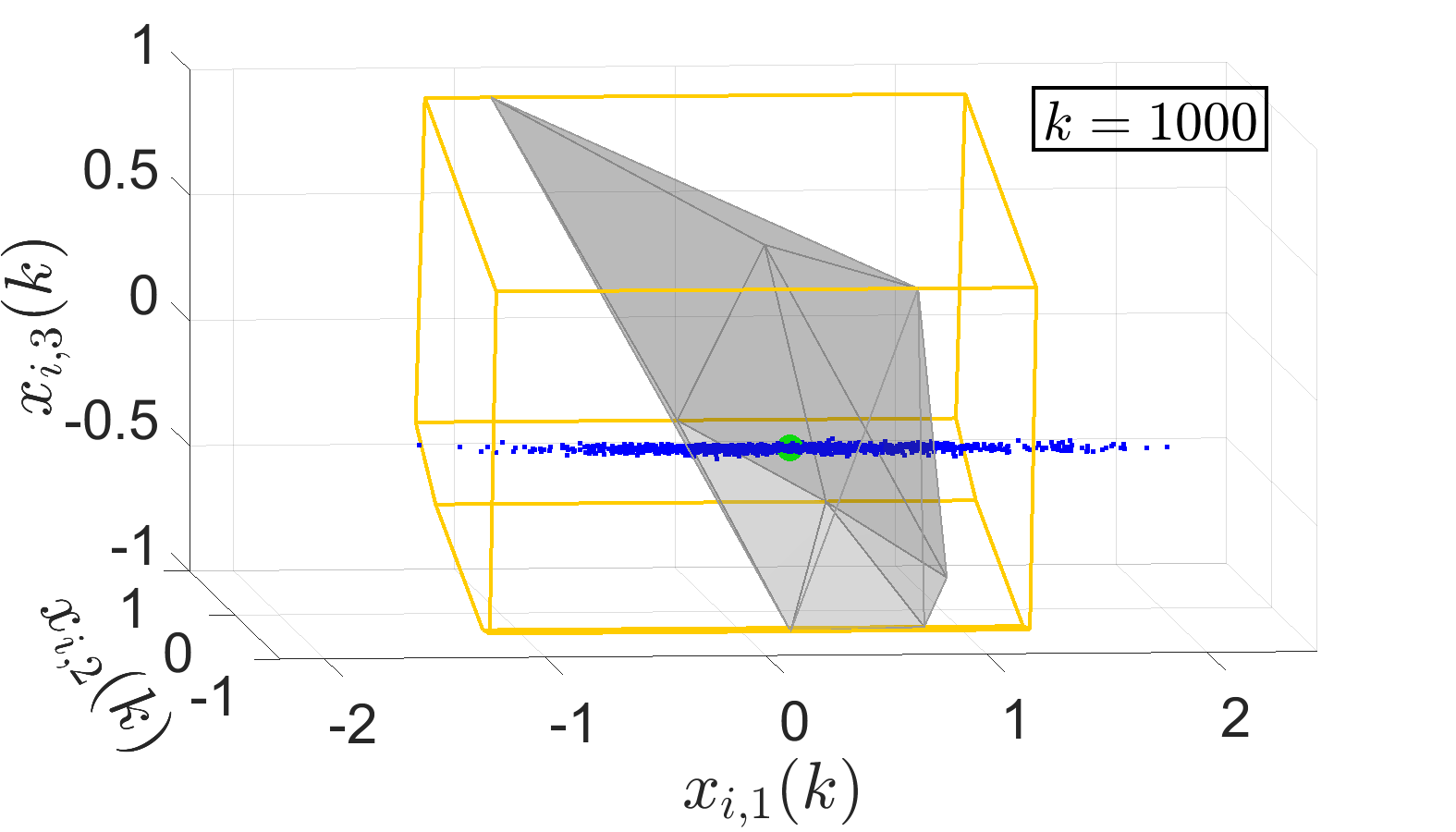}
		\caption{$\lambda=3.0,\upsilon=0.80$}
		\label{fig:hausdorff_3d1}
	\end{subfigure}
	\begin{subfigure}{0.4\linewidth}
		\centering
		\includegraphics[width=1.0\linewidth]{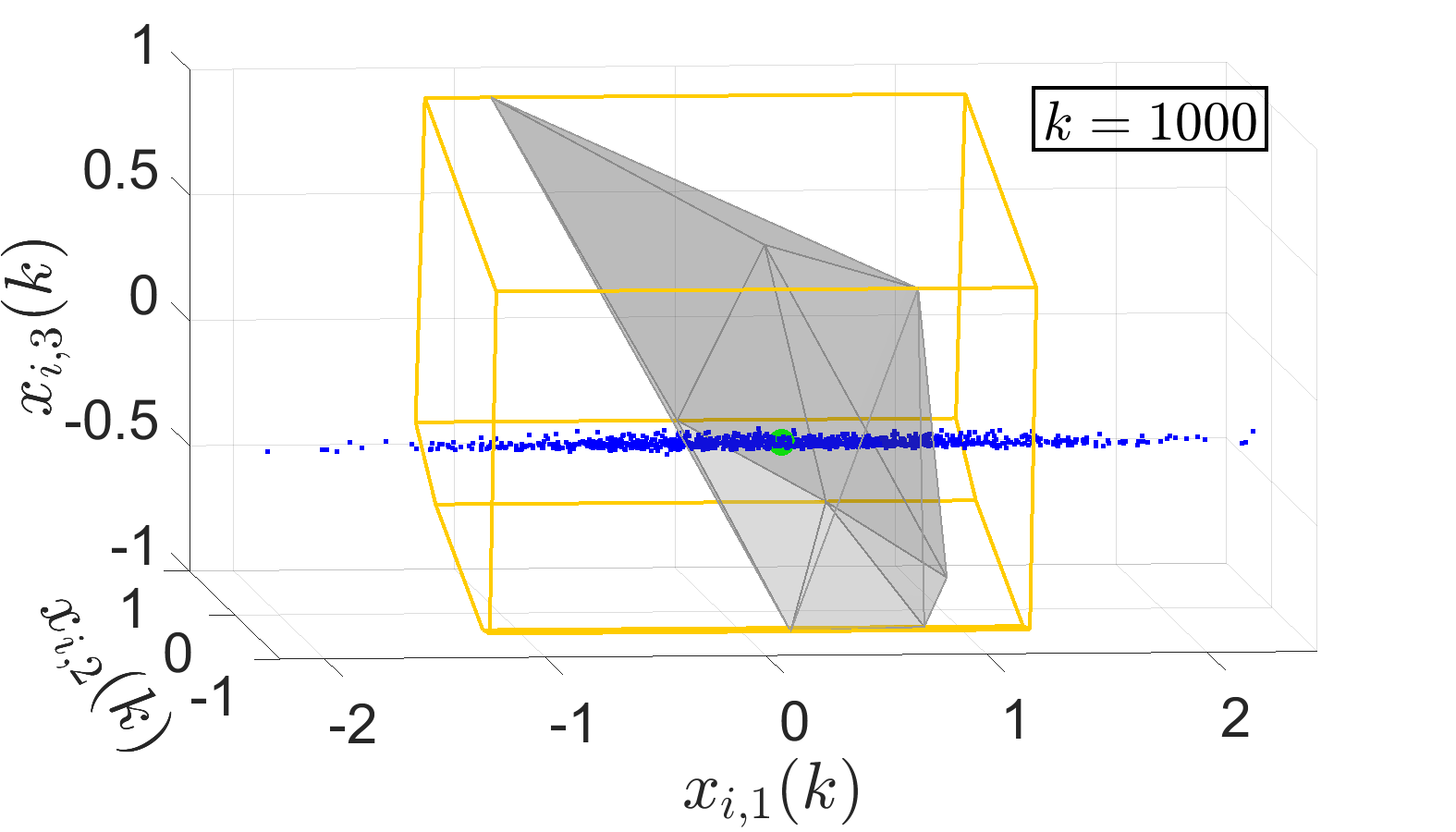}
		\caption{$\lambda=3.0,\upsilon=0.85$}
		\label{fig:hausdorff_3d2}
	\end{subfigure}
	\begin{subfigure}{0.4\linewidth}
		\centering
		\includegraphics[width=1.0\linewidth]{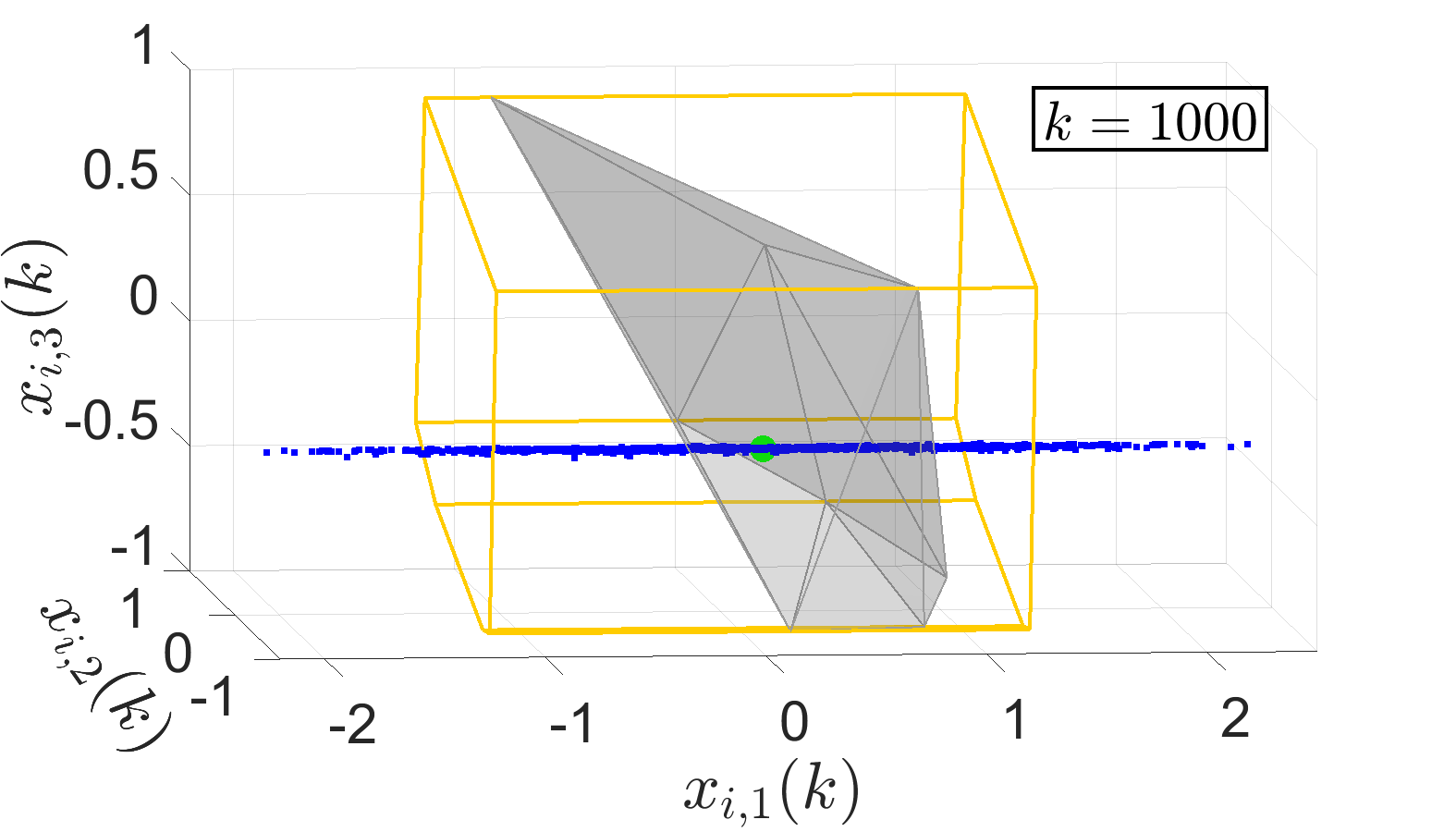}
		\caption{$\lambda=3.5,\upsilon=0.85$}
		\label{fig:hausdorff_3d3}
	\end{subfigure}
	\begin{subfigure}{0.4\linewidth}
		\centering
		\includegraphics[width=1.0\linewidth]{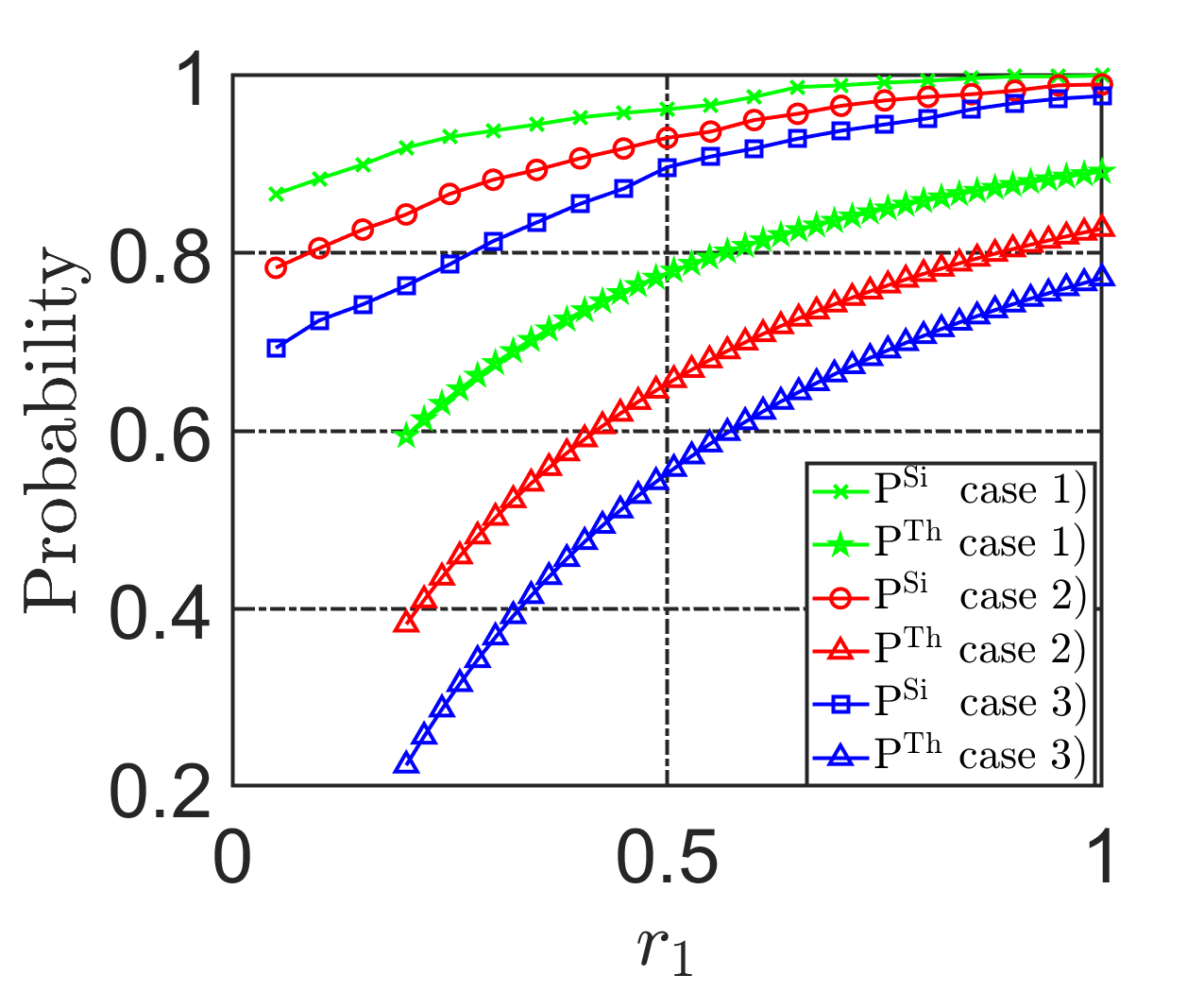}
		\caption{$\mathrm{P}\{\xi\}\in \mathcal{C}\}$}
		\label{fig:hausdorff_3d_pr}
	\end{subfigure}
\vspace{-5pt}
\caption{$3$-dimensional Hausdorff Distance Results}
\label{fig:three_hausdorff}
\end{figure}
\par Finally, we conduct the simulation to illustrate the final value distribution using
Mahalanobis distance in the $3$-dimensional case shown in Fig. \ref{fig:three_mahalabinos}. The noise parameters are identical to those in the $2$-dimensional case. For the gray hyperspheroids shown in the figure, we choose $\chi=3$. From Fig. \ref{fig:mahalanobis_3d_pr}, we observe that $\mathrm{P}^{\mathrm{Si}}$ are greater than the theoretical one $\mathrm{P}^{\mathrm{Th}}$, demonstrating out theoretical results in Theorem \ref{th:madis}.

\begin{figure}[htbp]
	\centering
	\begin{subfigure}{0.4\linewidth}
		\centering
		\includegraphics[width=1.0\linewidth]{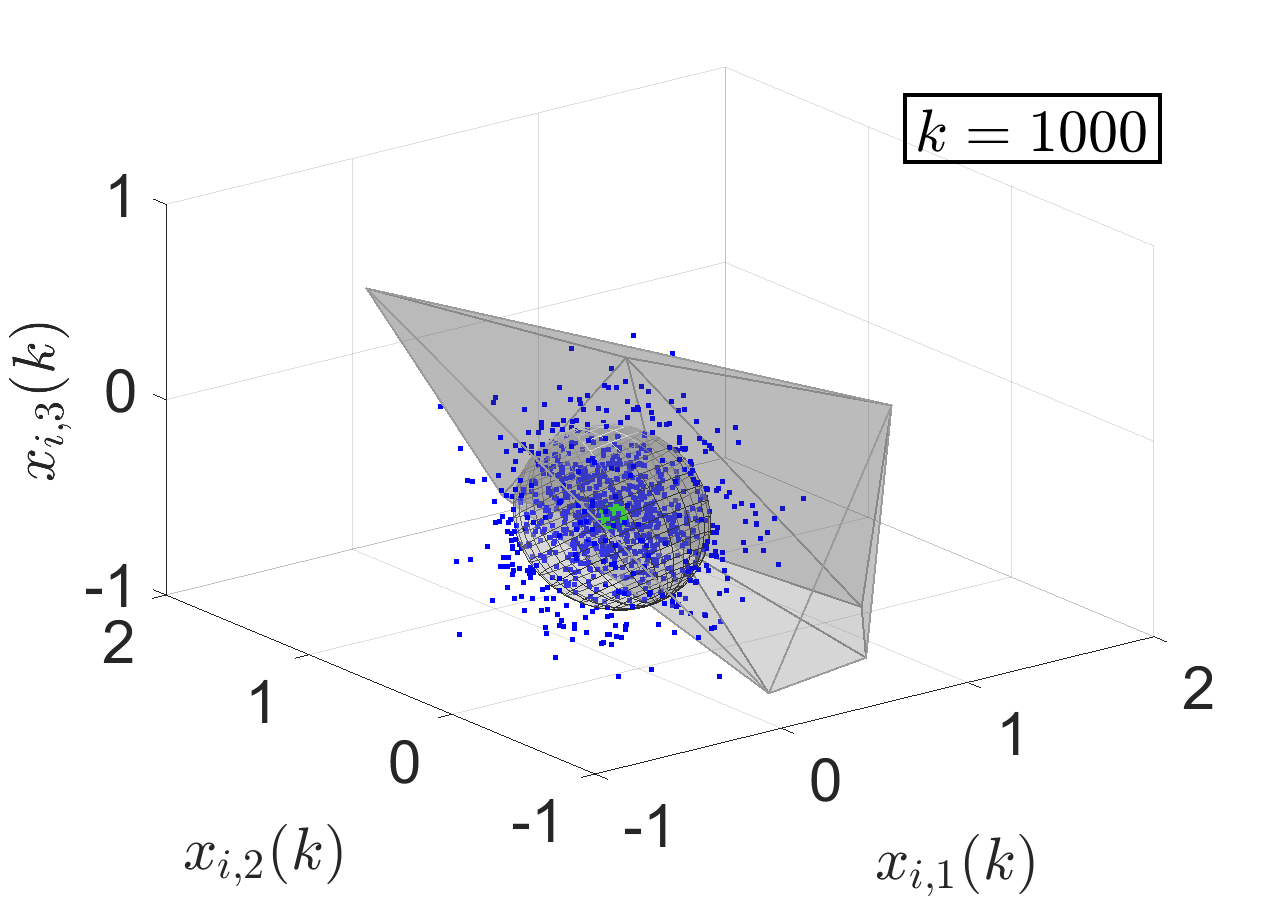}
		\caption{$\lambda=2.0,\upsilon=0.75$}
		\label{fig:mahalanobis_3d1}
	\end{subfigure}
	\begin{subfigure}{0.4\linewidth}
		\centering
		\includegraphics[width=1.0\linewidth]{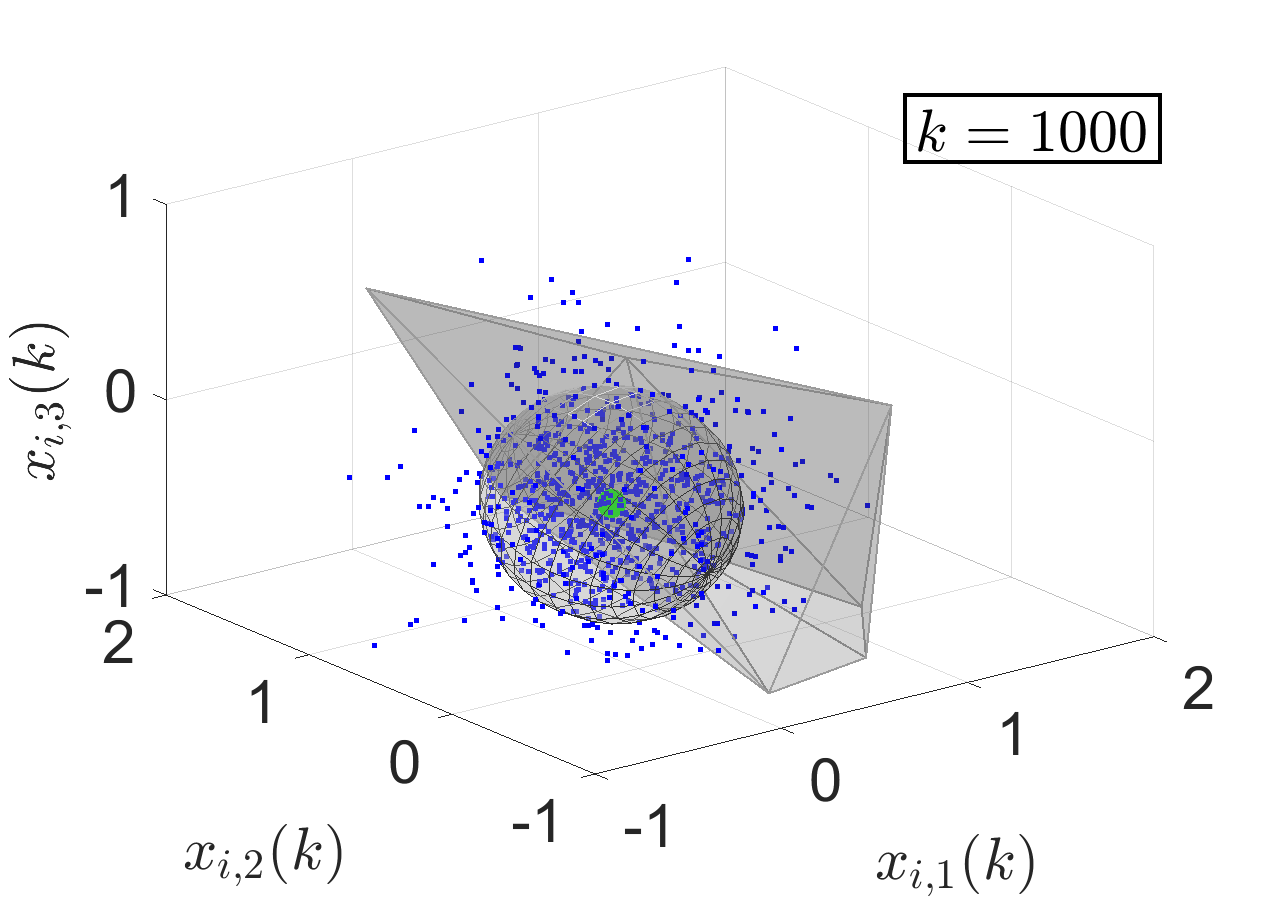}
		\caption{$\lambda=2.5,\upsilon=0.75$}
		\label{fig:mahalanobis_3d2}
	\end{subfigure}
	\begin{subfigure}{0.4\linewidth}
		\centering
		\includegraphics[width=1.0\linewidth]{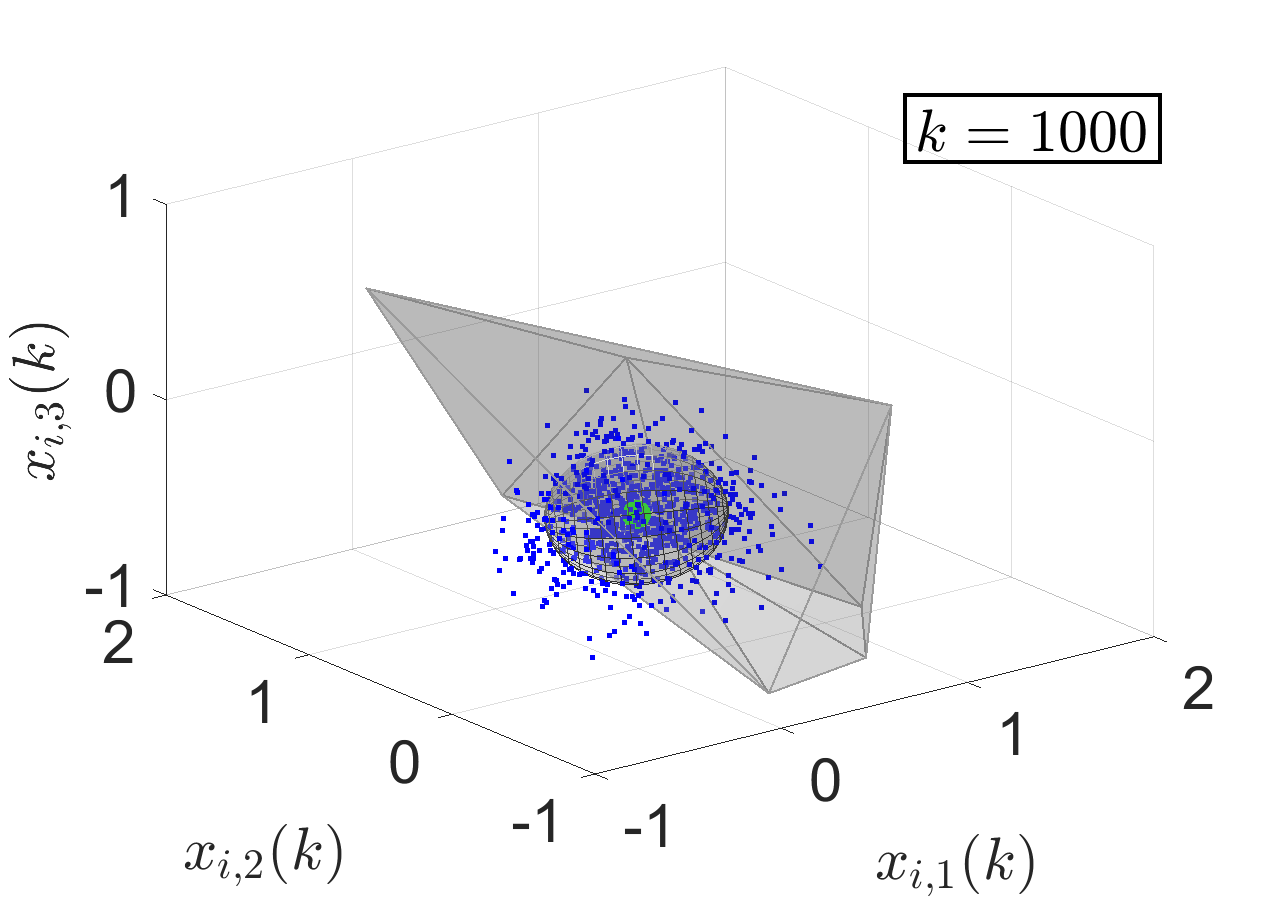}
		\caption{$\lambda=2.0,\upsilon=0.65$}
		\label{fig:mahalanobis_3d3}
	\end{subfigure}
	\begin{subfigure}{0.4\linewidth}
		\centering
		\includegraphics[width=1.0\linewidth]{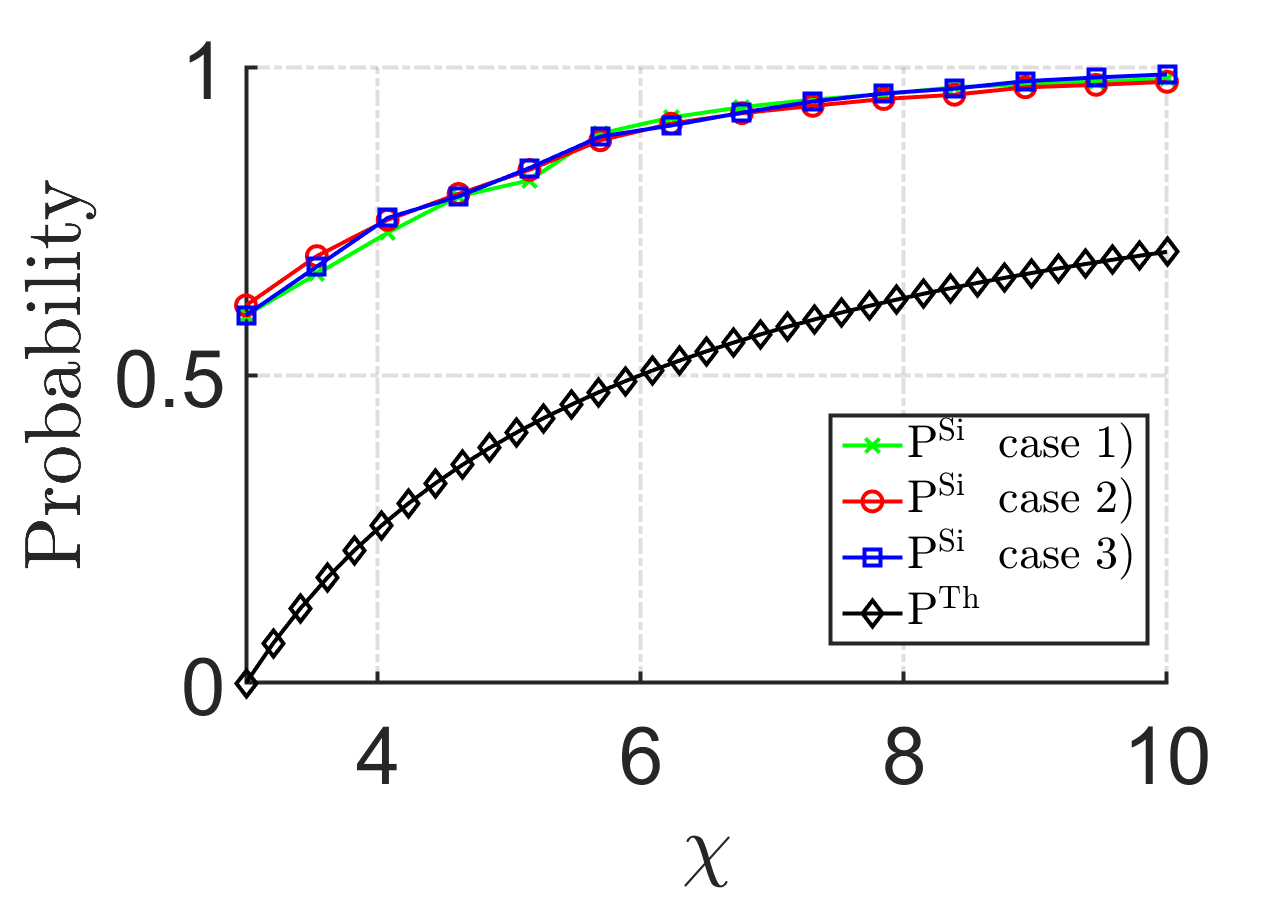}
		\caption{$\mathrm{P}\{[D_M(\xi)]^2\leq \chi\}$}
		\label{fig:mahalanobis_3d_pr}
	\end{subfigure}
\vspace{-1pt}
\caption{$3$-dimensional Mahalabinos Distance Results}
\label{fig:three_mahalabinos}
\end{figure}

\subsection{Privacy preserving perfomance}
\par We consider a system with five normal robots and one faulty robot in a complete graph. The faulty robot sends arbitrary, different messages to each normal robot, randomly chosen from $[-0.7,-0.4] \times [0.4, 0.7]$. Each normal robot exchanges state information with all others. The parameters are set as $\lambda = 2.0$, $\upsilon = 0.75$, and $\gamma_l = 0.4$. The algorithm is executed $1000$ times, each with $1000$ iterations, for two initial states shown in Fig. \ref{fig:initial_states_compare}. The initial state distance is $\mathrm{dist}(\overline{x}_0, \overline{x}_0') = 0.2236$. Normal robots in group $1$ start at $\overline{x}_0$, while those in group $2$ start at $\overline{x}_0'$. The results are presented in Fig. \ref{fig:final_states_compare}. Blue circles indicate the final values of group $1$, with the green circle as their sample mean. Blue crosses denote the final values of group $2$, and the green triangle represents their sample mean.

\begin{figure}[H]
  \begin{subfigure}{0.48\linewidth}
    \centerline{\includegraphics[width=0.9\linewidth]{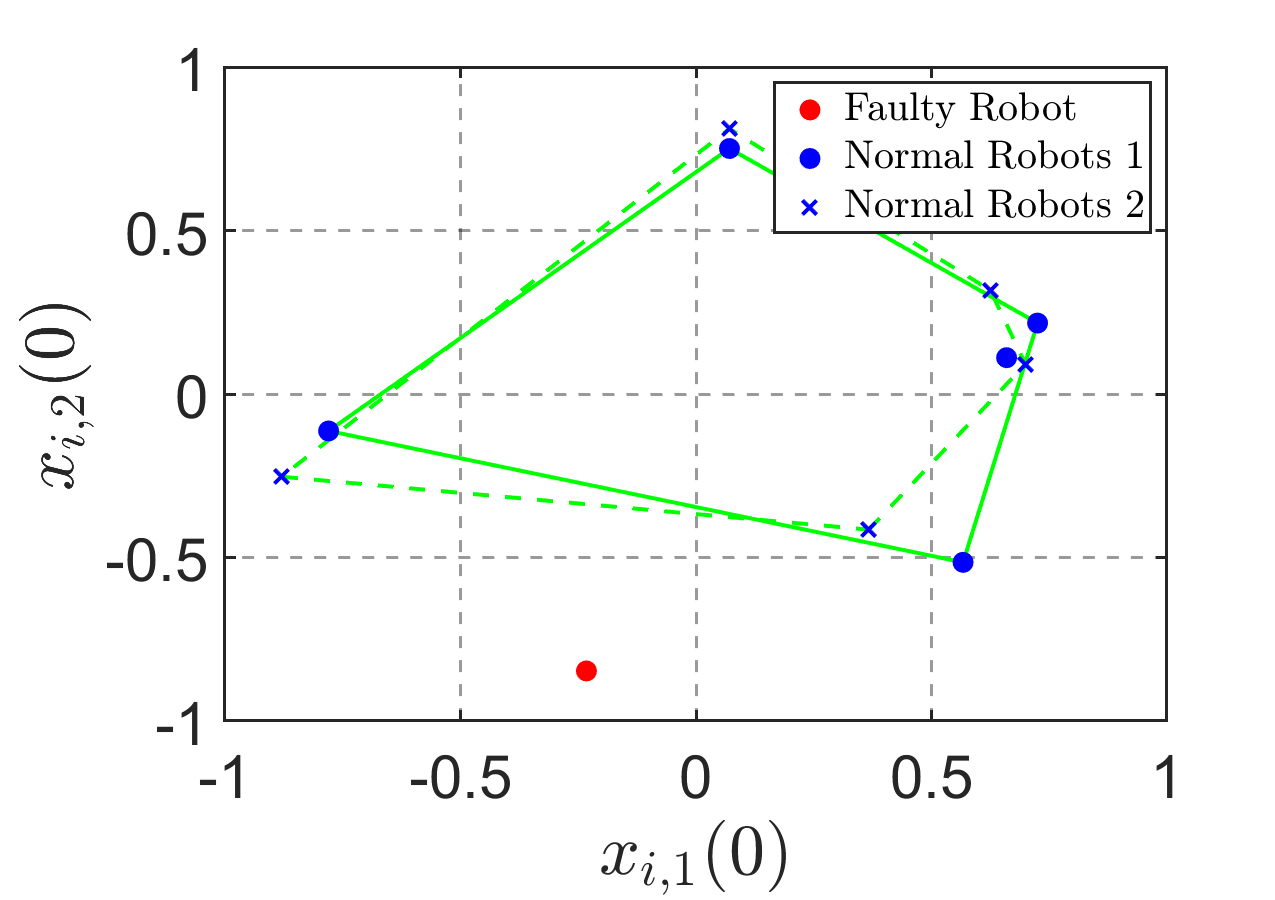}}
    \vspace{-5pt}
    \caption{Comparison of $\overline{x}_0$ and $\overline{x}_0'$}  
    \label{fig:initial_states_compare}  
  \end{subfigure}
  \hfill
  \begin{subfigure}{0.48\linewidth}
    \centerline{\includegraphics[width=0.9\linewidth]{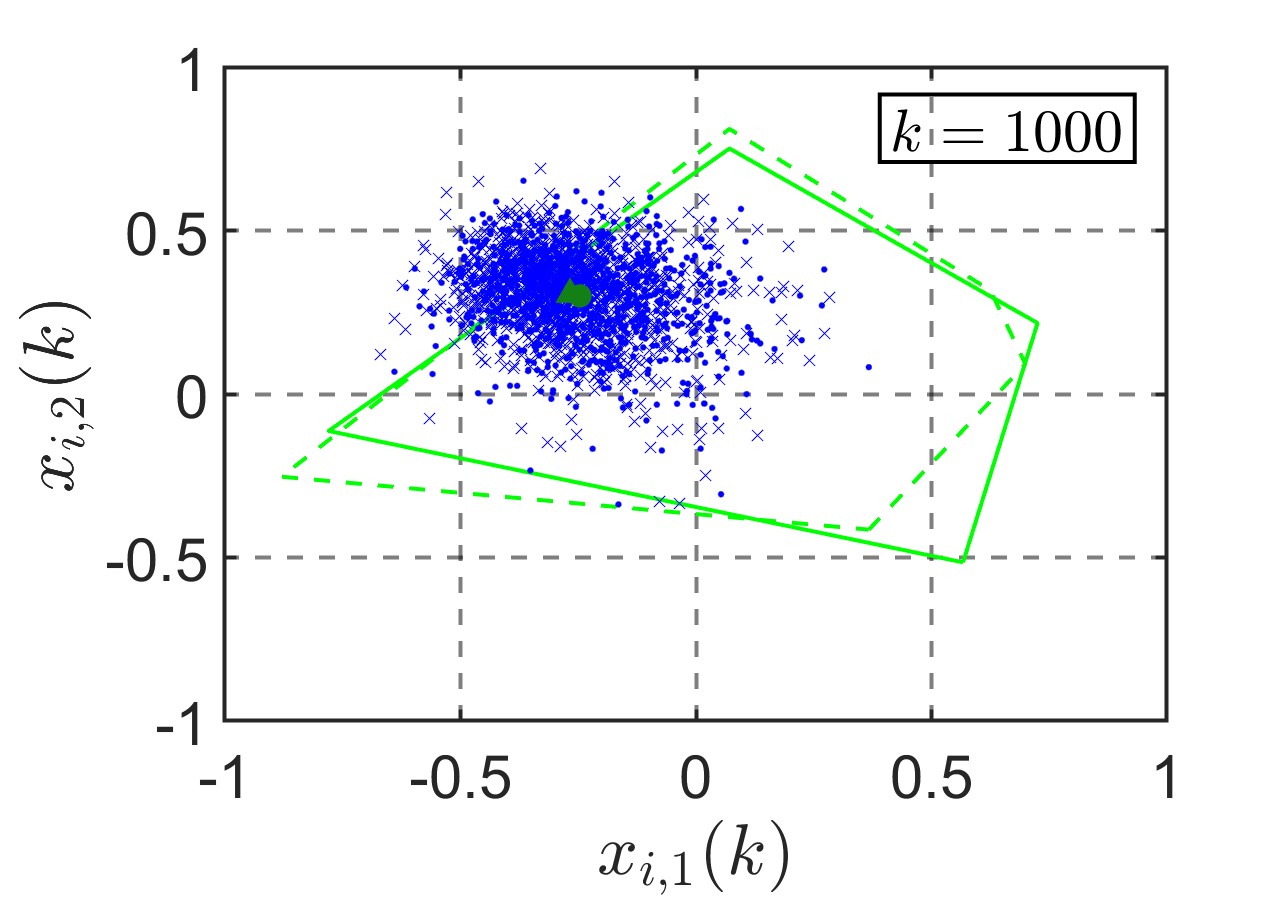}}
    \vspace{-5pt}
    \caption{Comparison of final values}
    \label{fig:final_states_compare}  
  \end{subfigure}
  \vspace{-2pt}
  \caption{Simulations for initial states $\overline{x}_0$ and $\overline{x}_0'$}
\end{figure}

Substituting the parameters into the result of Theorem \ref{th:cgp_rvs}, we obtain the following expression for the R{\'e}nyi divergence.
\begin{equation}
D_\alpha(\mathbf{y}_{\overline{\mathbf{x}}_0}(\mathbf{N})\|\mathbf{y}_{\overline{\mathbf{x}}_0'}(\mathbf{N}))\leq0.9374 \alpha.
\end{equation}
We collect all the data transmitted by the normal robots over $1000$ iterations and repeat this process $1000$ times to fit the distribution. Finally, by setting $\alpha = 2.0$, we obtain a Rényi divergence of $D_\alpha(\mathbf{y}_{\overline{\mathbf{x}}_0}(\mathbf{N})\|\mathbf{y}_{\overline{\mathbf{x}}_0'}(\mathbf{N}))=1.2863$ between the two distributions, which is smaller than $0.9374\times 2=1.8748$.

\section{Conclusion}\label{sec:conclusion}
\par In this paper, we studied private resilient vector consensus in multi-agent systems with faulty agents by adding Gaussian noise. We established resilient vector consensus in expectation and analyzed convergence accuracy using the Mahalanobis distance between the final value and its expectation and the Hausdorff distance between the two convex hulls with and without noise. We then proved $\rho$-CGP for initial states and compared it with $(\varepsilon,\delta)$-differential privacy. Numerical simulations validated our results. Future work will focus on identifying necessary and sufficient conditions for resilient consensus, finding the explicit lower bound for the nonzero entr1ies of $\mathbf{M}(t)$, deriving the exact distribution of the final value, and optimizing noise selection for CGP under different scenarios.

\section*{References}

\bibliographystyle{ieeetr}
\bibliography{ref}

\begin{IEEEbiography}[{\includegraphics[width=1in,height=1.25in,clip,keepaspectratio]{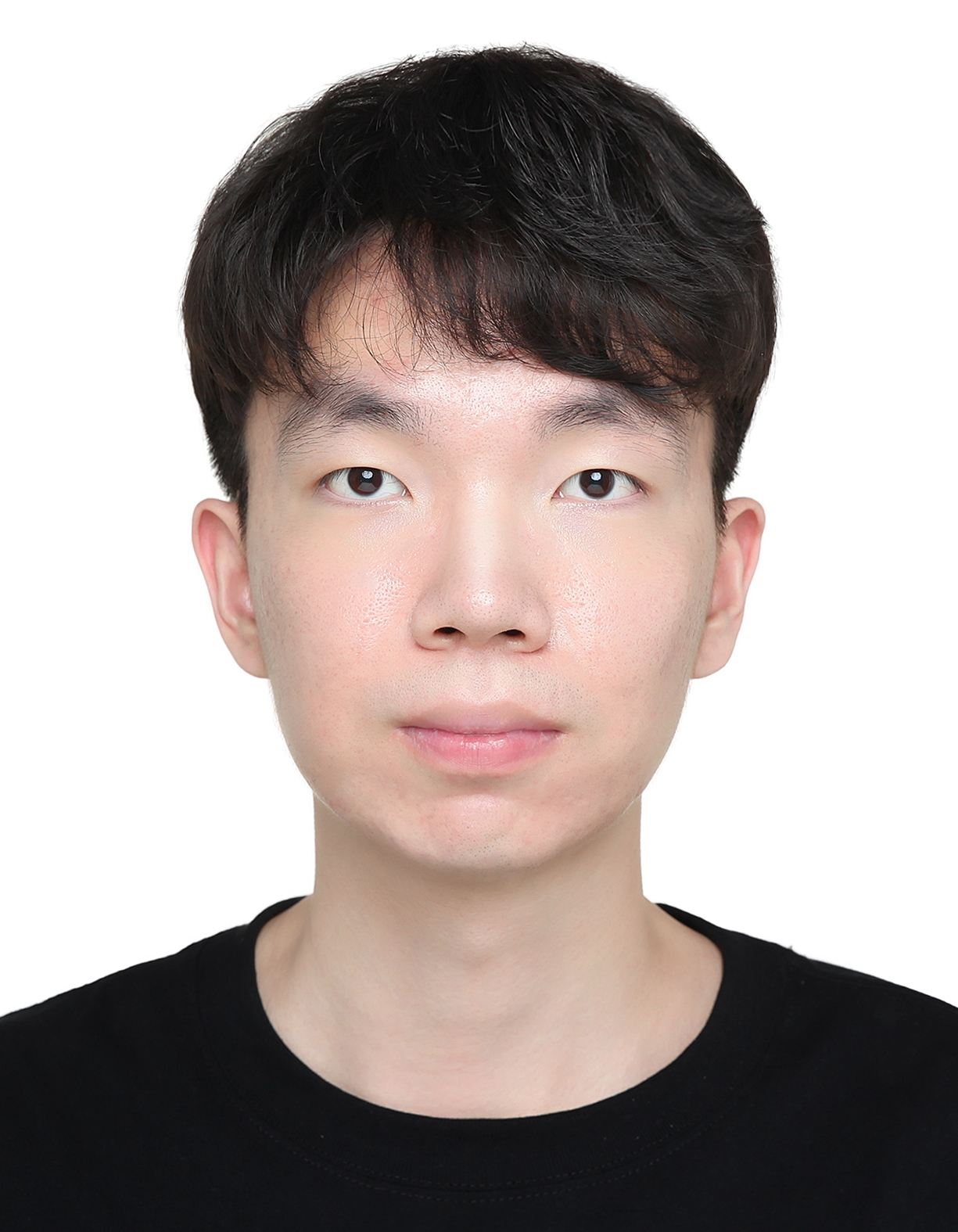}}]{Bing Liu} received the B.Sc. degree in Automation from Tongji University, Shanghai, China, in 2023. He is currently studying for his doctor's degree with the College of Control Science and Engineering, Zhejiang University, Hangzhou, China.
His current research interests include multi-agent systems and differential privacy.
\end{IEEEbiography}
\begin{IEEEbiography}[{\includegraphics[width=1in,height=1.25in,clip,keepaspectratio]{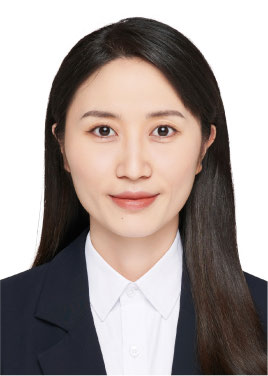}}]{Chengcheng Zhao} (Member, IEEE) received the
B.Sc. degree in measurement and control technology
and instruments from Hunan University, Changsha,
China, in 2013, and the Ph.D. degree in control
science and engineering from Zhejiang University,
Hangzhou, China, in 2018. She was a PostDoctoral Fellow with the College of Control
Science and Engineering, Zhejiang University,
from 2018 to 2021. She is currently a Researcher
with the College of Control Science and Engineering, Zhejiang University. Her research interests
include consensus and distributed optimization, and security and privacy in
networked systems. She received the IEEE PESGM 2017 Best Conference
Papers Award, and one of her papers was shortlisted in the IEEE
ICCA 2017 Best Student Paper Award Finalist. She is an Editor of \emph{Wireless
Networks} and \emph{IET Cyber-Physical Systems: Theory and Applications}.
\end{IEEEbiography}
\begin{IEEEbiography}[{\includegraphics[width=1in,height=1.25in,clip,keepaspectratio]{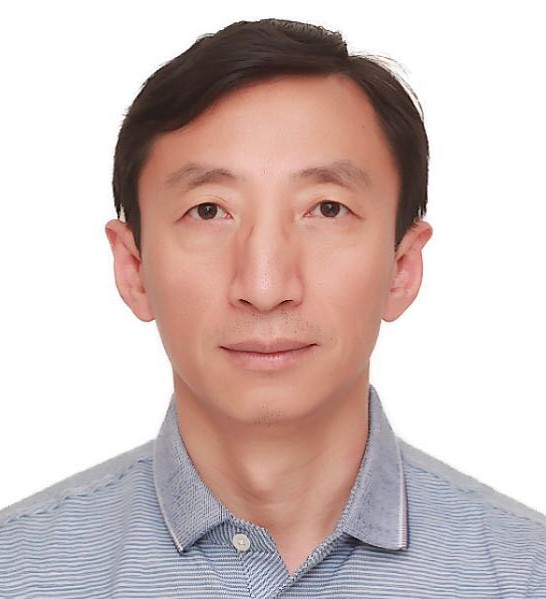}}]{Li Chai} (Member, IEEE) received the B.Sc. degree in applied mathematics and the M.S. degree in control science and engineering from Zhejiang University, China, in 1994 and 1997, respectively, and the Ph.D degree in electrical engineering from the Hong Kong University of Science and Technology, Hong Kong, in 2002. From August 2002 to December 2007, he was at Hangzhou Dianzi University, China. He worked as a professor at Wuhan University of Science and Technology, China, from 2008 to 2022. In August 2022, he joined Zhejiang University, China, where he is currently a professor at the College of Control Science and Engineering. He has been a postdoctoral researcher or visiting scholar at Monash University, Newcastle University, Australia and Harvard University, USA. His research interests include stability analysis, distributed optimization, filter banks, graph signal processing, and networked control systems. Professor Chai is the recipient of the Distinguished Young Scholar of the National Science Foundation of China. He has published over 100 fully refereed papers in prestigious journals and leading conferences. He serves as the Associate Editor of IEEE Transactions on Circuit and Systems II: Express Briefs, Control and Decision and Journal of Image and Graphs.
\end{IEEEbiography}
\begin{IEEEbiography}[{\includegraphics[width=1in,height=1.25in,clip,keepaspectratio]{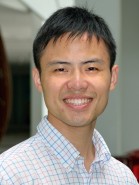}}]{Peng Cheng} (Member, IEEE) received the B.Sc. and Ph.D. degrees in control science and engineering from Zhejiang University, Hangzhou, China, in 2004 and 2009, respectively. He is currently a Professor and Dean of the College of Control Science and Engineering, Zhejiang University. He has been awarded the 2020 Changjiang Scholars Chair Professor. He has received State Science and Technology Progress Award, and MOE Natural Science Award. He serves as Associate Editors for the IEEE Transactions on Control of Network Systems. He also serves/served as Guest Editors for IEEE Transactions on Automatic Control and IEEE Transactions on Signal and Information Processing over Networks. His research interests include networked sensing and control, cyber-physical systems, and control system security.
\end{IEEEbiography}
\begin{IEEEbiography}[{\includegraphics[width=1in,height=1.25in,clip,keepaspectratio]{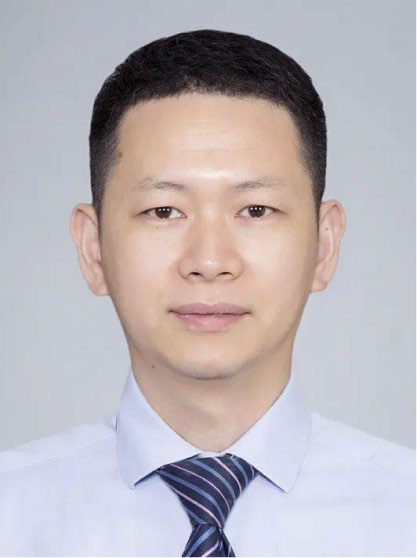}}]{Jiming Chen} (Fellow, IEEE) received the B.Sc. and Ph.D. degrees in control science and engineering from Zhejiang University, Hangzhou, China, in 2000 and 2005, respectively. He is currently a Professor with the College of Control Science and Engineering and the Deputy Director of the State Key Laboratory of Industrial Control Technology, Zhejiang University. He is also the President of Hangzhou Dianzi University, Hangzhou. He serves as an Associate Editor for the IEEE Transactions on Automatic Control. His research interests include the Internet of Things, sensor networks, networked control, and control system security.
\end{IEEEbiography}

\end{document}